\documentclass{LMCS}

\def\doi{8 (3:01) 2012}
\lmcsheading%
{\doi}
{1--54}
{}
{}
{Oct.~20, 2011}
{Jul.~31, 2012}
{}

\usepackage[usenames,dvipsnames]{color}
\usepackage{array}
\usepackage{stmaryrd}
\usepackage{qsymbols}
\usepackage{enumerate,hyperref}
\usepackage{amsthm}
\usepackage{MnSymbol}
\theoremstyle{plain}
\newtheorem{theorem}{Theorem}[section]
\newtheorem{lemma}[theorem]{Lemma}
\newtheorem{corollary}[theorem]{Corollary}
\newtheorem{proposition}[theorem]{Proposition}
\theoremstyle{definition}
\newtheorem{definition}[theorem]{Definition}
\newtheorem{remark}[theorem]{Remark}
\newcommand{\heapvar}{\mathcal{H}}
\newcommand{\permvar}{\mathcal{P}}

\newcommand{\framed}[1]{\mathsf{sframed}(#1)}

\newcommand{\wpCh}[1]{wp_{\mathrm{ch}}(#1)}

\newcommand{\assume}{\texttt{assume}}
\newcommand{\assert}{\texttt{assert}}
\newcommand{\inhale}{\texttt{inhale}}
\newcommand{\exhale}{\texttt{exhale}}
\newcommand{\havoc}{\texttt{havoc}}
\newcommand{\Null}{\texttt{null}}
\newcommand{\assumeSL}{\texttt{assume}\ensuremath{^*}}
\newcommand{\assertSL}{\texttt{assert}\ensuremath{^*}}
\newcommand{\wand}{\mathbin{-\!\!*}}
\newcommand{\full}{\ensuremath{1}}
\newcommand{\none}{\ensuremath{0}}
\newcommand{\Perm}{\ensuremath{P}}
\newcommand{\Heap}{\ensuremath{H}}
\newcommand{\PHeap}{\ensuremath{h}}

\newcommand{\weakenable}{weakening-closed}
\newcommand{\Weakenable}{Weakening-closed}

 \def\code#1{\ensuremath{\mathsf{#1}}}

\newcommand{\agrees}[1]{\stackrel{#1}{\equiv}}

\newcommand{\modelsSL}{\models_{\scriptscriptstyle\textit{SL}}}
\newcommand{\modelsFO}{\models_{\scriptscriptstyle\textit{FO}}}
\makeatletter
\newcommand{\modelsISL}{\ensuremath{\@ifnextchar\bgroup {\modelsISLargs}{\modelsISLsymb}}}
\newcommand{\modelsISLsymb}{\models_{\scriptscriptstyle\SLE}}
\newcommand{\modelsISLargs}[4]{\ensuremath{#1, #2, #3 \modelsISLsymb #4}}
\def\modelsSLE{\modelsISL}
\newcommand{\fr}{\@ifnextchar\bgroup {\frargs}{\frsymb}}
\newcommand{\frsymb}{\textit{framed}}
\newcommand{\frargs}[5]{\ensuremath{\frsymb(#2,(#3,#4,#5),#1)}}
\makeatother
\newcommand{\notmodelsISL}{\not\models_{\scriptscriptstyle\SLE}}

\newcommand{\SL}{\ensuremath{\textit{SL}}}
\newcommand{\SLE}{\ensuremath{\textit{TPL}}}
\newcommand{\ESL}{\SLE}
\newcommand{\pointsto}[3]{\ensuremath{{#1}\stackrel{#2}{\mapsto}{#3}}}
\newcommand{\first}[1]{\ensuremath{{\downarrow_1}({#1})}}
\newcommand{\second}[1]{\ensuremath{{\downarrow_2}({#1})}}

\newcommand{\restrict}[2]{\ensuremath{{#1}{\upharpoonright}{#2}}}
\newcommand{\condheap}[3]{\ensuremath{(#1\;?\;#2 : #3)}}
\newcommand{\rds}[1]{\ensuremath{{\textit{rds}}({#1})}}
\newcommand{\entails}{\modelsISL}
\newcommand{\valid}{\modelsISL}
\newcommand{\isequiv}{\equiv_{\SLE}}

\newcommand{\extend}[5]{\ensuremath{(#1, #2, #3)\lhd#4 \modelsISL #5 }}

\makeatletter
\def\operator#1{\ensuremath{\@ifnextchar\bgroup {\operatorarg{#1}}{#1}}}
\def\operatorarg#1#2{{#1}{(#2)}}

\newcommand{\eframed}[1]{\mathsf{ExtFrm}(#1)}
\newcommand{\deframed}[1]{\mathsf{DisExtFrm}(#1)}
\newcommand{\globalExts}[1]{\mathsf{globalExts}(#1)}
\newcommand{\globalDisjExts}[1]{\mathsf{globalDisjExts}(#1)}
\newcommand{\localDisjExts}[1]{\mathsf{localDisjExts}(#1)}

\def\dom{\operator{\textit{dom}}}
\def\Env{\operator{\ensuremath{`s}}}

\makeatother

\newcommand{\acc}{\textit{acc}}

\newskip \point \point =1pt
\setbox134=\hbox{\leavevmode\raise0\point\hbox{${\langle}\kern-2.5\point{\langle}$}}
\setbox135=\hbox{\raise0\point\hbox{${ \rangle}\kern-2.5\point{ \rangle}$}}
\def \semEl#1{\copy134{#1}\copy135}

\newcommand{\toSO}[1]{\llfloor #1 \rrfloor}
\setbox136=\hbox{\leavevmode\raise0\point\hbox{${\lfloor}\kern-3\point{\lfloor}$}}
\setbox137=\hbox{\raise0\point\hbox{${ \rfloor}\kern-3\point{ \rfloor}$}}

\setbox138=\hbox{\leavevmode\raise0\point\hbox{${[}\kern-1.8\point{[}$}}
\setbox139=\hbox{\raise0\point\hbox{${]}\kern-1.8\point{]}$}}
\newcommand{\semSLE}[2]{\llbracket{#1}\rrbracket_{#2}}
\newcommand{\semCE}[2]{\llbracket{#1}\rrbracket_{#2}}
\let\toSOE=\toSO
\newcommand{\SLtoC}[1]{\ensuremath{[#1]}}

\newcommand{\imp}{\rightarrow}

\newcommand{\as}[1]  {{{#1}}}

\newcommand{\newlist}[1]{
   \newcounter{#1}
   \setcounter{#1}{0}
   \expandafter\newcommand\csname listcomm#1\endcsname{}
   \expandafter\newcommand\csname#1list\endcsname%
	{\csname listcomm#1\roman{#1}\endcsname}%
   \expandafter\newcommand\csname add#1\endcsname[1]{
      \addtocounter{#1}{1}
      \expandafter\newcommand\csname listcomm#1\roman{#1}\endcsname%
      { \addtocounter{#1}{-1} \csname listcomm#1\roman{#1}\endcsname \par ##1 \addtocounter{#1}{1}}
   }
}

\newlist{proofs}

\newcommand{\dispproof}[6]{

\medskip\noindent\begin{minipage}{\linewidth}
\begin{#1}[#2]
\label{#3}
#4%
\ifthenelse{\equal{}{#5}}{
\hfill  (See page~\pageref{#3:proof} for proof.)
\end{#1}
\end{minipage}
}{%

\end{#1}
\end{minipage}\medskip

\noindent\begin{minipage}{\linewidth}\emph{Proof sketch:}
  #5
\hfill (Full proof on page~\pageref{#3:proof})
\end{minipage}}
}

\newcommand{\dispAppproof}[6]{

\noindent
\begin{minipage}{\linewidth}
\textbf{\MakeUppercase#1 \ref{#3}}\ifthenelse{\equal{}{#2}}{}{\ (#2)}\textbf{.}
\label{#3:proof}
#4
\end{minipage}
\proof
#6
}

\newcommand{\addprop}[6]{
\dispproof{#1}{#2}{#3}{#4}{#5}{#6}
\addproofs{\dispAppproof{#1}{#2}{#3}{#4}{#5}{#6}}
}
\begin{document}

\newcommand{\SomethingLogic}{Total Heaps Permission Logic}
\title{The Relationship Between Separation Logic and Implicit Dynamic Frames}

\author[M.~J.~Parkinson]{Matthew J.~Parkinson\rsuper a}
\address{{\lsuper a}Microsoft Research Cambridge}
\email{mattpark@microsoft.com}
\author[A.~J.~Summers]{Alexander J.~Summers\rsuper b}
\address{{\lsuper b}ETH Zurich}
\email{alexander.summers@inf.ethz.ch}

\keywords{Separation Logic, Implicit Dynamic Frames, Verification Logics, Partial Heaps, Permissions, Concurrency}
\subjclass{F.3.1}

\begin{abstract}
Separation logic is a concise method for specifying programs that manipulate dynamically allocated storage. Partially inspired by separation logic, Implicit Dynamic Frames has recently been proposed, aiming at first-order tool support. In this paper, we precisely connect the semantics of these two logics. We define a logic whose syntax subsumes both that of a standard separation logic, and that of implicit dynamic frames as sub-syntaxes. We define a total heap semantics for our logic, and, for the separation logic subsyntax, prove it equivalent the standard partial heaps model. In order to define a semantics which works uniformly for both subsyntaxes, we define the novel concept of a \emph{minimal state extension}, which provides a different (but equivalent) definition of the semantics of separation logic implication and magic wand connectives, while also giving a suitable semantics for these connectives in implicit dynamic frames. We show that our resulting semantics agrees with the existing definition of weakest pre-condition semantics for the implicit dynamic frames fragment. Finally, we show that we can encode the separation logic fragment of our logic into the implicit dynamic frames fragment, preserving semantics. For the connectives typically supported by tools, this shows that separation logic can be faithfully encoded in a first-order automatic verification tool (Chalice).
\end{abstract}

\maketitle

\section{Introduction}
Separation logic (SL)~\cite{IshtiaqOHearn'01,OHearnReynoldsYang'01} is a popular approach to specifying the behaviour of programs, as it naturally deals with the issues of aliasing.  Separation logic assertions extend classical logic with extra connectives and predicates to describe memory layout. This makes it difficult to reuse current tool support for verification.  \emph{Implicit dynamic frames} (IDF)~\cite{Smans'09} was developed to give the benefits of separation logic specifications, while leveraging existing tool support for first-order logic.

Although IDF was partially inspired by separation logic, there are many differences between SL and IDF that make understanding their relationship difficult.  SL does not allow expressions that refer to the heap, while IDF does.  SL is defined on partial heaps, while IDF is defined using total heaps and permission masks.  The semantics of IDF are only defined by its translation to first-order verification conditions, while SL has a direct Kripke semantics for its assertions.  These differences make it challenging to understand the relationship between the two approaches.

In this paper, we investigate the formal relationship between the two approaches. As a medium for this comparison, we define a verification logic (which we name \emph{\SomethingLogic}) whose syntax includes both that of a typical separation logic, and that of implicit dynamic frames. We define a semantics for \SomethingLogic{} based on states which incorporate a \emph{total heap} and a separate \emph{permission mask}, that we show both captures the original semantics of separation logic, and correctly captures the semantics of IDF.
 Intuitively, the permission mask specifies the locations in the heap which are safe to access.  Our formulation allows expressions that access the heap to be defined, and this complicates the definition of the separation logic ``magic wand'' and implication connectives.  
In order to define a suitable semantics for these connectives which is compatible with both approaches, we introduce the novel concept of \emph{minimal extensions} of a state, and use this to define a novel semantics for these connectives, which nonetheless agrees with the original semantics for the separation logic fragment of our logic. Correctly reflecting the standard semantics of the separating conjunction and magic wand allows us to use these connectives to define the usual separation logic notion of weakest pre-conditions of commands.

In order to show that our logic correctly captures the semantics of the IDF formulas, we focus on the form of IDF found in the concurrent verification tool Chalice~\cite{LeinoM09}. As the semantics of IDF formulas are only defined indirectly via weakest pre-condition calculations for a language using them, we show that the verification conditions (VCs) generated by the existing Boogie2~\cite{LeinoBoogie2} encoding and the VCs generated from the separation logic proof rules are logically equivalent.  This shows that our model directly captures the existing semantics of IDF.

We make use of these strong correspondences to define an encoding of separation logic into implicit dynamic frames that preserves semantics. We then define a subsyntax of separation logic (corresponding to the logical connectives supported by many practical tools), which maps onto the assertion language supported by Chalice, and show that this fragment of separation logic can, via our correspondences, be handled in a purely first-order prover.

%
%

\subsubsection*{Outline}
The paper is structured as follows. We begin by presenting the background definitions of both separation logic and implicit dynamic frames (\S\ref{sec:background}). We then provide an overview of the challenges in defining our logic and semantics, and present \SomethingLogic{} (\S\ref{sec:merged}), characterising various properties of our total heap semantics. We prove the correspondence between VCs as calculated in separation logic and in implicit dynamic frames (\S\ref{sec:vc}), and then combine our proven results to show how to map a fragment of separation logic into contracts which can be verified by the Chalice tool, preserving their original semantics (\S\ref{sec:mapping}).  Finally, we discuss related work (\S\ref{sec:relatedwork}), consider possible extensions and conclude (\S\ref{sec:extensions}).


The contributions of this paper are as follows:
\begin{iteMize}{$\bullet$}
\item We define a total heaps semantics for a logic whose syntax subsumes a separation logic, and prove that, for the separation logic fragment, our total heaps semantics is equivalent with the standard (partial heaps) semantics for the separation logic.
\item We define a direct semantics for the implicit dynamic frames logic (the specification logic of the Chalice tool), which has so far only been given a semantics implicitly, via verification conditions.
\item We show how to encode a standard fragment of separation logic into an implicit dynamic frames setting, preserving its semantics.
\item We show that verification conditions as computed for separation logic coincide via our translation and semantics with the verification conditions computed by Chalice.
\item We present the notion of \emph{minimal extensions} of a state, and show how it can be used to define the semantics of the separation logic implication and magic wand connectives in a new way.
\end{iteMize}

\subsubsection*{Extensions with regard to the conference version}
This paper extends the conference version \cite{ParkinsonSummers'11} by providing a different definition of implication that corresponds to that used in Chalice.  The conference version provided a definition of implication that was correct with respect to separation logic, but on the formulas used in Chalice it had undesirable behaviour.  We have altered the definitions of implication and magic wand to correctly model both Chalice and separation logic.

The paper provides extended discussions of the design of the logic, detailing the requirements which come from each of our target logics. We explicitly define and discuss the various notions of \emph{state extension} which were used implicitly in the formulations of the technical definitions in our precursor paper. The semantics of implication in the logic is discussed in detail, and a new concept of \emph{minimal extensions} is used to obtain a semantics which works well for both target logics. The resulting semantics is formulated differently from the traditional presentation of implication in intuitionistic separation logic; our definition requires checking the subformulas in fewer states.

The syntactic condition on when a Chalice assertion was considered self-framing in the conference paper was overly restrictive, in that it did not reflect that Chalice takes account of the restrictions provided by an assertion: for instance, $acc(x.f,1) * y=x * y.f=5$ would not have been considered self-framing in the conference version, but is in this paper, and in Chalice.

We provide a new section (\S\ref{sec:mapping}) which shows how to combine the previously-proved results to explicitly show that a fragment of separation logic can be equivalently verified using separation logic weakest pre-conditions, or (via an encoding) using implicit dynamic frames specifications and weakest pre-condition calculations.

Finally, we provide full details of all proofs.

\section{Background and Motivation}
\label{sec:background}
\subsection{Standard Separation Logic}
Separation logic~\cite{IshtiaqOHearn'01,OHearnReynoldsYang'01} is a verification logic which was originally introduced to handle the verification of sequential programs in languages with manual memory management, such as C. The key feature of the logic is the ability to describe the behaviour of commands in terms of disjoint heap fragments, greatly simplifying the work required when ``framing on'' extra properties in a modular setting. Since its inception, separation logic has evolved in a variety of ways. In particular, variants of separation logic are now used for the verification of object-oriented languages with garbage collection, such as Java and C${}^\sharp$ \cite{parkinson_thesis}.

In order to handle concurrency, separation logic has been extended to consider its basic points-to assertions as \emph{permissions}~\cite{ohearn:csljournal}, determining which thread is allowed to read and write the corresponding state. To gain flexibility, \emph{fractional permissions}~\cite{boyland03,bornat05} were introduced, allowing the permissions governed by points-to assertions to be split and recombined. A fractional permission is a rational number $\none<`p\leq \full$, where $\full$ denotes full and exclusive (read/write) permission, and any other permission denotes read-only permission. In this paper we focus on the following core fragment of separation logic with fractional permissions.
\begin{definition}[Separation Logic Assertions (\SL)]\label{defn:seplogicsyntax}
We assume a set of \emph{object identifiers}\footnote{These could be considered to be addresses, but we choose to be parametric with the concrete implementation of the heap.}, ranged over by $`i$. We also assume a set of \emph{field identifiers}, ranged over by $f$. \emph{Values}, ranged over by $v$ are either object identifiers, integers, or the special value $\Null$.

 The syntaxes of separation logic expressions (ranged over by $e$) and assertions (ranged over by $a$) are defined as follows\footnote{Note that variables $x$ need not be program variables, but can also be specification-only variables (sometimes called \emph{logical}, \emph{ghost} or \emph{specification variables})}. In this definition, $n$ ranges over integer constants, and $\none<`p\leq \full$.
\[   \begin{array}{rcl}
   e &\ ::=\ & x \mid \Null \mid n\\
   a &\ ::=\ & e=e\ \mid\ \pointsto{e.f}{`p}{e}\ \mid\ a * a\ \mid\ a\wand a\ \mid\ a \wedge a \ \mid\ a \vee a\ \mid\ a\imp a
\mid \exists x.\; a
   \end{array}\]
 We will refer to this separation logic simply as \SL{} hereafter.

\end{definition}
The key feature of separation logic is the facility to reason locally about separate heap portions. As such, the standard semantics for separation logic is formulated in terms of judgements parameterised by partial heaps (sometimes called \emph{heap fragments}), which can be split and combined together as required. The critical new connectives are the \emph{separating conjunction} $*$, and the \emph{magic wand} $\wand$. The separating conjunction $a_1 * a_2$ expresses that $a_1$ and $a_2$ are true and depend on disjoint fragments of the heap. The magic wand $a_1\wand a_2$ expresses that if any extra partial heap satisfying $a_1$ is combined with the current partial heap, then the resulting heap is guaranteed to satisfy $a_2$.

Fractional permissions\footnote{Chalice, described in the next subsection, actually uses a slight variation on fractional permissions to make automatic theorem proving easier.}~\cite{bornat05,boyland03} are employed to manage shared memory concurrency in the usual way - a thread may only read from a heap location if it has a non-zero permission to the location, and it may only write to a location if it has the whole (full) permission to it. By careful permission accounting, it can then be guaranteed that a thread can never modify a heap location while another thread can read it. Note that permissions are handled (via points-to predicates $\pointsto{e.f}{`p}{e'}$) on a per-field basis: it is possible for an assertion to provide permission for only one field of an object. This fine granularity of permissions allows for greater flexibility in the resulting logic - it can be specified that different threads have access to different fields of an object at the same time, for example. Combination of partial heaps includes combination of their permissions, where they overlap.

\begin{definition}[Partial Fractional Heaps~\cite{bornat05}]
\mbox{}\label{PartialFractionalHeap}
\begin{iteMize}{$\bullet$}
\item
A \emph{partial fractional heap} \PHeap{} is a partial function from pairs $(`i,f)$ of object-identifier and field-identifier to pairs $(v,`p)$ of value and \emph{non-zero}
permission $`p$. Partial heap lookup is written $\PHeap[`i,f]$, and is only defined when $(`i,f)\in\dom(\PHeap)$.
\item
Partial heap extension: $\PHeap_1\subseteq\PHeap_2$, iff $\forall (`i,f)\in\dom(\PHeap_1).\  \first{\PHeap_1[`i,f]} = \first{\PHeap_2[`i,f]}$
and $\second{\PHeap_1[`i,f]} \leq \second{\PHeap_2[`i,f]}$.
\item Partial heap compatibility: $\PHeap_1\perp\PHeap_2$ iff $\forall (`i,f)\in\dom(\PHeap_1)\cap\dom(\PHeap_2).\ \first{\PHeap_1[`i,f]} = \first{\PHeap_2[`i,f]}\ \wedge\  \second{\PHeap_1[`i,f]} + \second{\PHeap_2[`i,f]} \leq \full $.
\item The combination of two partial heaps, written $\PHeap_1 * \PHeap_2$, is defined only when $\PHeap_1\perp\PHeap_2$ holds, by the following equations:
\[\begin{array}{l}
\dom(\PHeap_1 * \PHeap_2) = \dom(\PHeap_1)\cup\dom(\PHeap_2)\\
\forall(`i,f)\in\dom(\PHeap_1 * \PHeap_2). \\
\begin{array}{ll}(\PHeap_1 * \PHeap_2)[`i,f] = \left\{
   \begin{array}{ll}
    \multicolumn{2}{l}{(\first{\PHeap_1[`i,f]},\second{\PHeap_1[`i,f]})
         \hfill \textit{ if }(`i,f)\notin \dom(\PHeap_2)}\\
    \multicolumn{2}{l}{(\first{\PHeap_2[`i,f]},\second{\PHeap_2[`i,f]})
         \hfill  \textit{ if }(`i,f)\notin \dom(\PHeap_1)}\\
    (\first{\PHeap_1[`i,f]},(\second{\PHeap_1[`i,f]}+ \second{\PHeap_2[`i,f]}))
         & \textit{ otherwise }\\
   \end{array}\right.\\
\end{array}\end{array}\]
\end{iteMize}
We use $\downarrow_n$ to denote the $n$th component of a tuple.
\end{definition}

There are two main flavours of separation logic studied in the literature: \emph{classical} separation logic, and \emph{intuitionistic} separation logic~\cite{IshtiaqOHearn'01}. In this paper, we consider  \emph{intuitionistic} separation logic. In intuitionistic separation logic, truth of assertions is closed under heap extension, which is appropriate for a garbage-collected language such as Java/C$^\sharp$, rather than a language with manual memory management, such as C. The standard intuitionistic separation logic semantics for our fragment \SL{} is defined as follows~\cite{parkinson_thesis}.
\begin{definition}[Standard Semantics for \SL~\cite{bornat05}]\label{defn:seplogicsemantics}
Environments $\Env$ are partial functions\footnote{However, we assume that all applications of environments are well-defined; i.e., whenever we write $\Env(x)$, that $x\in\dom(\Env)$. This assumption is justified so long as the program and specifications are type-checked appropriately.} from variable names to values. Separation logic expression semantics, $\semSLE{e}{\Env}$ are defined by $\semSLE{x}{\Env} = \Env(x)$, $\semSLE{n}{\Env} =n$ and $\semSLE{\Null}{\Env} = \Null$. The semantics of assertions is then as follows:
\[
  \begin{array}{ll}
    \PHeap, \Env \modelsSL \pointsto{e_1.f}{`p}{e_2}
    &\iff
    \second{\PHeap[\semSLE{e_1}{\Env}, f]} \geq `p \;\; \land \;\;\first{\PHeap[\semSLE{e_1}{\Env}, f]} = \semSLE{e_2}{\Env}
    \\[0.5em]
    \PHeap, \Env \modelsSL e = e' &\iff
    \semSLE{e}{\Env} = \semSLE{e'}{\Env}
    \\[0.5em]
    \PHeap, \Env \modelsSL a_1 * a_2  &\iff
    \exists \PHeap_1, \PHeap_2.(\;
   \PHeap = \PHeap_1 * \PHeap_2 \ \land\
    \PHeap_1, \Env \modelsSL a_1
     \;\; \land \;\;
     \PHeap_2, \Env \modelsSL a_2)
    \\[0.5em]
    \PHeap, \Env \modelsSL a_1 \wand a_2  &\iff
    \forall \PHeap'.(\;
   \PHeap' \perp \PHeap\  \land \
    \PHeap', \Env \modelsSL a_1 \;\; \Rightarrow\;\;
     \PHeap{*}\PHeap', \Env \modelsSL a_2)
    \\[0.5em]
    \PHeap, \Env \modelsSL a_1 \wedge a_2  &\iff
    \PHeap, \Env \modelsSL a_1 \;\; \land \;\;
     \PHeap, \Env \modelsSL a_2
    \\[0.5em]
    \PHeap, \Env \modelsSL a_1 \vee a_2  &\iff
    \PHeap, \Env \modelsSL a_1 \;\; \vee \;\;
     \PHeap, \Env \modelsSL a_2
    \\[0.5em]
    \PHeap, \Env \modelsSL a_1 \imp a_2  &\iff
    \forall \PHeap'.(\;
    \PHeap' \perp \PHeap\  \land \
    \PHeap{*}\PHeap', \Env \modelsSL a_1 \;\; \Rightarrow\;\;
     \PHeap{*}\PHeap', \Env \modelsSL a_2)
    \\[0.5em]
    \PHeap, \Env \modelsSL \exists x. \; a  &\iff
    \exists v.(\;
	\PHeap, \Env[x \mapsto v] \modelsSL a)
  \end{array}
\]\smallskip
\end{definition}

\noindent The semantics for the separating conjunction and magic wand express the required splitting and combination of partial heaps. The semantics for logical implication $\imp$ 
considers all possible extensions of the current heap, so that assertion truth is closed under heap extension~\cite{IshtiaqOHearn'01}. In examples, we will sometimes write $\pointsto{e.f}{`p}{\_}$ as a shorthand for $\exists x. \; \pointsto{e.f}{`p}{x}$.

\subsubsection{Assume/Assert}\label{sec:assumeassert}
Verification in Boogie2 \cite{LeinoBoogie2} and related technologies uses two commands commonly to encode verification:
\assume\ $A_1$ and \assert\ $A_1$.  The first allows the verification to work forwards with the additional assumption of $A_1$, while the second requires $A_1$ to hold otherwise it will be considered a fault.
These can be given weakest precondition semantics of:

\noindent $\qquad \qquad \textit{wp}(\assert\ A_1,A_2) = A_1 \land A_2 \qquad \qquad \textit{wp}(\assume\ A_1,A_2) =  A_1 \Rightarrow A_2$

\noindent From a verification perspective, these primitives can be used to encode many advanced language features. For example, in a modular verification setting with a first-order assertion language, a method call can be encoded by a sequence $\assert\ \textit{pre};$ $\havoc(\textit{Heap});\ \assume\ \textit{post}$, in which $\textit{pre}$ and $\textit{post}$ are the pre- and post-conditions of the method respectively, and $\havoc(.)$ is a Boogie command that causes the prover to forget all knowledge about a variable/expression.

With separation logic, there are two forms of conjunction and implication, the standard (additive) ones $\land$ and $\imp$, and the separating (multiplicative) ones $*$ and $\wand$. This naturally gives rise to a second form of assume and assert for the multiplicative connectives (\assumeSL\ and \assertSL), with the following weakest precondition semantics:

\noindent $\qquad \qquad \textit{wp}(\assertSL\ A_1,\ A_2) = A_1 * A_2 \qquad \qquad \textit{wp}(\assumeSL\ A_1,\ A_2) = A_1 \wand A_2$

\noindent These commands can be understood as follows: $\assertSL\ A_1$ removes a heap fragment satisfying $A_1$, and $\assumeSL\ A_1$ adds a heap fragment satisfying $A_1$. In a verification setting where assertions express permissions as well as functional properties, these can be used to correctly model the transfer of permissions when encoding various constructs. In a separation logic setting, a method call can be encoded as $\assertSL\ \textit{pre}; \assumeSL\ \textit{post}$.

In Chalice, which handles an assertion logic based on implicit dynamic frames, functional verification is based on two new commands: \exhale\ $A_1$ and \inhale\ $A_1$, which are also given an intuitive semantics of removing and adding access to state.  One outcome of this paper is to make this intuitive connection between \exhale/\inhale\ and \assertSL/\assumeSL\ formal, by defining a concrete and common semantics which can correctly characterise both assertion languages.

\subsection{Chalice and Implicit Dynamic Frames}
The original concept of Dynamic Frames comes from the PhD thesis of Kassios \cite{Kassios_thesis,Kassios'06}. The idea is to tackle the frame problem by allowing method specifications to declare the portion of the heap they may modify (a ``frame'' for the method call) via functions of the heap. The computed frames are therefore dynamic, in the sense that the actual values determined by these functions may change as the heap itself gets modified. Implicit dynamic frames~\cite{Smans'09,Smans_thesis} takes a different approach to computing frames - a first-order logic is extended with a new kind of assertion called an \emph{accessibility predicate} (written e.g., as $\acc(x.f)$) whose role is to represent a permission to a heap location $x.f$. In a method pre-condition, such an accessibility predicate indicates that the method requires permission to $x.f$ in order to be called - usually because this location might be read or written to in the method implementation. By imposing the restriction that heap dereference expressions (whether in assertions or in method bodies) are only allowed if a corresponding permission has already been acquired, this specification style allows a method frame to be calculated implicitly from its pre-condition.

Chalice~\cite{LeinoM09} is a tool written for the automatic verification of concurrent programs. It handles a fairly simple imperative language, with classes (but no inheritance), and several interesting concurrency features (locks, channels, fork/join of threads). The tool proves partial correctness of method specifications, as well as absence of deadlocks. The core of the methodology is based on the implicit dynamic frames specification logic, using accessibility predicates to handle the permissions necessary to avoid data races between threads.

In this paper we ignore the deadlock-avoidance aspects of Chalice, and focus on the aspects which guarantee functional correctness. Verification in Chalice is defined via an encoding into Boogie2, in which two intermediate auxiliary Chalice commands $\exhale\ p$ and $\inhale\ p$ are used. These commands reflect the removal and addition of permissions from the state, as well as expressing assertions and assumptions about heap values. For example, method calls are represented by $\exhale\ \textit{pre};\inhale\ \textit{post}$. The command $\exhale\ \textit{pre}$ has the effect of giving up any permissions mentioned in accessibility predicates in $\textit{pre}$, and generating $\assert$ statements for any logical properties such as heap equalities. Dually, $\inhale\ \textit{post}$ has the effect of adding any permissions mentioned in $\textit{post}$ and \emph{assuming} any logical properties.

\begin{definition}[Our Chalice Subsyntax]\label{defn:chalicesyntax}
Expressions $E$, boolean expressions $B$ and assertions $p$  in our fragment of Chalice are given by the following syntax definitions:
\[
  \begin{array}{rcl}
   E & ::= & x \mid n \mid \Null \mid E.f\\
   B & ::= & E=E \mid E \neq E \mid B * B \\
   p & ::= & B \mid \acc(E.f,`p) \mid p * p \mid  B \imp p \\
  \end{array}
\]\smallskip
\end{definition}

\noindent Note that Chalice actually uses the symbol for logical conjunction ($\land$ or \code{\&\&}) where we write $*$ above. However, in terms the semantics of the logic this is misleading - in general it is not the case that $p\land p$ (as written in Chalice) is equivalent to $p$. Chalice's conjunction treats permissions multiplicatively, that is, $\acc(x.f,1{/}2)\land\acc(x.f,1{/}2)$ is equivalent to $\acc(x.f,\full)$, while $\acc(x.f,\full)\land\acc(x.f,\full)$ is actually equivalent to falsity (it describes a state in which we have more than the full permission to the location $x.f$). As we will show, Chalice conjunction is actually directly related to the separating conjunction of separation logic, hence our choice of notation here. Where we use the symbol $\land$ later in the paper, we mean the usual (additive) conjunction, just as in SL or first order logic.

Chalice performs verification condition generation via an encoding into Boogie2, which makes use of two special variables $\permvar$\ and $\heapvar$. The former maps object-identifier and field-name pairs to permissions, in this instance a fractional permission, and is used for bookkeeping of permissions\footnote{Technically, one should think of $\permvar$ as a \emph{ghost variable}, since it does not correspond to real data of the original program.}. The latter maps object-identifier and field-name pairs to values, and is used to model the heap of the real program. These maps can be read from (e.g., $\permvar[o,f]$) and updated (e.g., $\permvar[o,f] := \full$) from within the Boogie2 code, which allows Chalice to maintain their state appropriately to reflect the modifications made by the source program. In particular, the $\inhale$ and $\exhale$ commands have semantics which include modifications to the $\permvar$ map, to reflect the addition or removal of permissions by the program.

The critical aspect of Chalice's approach to data races, is to guarantee that assertions about the heap are only allowed when at least some permission is held to each heap location mentioned. This means that assertions cannot be made when it might be possible for other threads to be changing these locations - all logical properties used in the verification are then made robust to possible interference. This is enforced by requiring that assertions used in verification contracts are \emph{self-framing}~\cite{Kassios'06} - which means that the assertion includes enough accessibility predicates to ``frame'' its heap expressions. For example, the assertion $x.f=5$ is not self-framing, since it refers to the heap location $x.f$ without permission. On the other hand, $(\acc(x.f,\full) * x.f=5)$ is self-framing.

\section{\SomethingLogic\ (\texorpdfstring{\SLE}{SLE})}
\label{sec:merged}
\label{sec:tsl}
\subsection{Race-free Assertions}
In order for a static verification tool to be able to reason soundly about concurrent programs, a crucial aspect is to be able to give a well-defined semantics to the assertion language employed. Since other executing threads may interfere with the execution of code being verified (for example, by writing to heap locations which the current program also accesses; a data race), assertions which describe properties of the heap do not, in general, even have a well-defined semantics. For example, consider the simple assertion $x.f > 5$. Such an assertion only has a well-defined semantics (at verification time) if, at runtime, the heap location $x.f$ is guaranteed not to be subject to a data race. If another thread writes to this location at the ``same time'' as the assertion is checked to hold, the truth of the assertion becomes non-deterministic, depending on the interleaving of the memory accesses by the two threads. This makes any reasoning about expressions such as $x.f$ as expressions in a logical sense unreliable: assertions such as $x.f > x.f$ could even be ``true'' at runtime, due to interference; a behaviour which any useful verifier will struggle to mimic accurately.

For these reasons, verifiers for concurrent programs need to use a verification methodology and assertion language whose semantics avoids data races. Both separation logic and implicit dynamic frames attack this problem by employing notions of (fractional) \emph{permissions}. Permissions permit access to particular heap locations, and can be passed around between threads, with the crucial property that a thread is only allowed to write to a heap location if no other thread holds a permission to the location. By imposing suitable restrictions on the assertion language used for verification, one can then guarantee a data-race-free semantics by passing permissions explicitly along with heap-dependent assertions, and enforcing the policy that an assertion may only mention a heap location if it also carries at least some permission to that location. In implicit dynamic frames, these permissions are represented by ``accessibility predicates'' $\acc(E.f,`p)$, denoting $`p$ permission to location $E.f$. For example, while the assertion $x.f > 5$ does not have a well-defined semantics on its own, the compound assertion $\acc(x.f,`p) * x.f > 5$ does - the presence of the permission to the location $x.f$ guarantees that its value in the heap is robust to interference. More generally, any heap-dependent expression in an assertion can only be given a meaning by its value being fixed with a permission to the appropriate heap locations. A ``self-framing'' assertion is one which is only satisfied in states which carry enough permissions to fix the values of all heap locations on which it depends; not all assertions are self-framing ($x.f > 5$ is not), but only such assertions can generally be used for verification contracts.  The fact that, in implicit dynamic frames, the permission to access a heap location can come from a different part of an assertion (e.g., conjunct) than the constraints on the value at that location, is the main challenge in giving a correct semantics for the logic.

The same challenge does not arise in separation logic, which does not allow heap-dependent expressions, instead providing the special ``points-to'' predicates as the sole way of handling heap accesses. A ``points to'' predicate $\pointsto{e_1.f}{`p}{e_2}$ plays a dual role in the logic - it provides knowledge of the value $e_2$ of the heap location $e_1.f$, and it also provides a permission $`p$ to this location, making the value robust to interference. Because there is no other way to refer directly to heap locations, one cannot ever talk about the value of a location without having some permission to that location. The observation of this dual role leads naturally to the idea of encoding separation logic assertions into implicit dynamic frames by replacing every points-to predicate $\pointsto{e_1.f}{`p}{e_2}$ by a permission and a heap-dependent expression: $\acc(e_1.f,`p) * e_1.f = e_2$. This observation can be used as the basis of a comparison and translation between the two logics. In fact, our approach is to give a uniform semantics for a logic which subsumes both separation logic and implicit dynamic frames constructs, and then show that the primitive constructs of the former can be represented in the latter.

\subsection{Overview of Our Approach}
In order to formally relate the two paradigms of separation logic and implicit dynamic frames, we define a new logic which subsumes both syntaxes. We call this logic \SomethingLogic\ (\SLE). This logic includes as primitives both the ``points-to'' predicates of separation logic, and the ``accessibility predicates'' of implicit dynamic frames, along with an expression syntax which permits heap-dependent expressions. As we will show formally in this paper, this is actually redundant; one can encode the SL-style primitives into implicit dynamic frames. However, our \SLE{} serves as a uniform basis for comparing these two logics. We also include all of the common connectives used in separation logic, which subsume those typically implemented in tools (based on either approach).

Our approach is to define a semantics for \SLE{}, based on states consisting of a stack (giving meaning to variables), a \emph{total heap}, and a \emph{permissions mask} (defining which locations in the total heap have reliable values for the current thread). Because our semantics is defined compositionally, it actually gives a meaning to assertions which are not (by themselves) well-formed in all states. As discussed above, assertions which mention heap-dependent expressions such as $x.f > 5$ are not necessarily well-defined when considered in isolation. However, because the IDF approach allows for such assertions as \emph{subformulas} of a well-formed assertion, and because we want to define a compositional semantics for our logic, we are obliged to give such assertions a semantics, even though (by themselves) they cannot be used in either approach. In some sense, by encompassing both SL and IDF, we actually make our assertion logic \emph{too} general. The presence of ill-formed assertions in our general logic means that (just as in implicit dynamic frames) we will later have to introduce additional concepts such as \emph{self-framing assertions}, in order to identify the fragments of our logic which are well-behaved.

\begin{definition}[\SomethingLogic]
We define the expressions $E$ and assertions $A$ of \emph{\SomethingLogic} (\ESL), by the following grammar (in which $n$ stands for any integer constant):
\[
  \begin{array}{rcl}
   e & ::= & x \mid \Null \mid n\\
   E & ::= & e \mid E.f\\
   A & ::= & E=E \mid \pointsto{E.f}{`p}{E} \mid A * A \mid A\wand A \mid A \wedge A  \mid A \vee A \mid A\imp A  \mid \acc(E.f,`p)
\mid \exists x.\; A\\
  \end{array}
\]\smallskip
\end{definition}

\noindent Note that the syntax of separation logic assertions (ranged over by $a$; see Definition \ref{defn:seplogicsyntax}) is a strict subset of the \ESL{} assertions $A$ defined above. The syntax of separation logic expressions $e$ is also a strict subset of \SLE{} expressions $E$. Similarly, the syntax of Chalice assertions (cf. Definition \ref{defn:chalicesyntax}) is a subset of our \ESL{} syntax.

Our strategy for the rest of the paper is as follows. We will first investigate carefully how to define a suitable total-heaps-based semantics for \ESL{}. In particular, we spend considerable attention on the definition of cases which correspond to modelling state extension (the implication and magic wand connectives).

We will then show that, for the subsyntax which corresponds to separation logic assertions, our total heaps semantics coincides with the traditional partial-heaps-based semantics of the logic (cf. Definition \ref{defn:seplogicsemantics}). Thus, we define a total-heaps model for separation logic, which is consistent with the standard model. We will further show that the subsyntax of \ESL{} which covers SL assertions can be mapped into the IDF subsyntax, preserving the semantics of the assertions. Thus, we can faithfully map from the SL world to an IDF-based assertion language.

In Section \ref{sec:vc}, we will show how to connect our \ESL{} model to the Chalice verification methodology. In particular, we will show that weakest pre-conditions as calculated by Chalice for a first-order theorem prover, are equivalent to weakest pre-conditions as calculated in separation logic. By combining this result with our ability to faithfully reflect traditional SL semantics in \ESL{}, we can show the equivalence of the overall approaches. In particular, for the subsyntax of SL typically supported by automatic tools, we can show that we can encode programs with SL specifications as programs with IDF specifications, and compute equivalent weakest pre-conditions for direct verification by a theorem prover.

\subsection{Total vs Partial Heaps}
One important technical challenge faced in defining a semantics for both logics, is that the semantics of separation logic is defined using \emph{partial heaps} (representing heap fragments, which can be split and recombined), while the implementation of implicit dynamic frames employs a mutable \emph{total heap}, and a separate \emph{permissions mask} to keep track of the permissions held in the current state. In order to make a uniform semantics for the two logics, we needed to bridge this gap between the two paradigms. We achieve this by employing only total heaps and permission masks, and using these to define a semantics that faithfully captures the traditional partial-heaps-based model for the SL subsyntax.

\begin{definition}[Total Heaps and Permission Masks]
A \emph{total heap} \Heap{} is a total map from pairs of object-identifier $o$ and field-identifier $f$ to \emph{values} $v$. Heap lookup is written $\Heap[o,f]$. We write \emph{field location} to mean a pair of object-identifier and field-identifier.\\
\noindent A \emph{permission mask} \Perm{} is a total map from pairs of object-identifier and field-identifier to permissions. Permission lookup is written $\Perm[o,f]$. \\
\noindent We write $\Perm_1\subseteq\Perm_2$ for \emph{permission mask extension}, i.e., $\forall (o,f).\ \Perm_1[o,f] \leq \Perm_2[o,f]$.  \\
\noindent We write $\emptyset$ for the \emph{empty permission mask}; i.e., the mask which assigns $\none$ to all locations.\\
\noindent We write $\rds{\Perm}$ for the set of field locations with non-zero permissions in $\Perm$, that is, $\{ (o,f) \mid \Perm[o,f] > 0 \}$. We write $\overline{\rds{\Perm}}$ for the complement of this set of locations.\\
A \emph{state} is a triple $(\Heap,\Perm,\Env)$ consisting of a heap, a permission mask and an environment $\Env$.
%
\noindent Two permission masks $\Perm_1$ and $\Perm_2$ are \emph{compatible}, written $\Perm_1 \perp \Perm_2$, if it holds that:
\[\forall (o,f).\  \Perm_1[o,f] + \Perm_2[o,f] \leq \full\]
\noindent The \emph{combination} of two permission masks, written $\Perm_1 * \Perm_2$ is undefined if $\Perm_1$ and $\Perm_2$ are not compatible, and is otherwise defined pointwise to be the following permission mask:
\[(\Perm_1 * \Perm_2)[o,f] =
\Perm_1[o,f] + \Perm_2[o,f]
\]
\noindent We define the greatest lower bound of two masks:
\[
(\Perm_1 \sqcap \Perm_2) [o,f] =
\mathsf{min}(\Perm_1[o,f],\Perm_2[o,f])
\]
and the least upper bound of two masks:
\[
(\Perm_1 \sqcup \Perm_2) [o,f] =
\mathsf{max}(\Perm_1[o,f],\Perm_2[o,f])
\]%
\noindent Finally, we define a partial operation of subtraction on permission masks, $\Perm_1 - \Perm_2$. It is defined if and only if $\Perm_2 \subseteq  \Perm_1$, and is defined by:
\[
(\Perm_1 - \Perm_2) [o,f] = (\Perm_1[o,f] - \Perm_2[o,f])
\]
\end{definition}

The crucial observation relating our logic to SL is that, while we will use total heaps in our semantics, we will actually only allow assertions to depend on a subheap; those locations which the current thread can ``read''; i.e., that it has at least some permission to. We employ the notation $\rds{\Perm}$ (where $\Perm$ is a permission mask), to talk about this set of locations. We will design our semantics such that, for all separation logic assertions, their semantics in a state with heap $\Heap$ and permissions $\Perm$ corresponds to their traditional semantics using the (partial) heap obtained by \emph{restricting} $\Heap$ to just the domain $\rds{\Perm}$ (this restriction is formally written as $\restrict{\Heap}{\Perm}$ later). This idea reflects the intuition that all other locations in the (total) heap $\Heap$ have unreliable values (which may be subject to interference from other threads); only assertions which are appropriately ``framed'' by sufficient permissions, can be relied upon in a concurrent setting.

When we want to explicitly express that an assertion is robust to interference from other threads, we can do so by considering the evaluation of the assertion in all heaps which agree with the current one on the locations which the current thread can read, according to the permissions mask. Effectively, we introduce a ``havoc'' of all of the locations to which we hold no permission, and check that the assertion is still guaranteed after all such locations are assigned arbitrary values. In order to define this operation, we introduce the concept of two heaps ``agreeing'' on the permissions in a mask (as well some other heap constructions) as follows:

\begin{definition}[Total Heap Operations]
Two heaps $\Heap_1$ and $\Heap_2$ \emph{agree on a set of object field locations} $F$, written $\Heap_1 \agrees{F} \Heap_2$, if the two heaps contain the same value for each location, i.e.,
\[
\Heap_1 \agrees{F} \Heap_2
\iff
\forall (o,f)\in F.\ \Heap_1[o,f] = \Heap_2[o,f]
\]
Two heaps $\Heap_1$ and $\Heap_2$ \emph{agree on permissions} $\Perm$, written $\Heap_1 \agrees{\Perm} \Heap_2$, if the two heaps agree on all field locations given non-zero permission by $\Perm$, i.e.,
\[
\Heap_1 \agrees{\Perm} \Heap_2
\iff
\Heap_1 \agrees{\rds{\Perm}} \Heap_2
\]
The \emph{restriction of $\Heap$ to $\Perm$}, written $\restrict{\Heap}{\Perm}$ is a \emph{partial fractional heap} (Definition~\ref{PartialFractionalHeap}), defined by:
\[\begin{array}{l}
\dom(\restrict{\Heap}{\Perm}) = \rds{\Perm}\\
\forall(o,f)\in\dom(\restrict{\Heap}{\Perm}).\ (\restrict{\Heap}{\Perm})[o,f] = (\Heap[o,f],\Perm[o,f])\\\end{array}\]
The \emph{conditional merge of $\Heap_1$ and $\Heap_2$ over a set of locations $F$}, written $\condheap{F}{\Heap_1}{\Heap_2}$ is a total heap defined by:
\[ \condheap{F}{\Heap_1}{\Heap_2}[o,f] = \left\{\begin{array}{ll} \Heap_1[o,f] & \textit{ if }(o,f) \in F\\
\Heap_2[o,f] & \textit{ otherwise}\end{array}\right. \]
We write $\condheap{\Perm}{\Heap_1}{\Heap_2}$ as a shorthand for
$\condheap{\rds{\Perm}}{\Heap_1}{\Heap_2}$.
\end{definition}

We can make use of these operations on total heaps to define what it means for an assertion to be \emph{stable} in a certain state (which intuitively means that its truth only depends on heap locations to which it also requires permission to be held).
%
%
\newcommand{\interfere}[2]{\mathsf{interfere}(#1,#2)}


\newcommand{\stablewith}[2]{\mathsf{stable{-}with}_{#1}\,(#2)}

\newcommand{\stable}[1]{\mathsf{stable}(#1)}
\newcommand{\localExts}[1]{\mathsf{localExts}(#1)}

\begin{definition}[Interference, Stability, Self-Framing and Pure Assertions]
Given a heap $\Heap$ and a permissions mask $\Perm$, the \emph{interfered heaps from $\Heap,\Perm$} 
 is a set of heaps defined by:
\[
\interfere{\Heap}{\Perm} = \{ \Heap' \mid \Heap' \agrees{\Perm} \Heap\}
\]
A set of states $S$ is \emph{stable with extra permissions $\Perm$}, written as $\stablewith{\Perm}{S}$, if the extra permissions are sufficient to make the set closed under interference; i.e.,
\[
\stablewith{\Perm}{S} \Leftrightarrow ( \forall (\Heap,\Perm',\Env)\in S.\ \forall \Heap'\in\interfere{\Heap}{\Perm*\Perm'}. (\Heap',\Perm',\Env)\in S)
\]
A set of states $S$ is \emph{stable}, written as $\stable{S}$, if the set is closed under interference; i.e.,
\[
\stable{S} \Leftrightarrow \stablewith{\emptyset}{S}
\]
We write $\semEl{A}$ to denote the set of states in which the assertion $A$ is true (the actual definition of our semantic judgement $\Heap,\Perm,\Env \modelsISL A$ will come later).
\[
\semEl{A} = \{ (\Heap,\Perm,\Env) \mid    \Heap,\Perm,\Env \modelsISL A \}
\]
An assertion $A$ is \emph{self-framing} if and only if the set of states satisfying it is stable; i.e., if $\stable{\semEl{A}}$ is true.

An assertion $A$ is \emph{pure} 
 if and only if it doesn't depend on permissions, i.e.,
\[\forall \Heap,\Perm,\Env.\ ((\Heap,\Perm,\Env)\in\semEl{A} \Rightarrow (\Heap,\emptyset,\Env)\in\semEl{A}) \]
\end{definition}
Intuitively, self-framing assertions are robust to arbitrary interference on the rest of the heap. For separation logic assertions, this property holds naturally, since it is impossible for an assertion to talk about the heap without including the appropriate ``points-to'' predicates, which force the corresponding permissions to be held. This is shown as a corollary (Corollary~\ref{lemma:seplogicframing}) of the main theorem in this section.

On the other hand, pure assertions which depend on heap values (such pure assertions are not supported in separation logic, but are employed in implicit dynamic frames) are naturally not self-framing. An assertion such as $x.f=5$ is considered pure (it does not mention any permissions or points-to predicates); it will be true in a state where we have no permissions, but in which the value of the heap location $x.f$ is $5$. Nonetheless, such a state is not stable; when we allow for interference, the value of $x.f$ can be modified, and the truth of the assertion need not be preserved.
%

\subsection{Pure Assertions and Separating Conjunction}

Our assertion language includes the \emph{separating conjunction} $A*B$ of separation logic (recall Definition \ref{defn:seplogicsemantics}), and permissions can be distributed multiplicatively across this conjunction. In particular, our semantics needs to enforce, for an assertion $A*B$, that the permissions required are the \emph{sum} of the permissions required in each of $A$ and $B$; we can model this by checking that we can split the permissions mask into two parts, using them to judge the respective conjuncts. A question which then arises is, what should happen to the heap when judging the separating conjunction? In the traditional separation logic semantics, which uses partial heaps, since the heap values and ``permissions'' are both tracked together in partial heap chunks, it is natural to divide up the partial heap in this case. In the case of partial fractional heaps, the two resulting heap chunks can still share values, but only to those locations in which both parts hold some permission. With a total-heaps semantics, we have a choice as to how to reflect this ``splitting''; we can either only split the permissions mask across the conjuncts, but leave the heap unchanged, or we can also try to simulate the splitting of heap values, say, by throwing away information about certain heap locations when judging the individual conjuncts. By looking only at the separation logic fragment of our logic, we cannot see a clear advantage either way; since all assertions in that logic can only read the same heap values that they provide permission to, the question of what should be done with other heap values is irrelevant there (and a partial heaps model gives no reasonable way to even phrase this decision, since it conflates the notions of permission to read a location, and the action of actually reading it). However, this question is pertinent in the case of implicit dynamic frames, where we have heap-dependent expressions which can occur in pure assertions.

For pure assertions, we want the property that, even when mentioned in a separating conjunction, they do not actually extend the ``heap footprint'' of what the assertion requires. In particular, we would like to retain the law (which holds in intuitionistic separation logic), that $A_1*A_2$ is equivalent to $A_1\wedge A_2$ when either of the two conjuncts are pure. In particular, this motivates that pure assertions should be allowed to depend on the same state as assertions they are conjoined with. Of course, in separation logic, where pure assertions are syntactically restricted to not mention the heap, this ``same state'' just means the environment $\Env$. But in implicit dynamic frames, we would like heap-dependent pure assertions to be allowed to depend on the same heap values that other conjuncts make readable by providing permissions; when interpreting assertions such as $\acc(x.f)*x.f==5$ this is exactly what we want.

For these reasons, in our total-heaps model, we define the semantics for separating conjunction with a split of the permissions mask, but no change to the heap. This is concretely achieved by checking that we can split $\Perm$ into two pieces, each of which are sufficient to judge the two sub-formulas; the particular definition (which will be provided as part of the definition of our full semantics later) is:
\[
   \begin{array}{l}
    \Heap, \Perm, \Env \modelsISL A_1 * A_2 \iff\\
 \ \ \ \    \exists \Perm_1, \Perm_2.\;(
   \Perm = \Perm_1 * \Perm_2 \ \land \ \Heap, \Perm_1, \Env \modelsISL A_1
     \ \land\
     \Heap, \Perm_2, \Env \modelsISL A_2)
     \end{array}
\]
%
In the case of some implicit dynamic frames assertions, this rule for treating separating conjunction may ``separate'' a heap-dependent expression from the permission used to fix its values. For example, consider a permissions mask $\Perm$ in which we have full permission to the location $x.f$ (and no other permissions), and a heap $\Heap$ in which $x.f$ has the value $5$. The assertion $\acc(x.f,1) * x.f=5$ is true in such a state. But, in treating the separating conjunction, we are forced to split $\Perm = \Perm_1*\Perm_2$ and put \emph{all} permission to $x.f$ into $\Perm_1$, in order to satisfy the left conjunct, leaving $\Perm_2$ to be the empty permissions mask. The fact that the sub-formula $x.f=5$ is eventually judged in a state in which we hold no permissions to the relevant location $x.f$ is not a problem - we only need to be sure that permission to this location is held \emph{somewhere} in the whole assertion, and not in this particular sub-formula. That is, the property of an assertion being well-formed (self-framing) is not enforced for its sub-formulas, but only for the assertion as a whole.

\subsection{Modelling Partial Heap Extensions}

One of the most difficult technical challenges in the design of our semantics was correctly handling the \emph{magic wand} ($\wand$) and \emph{implication} ($\imp$) connectives. In the traditional partial heaps semantics of separation logic (Definition \ref{defn:seplogicsemantics}), the semantics of both of these connectives involve considering \emph{extensions} of the current heap. In a semantics based on partial heaps, this is rather straightforward, but with total heaps and permission masks it is not so obvious how to model ``heap extension'' for these connectives.

One simple option is just to consider permission extension - leave the heap unchanged but consider all larger permissions masks. The problem with this rather-simplistic proposal is that it attaches significance to pre-existing values in our total heap even in the case where we previously had no permission to them. Since such values are generally meaningless, this doesn't give a well-behaved semantics. When one compares with the operation of extending partial heaps, which we are trying to simulate appropriately, it can be seen that the approach doesn't work; when a new heap location is added in a partial heaps model, it can take any value, whereas our total heap only has one value at any one time.

In order to avoid tying ourselves down to the values in our total heaps which are not necessarily currently meaningful, we can instead model heap extension by adding on extra permission and then \emph{havocing} (i.e., assigning arbitrary values to) the heap locations to which we have newly acquired permission. In this way, we make the original heap values stored at these locations irrelevant, and correctly reflect the general operation of adding on a fresh heap location in a partial-heaps-based model. To this end, we define several variants of this idea of how to model state extensions. The differences in the variants come from two decisions. Firstly, when we add on new permission, do we havoc the heap values at all locations to which we previously held no permission (we call this a \emph{global havoc}), or only those which the permissions newly allow us to read (we call this a \emph{local havoc})? Secondly, when we define the extensions of a state, are we interested in resulting states in which we combine the new permissions with those we held previously, or do we just want to describe the ``extra'' disjoint part of the state (using the new permissions, but not the old ones)? These questions give rise to the following four concepts:

\begin{definition}[State Extensions]\label{defn:extensions}
The set of \emph{locally-havoced extensions} of a state $(\Heap,\Perm,\Env)$ is the set of states in which extra permission is added, and possibly-new values are assigned to the newly-readable locations, i.e.,:
\[
\localExts{\Heap,\Perm,\Env} = \{ (\Heap',\Perm*\Perm',\Env) \mid \Perm'{\perp}\Perm \wedge \Heap' \agrees{\rds{\Perm}\cup\overline{\rds{\Perm'}}} \Heap \}
\]
The set of \emph{globally-havoced extensions} of a state $(\Heap,\Perm,\Env)$ is the set of states in which extra permission is added, and possibly-new values are assigned to the previously-unreadable locations, i.e.,:
\[
\globalExts{\Heap,\Perm,\Env} = \{ (\Heap',\Perm*\Perm',\Env) \mid \Perm'{\perp}\Perm \wedge \Heap' \in\interfere{\Heap}{\Perm} \}
\]
The set of \emph{locally-havoced disjoint extensions} of a state $(\Heap,\Perm,\Env)$ is the set of states in which extra permission is added, possibly-new values are assigned to the newly-readable locations, and only the extra permissions are kept in the results, i.e.,:
\[
\localDisjExts{\Heap,\Perm,\Env} = \{ (\Heap',\Perm',\Env) \mid \Perm'{\perp}\Perm \wedge \Heap' \agrees{\rds{\Perm}\cup\overline{\rds{\Perm'}}} \Heap  \}
\]
The set of \emph{globally-havoced disjoint extensions} of a state $(\Heap,\Perm,\Env)$ is the set of states in which extra permission is added, possibly-new values are assigned to the previously-unreadable locations, and only the extra permissions are kept in the results,
\[
\globalDisjExts{\Heap,\Perm,\Env} = \{ (\Heap',\Perm',\Env) \mid \Perm'{\perp}\Perm \wedge \Heap' \in\interfere{\Heap}{\Perm} \}
\]\smallskip
\end{definition}

As we will see later in this section, we have uses for all of these notions of state extension, and choosing the appropriate one at various points is important for our logic to have the right semantics.

\subsection{Minimal Extensions and Implication}\label{sect:implication}

We now consider how to design an appropriate semantics for implication in our logic, which should manage to work appropriately both for the separation logic and implicit dynamic frames fragments. Recall that, in the traditional semantics of (intuitionistic) separation logic (cf. Definition \ref{defn:seplogicsemantics}), an implicative assertion $a_1\imp a_2$ is true if, in all extensions of the current heap, whenever $a_1$ is true then $a_2$ is also true. Therefore, we need to be careful to appropriately model this idea of extending the current state when judging an implication. We use three examples here to guide the discussion of our design
:
\begin{eqnarray*}
\textit{Ex. 1:} &\pointsto{x.f}{1}{\_} * (\pointsto{x.f}{1}{5}\ \imp\ \pointsto{y.g}{1}{\_})\\
\textit{Ex. 2:} &\acc(x.f,1) * ((\acc(x.f,1) * x.f=5)\ \imp\ \acc(y.g,1))\\
\textit{Ex. 3:} &\acc(x.f,1) * (x.f=5 \imp \acc(y.g,1))
\end{eqnarray*}
In (intuitionistic) separation logic, the first formula is actually only true in states which have (full) permission to \emph{both} locations $x.f$ and $y.g$. The reason is that, in judging the implication subformula, we have to consider all extensions of the provided state. Unless the state in which we judge the implication has at least some permission to $x.f$ (and gives a value other than $5$ to this location), then when we consider all extensions of the heap we must consider the possibility that the new heap stores a value $5$ at this location. Since the left-hand conjunct $\pointsto{x.f}{1}{\_}$ requires full access to $x.f$, no permission to this location can be left over when judging the implication on the right. The formula can be formally shown to be equivalent to $\pointsto{x.f}{1}{\_}*\pointsto{y.g}{1}{\_}$ according to the standard semantics of Definition \ref{defn:seplogicsemantics}, as follows:
\begin{proof}
 It suffices to show
\[
	{\PHeap,\Env \modelsSL (\pointsto{x.f}{1}{5}\ \imp\ \pointsto{y.g}{1}{\_})
       \land (\semSLE{x}{\Env},f)\notin dom(\PHeap)}
\Rightarrow
       \first{\PHeap[\semSLE{y}{\Env},f]} = \full
\]
We can prove this by contradiction. We assume $(\semSLE{x}{\Env},f)\notin dom(\PHeap)$
and $\first{\PHeap[\semSLE{y}{\Env},f]} \neq \full$, and consider the semantics of the implication:
\[
   \forall \PHeap'.(\;
    \PHeap' \perp \PHeap\  \land \
    \PHeap{*}\PHeap', \Env \modelsSL \pointsto{x.f}{1}{5} \;\; \Rightarrow\;\;
     \PHeap{*}\PHeap', \Env \modelsSL \pointsto{y.g}{1}{\_})
\]
By choosing $\PHeap'$ to be the heap containing $x.f$ with full permission and value 5, and with no other location in its domain, we deduce a contradiction, since $\PHeap{*}\PHeap', \Env \modelsSL \pointsto{x.f}{1}{5}$ does hold, while $\PHeap{*}\PHeap', \Env \modelsSL \pointsto{y.g}{1}{\_}$ does not.
\end{proof}

The second example formula listed above is actually a translation of the first into implicit dynamic frames: it is only true in states which have (full) permission to \emph{both} locations $x.f$ and $y.g$ in our semantics, and for the same reasons as the previous example. However, this assertion goes beyond the syntax of implicit dynamic frames typically supported by tools; we will discuss this later.
The third formula should mean: we have (full) permission to access $x.f$, and if its value is currently $5$ then we also have full access to $y.f$. This kind of assertion is already supported by the Chalice tool, and has exactly that intuitive meaning.

We need to decide which of the notions of state extension, as given in Definition \ref{defn:extensions}, we should use to define our semantics for implication. Since, in the traditional separation logic semantics, the two sides of an implication get judged in the whole resulting heap (after the state extension), the latter two ``disjoint'' variants from the definition are not appropriate. However, we still have the choice between considering locally-havoced or globally-havoced extensions. Let us consider the (slightly-simpler) option of using globally-havoced extensions. This leads us to the following candidate semantics for implication:
\[\begin{array}{l}
  \Heap, \Perm, \Env \modelsISL A_1 \imp A_2  \stackrel{?}{\iff}
\\\multicolumn{1}{@{\qquad}r}{\forall \Perm',\Heap'.(\;
   \Perm' {\perp} \Perm\  \land\
    \Heap' \agrees{\Perm} \Heap \ \land\
    \Heap', \Perm * \Perm', \Env \modelsISL  A_1
  \Rightarrow\;\;
     \Heap', \Perm * \Perm', \Env \modelsISL A_2)}
\end{array}\]
This definition gives the correct meaning to our first example assertion: since we judge the implication in a state in which we have no permission to $x.f$ (all such permission has to be given to the left-hand side of the $*$), we have to consider heaps $\Heap'$ in which $x.f$ has taken on arbitrary values. In particular, there are some such heaps in which $x.f$ has the value $5$, and this forces the requirement that we must also have full permission to $y.g$, just as in the traditional separation logic semantics. In fact, using locally-havoced extensions in our definition (i.e., changing the constraint on $\Heap'$ to be $\Heap' \agrees{\rds{\Perm}\cup\overline{\rds{\Perm'}}} \Heap$), would also give the right semantics; since we are still required to consider the possibility that we add on permission to $x.f$ in the extra permissions $\Perm'$, and in this case, it is allowed for $\Heap'$ to differ with $\Heap$ on $x.f$'s value. In fact, it is generally the case that the choice of locally-havoced or globally-havoced extensions makes no difference when we consider separation logic assertions. Exactly the same arguments (and resulting semantics) apply to the second of our example assertions.

However, the candidate definition above does not in general have the correct meaning for implicit dynamic frames. In particular, our third example formula does not have the correct semantics. The change of heap forces us to consider extensions in which we alter the value of $x.f$ in the heap, and thus our third example also becomes equivalent to $acc(x.f) * acc(y.f)$, since the value restriction for $x.f$ is made irrelevant by this potential change. To see exactly what we need here, we need to again consider carefully the meaning of $x.f$ in a pure assertion. When such a heap-dependent expression occurs on the left of an implication, its intended meaning depends on where in the assertion we find permission to the heap location. There are two important questions: does a permission to this location \emph{also} occur on the left of the implication (e.g., in our second example formula), and does a permission to this location occur elsewhere in the assertion? If neither occurs, i.e., there is no permission guaranteed to $x.f$ anywhere in the assertion, then the expression is meaningless; we will consider this case ill-formed, and so can give it any semantics. If a permission to $x.f$ occurs on the left of the implication in question (as in our second example formula), then it is possible that the value of the heap location is fixed as part of the heap extension, and therefore we should allow this value to change as part of the extension, as in our last candidate semantics above. In particular, if we judge the truth of the implication in a state in which we \emph{start off} with no permission to the location $x.f$, then it is just by adding the new permission required on the left of the implication that we can read from the location, and so we should consider that the value may be different from that in our original state. Finally, if a permission does \emph{not} occur on the left of the implication, but \emph{does} occur elsewhere in the assertion (as in our third example formula $\acc(x.f,1) * (x.f=5 \imp \acc(y.g,1))$), then we should \emph{not} allow the value of $x.f$ to change when judging the implication; its value must be determined by the permission outside of the implication, and so should not be allowed to change when judging it.

This analysis leads us to the conclusion that, in our candidate semantics above, we are considering too many extensions.  We need to only consider the \emph{minimal} extensions of the state to make the left of the implication true; in particular, we should only add permissions if the left-hand side of the implication explicitly requires them, and only allow values to change at those locations to which we added new permission. To this end, we provide a definition to capture the idea of a \emph{minimal permission extension}, which expresses that the extra permission we add on does not make any more locations readable than is necessary to make the assertion we are concerned with true:
\begin{definition}[Minimal Permission Extensions]\label{defn:minimalextensions}
Starting from a state $(\Heap,\Perm,\Env)$, we say that \emph{$\Perm'$ is a minimal permission extension of $(\Heap,\Perm,\Env)$ to satisfy $A$}, which we write as $\extend{\Heap}{\Perm}{\Env}{\Perm'}{A}$, as described by the following formula:
\[
\begin{array}{ll}
    \extend{\Heap}{\Perm}{\Env}{\Perm'}{A} &\iff
    \Heap, \Perm * \Perm', \Env \modelsISL A \;\; \land \;\;\\
&\quad\quad\quad\forall \Perm''\cdot \Perm''\subseteq\Perm' \land \rds{\Perm''} \subset \rds{\Perm'} \Rightarrow \Heap, \Perm*\Perm'', \Env \notmodelsISL A
\end{array}
\]
\end{definition}

\noindent We abstract over the precise permission values in this minimal extension (by focusing on which locations are readable, using the $\rds{}$ concept), in order to avoid imposing restrictions on the underlying permissions model (in particular, our definition does not depend on there being a greatest lower bound for the acceptable permission values to satisfy an assertion). That is, a minimal permission extension can add on more permission than the assertion really requires, so long as it does not increase the set of locations accessible in the permissions mask by more than necessary.

Using the concept of a minimal permission extension, along with the notion of locally-havoced extensions, we can finally define the semantics for implication which works for our general logic:
\[
\begin{array}{@{}l@{}}\Heap, \Perm, \Env \modelsISL A_1 \imp A_2  \iff\\
\forall (\Heap',\Perm * \Perm',\Env)\in\localExts{\Heap,\Perm,\Env}\cdot
    \extend{\Heap'}{\Perm}{\Env}{\Perm'}{A_1}
\Rightarrow\;\;
     \Heap', \Perm * \Perm', \Env \modelsISL A_2
\end{array}\]
This definition can be informally understood as follows: $A_1\imp A_2$ is true in a state if, for all minimal extensions (and corresponding havocs) of the state such that $A_1$ holds, $A_2$ must hold as well. The extension of the state is modelled by adding on the permissions $\Perm'$, and allowing the values of the heap to be modified in exactly the locations which become newly-readable by adding on these permissions. Furthermore, we insist on the permissions added being minimal in the locations which they make readable, while still satisfying $A_1$.

This definition correctly captures that we sometimes need to consider changing values in the heap when judging an implication, but only when it is permissions in the left-hand side of the implication that allow the reading of the locations. For example, using the definition above, the third example formula $\acc(x.f,1) * (x.f=5 \imp \acc(y.f,1))$ is true exactly in a state where we have (full) permission to $x.f$, and where, if the current value of $x.f$ is $5$, we also have (full) permission to $y.f$. On the other hand, the second formula $\acc(x.f,1) * (\acc(x.f,1) * x.f=5 \imp \acc(y.f,1))$ is true exactly when we have full permission to both $x.f$ and $y.f$ - this is because the implication is evaluated in a state with no permission to $x.f$, and when our semantics considers extending this state by just enough permission to make the left-hand side true, we have to allow for the possibility that this makes $x.f$ newly readable, and thus, makes it take on a new value.

Although we did not discuss the semantics of the magic wand connective ($\wand$) in the above, similar considerations lead us to also use the concept of minimal extensions in its definition. The formal definitions of our semantics come in the following subsection.
%
%

\subsection{Magic Wand Semantics}

We design our semantics for the ``magic wand'' connective $\wand$ along similar lines to that for implication. In particular, if one writes a pure assertion $A_1$ on the left of a wand formula $A_1\wand A_2$, then we do not want the semantics of the assertion to consider havocing heap locations that $A_1$ refers to. Put another way, for pure $A_1$ we would like the property that $A_1\wand A_2$ is equivalent to $A_1\imp A_2$, which, as we have already decided, should have a semantics which allows $A_1$ to refer to heap values in our original state, unless extra permission to those locations is added in $A_1$. This reasoning leads us to again employ our notion of minimal permission extension, but this time we complement it with locally-havoced \emph{disjoint} extensions:
\[\begin{array}{@{}l@{}} \Heap, \Perm, \Env \modelsISL A_1 \wand A_2  \iff\\
\forall (\Heap',\Perm',\Env)\in\localDisjExts{\Heap,\Perm,\Env}.
\extend{\Heap'}{\emptyset}{\Env}{\Perm'}{A_1}
\;\;\Rightarrow\;\;
     \Heap', \Perm * \Perm', \Env \modelsISL A_2)\end{array}
     \]

\subsection{Total Heaps Semantics for \texorpdfstring{\ESL{}}{TPL}}

We can now define our semantics for assertions. We make use of the concept of a minimal permission extension (Definition \ref{defn:minimalextensions}) to describe minimal extensions to the whole state when judging implications and magic wand assertions:
\begin{definition}[Total Heap Semantics for \SLE]\label{defn:seplogictotalsemantics}

We define validity of \SLE-assertions with respect to a specified total heap $\Heap$ and permission mask $\Perm$ recursively on the structure of the assertion:
\[
\begin{array}{@{}l}
  \begin{array}{@{}l@{\;}l}
    \Heap, \Perm, \Env \modelsISL \pointsto{E.f}{`p}{E'}
    &\iff
    \Perm[\semCE{E}{\Env,\Heap}, f] \geq `p \ \land\ \Heap[\semSLE{E}{\Env,\Heap}, f] = \semCE{E'}{\Env,\Heap}
    \\[0.5em]
    \Heap, \Perm, \Env \modelsISL A_1 * A_2  &\iff
    \exists \Perm_1, \Perm_2.(
   \Perm = \Perm_1{*}\Perm_2 \ \land \ \Heap, \Perm_1, \Env \modelsISL A_1
     \ \land\
     \Heap, \Perm_2, \Env \modelsISL A_2)

    \\[0.5em]
    \Heap, \Perm, \Env \modelsISL A_1 \wand A_2  &\iff
\forall (\Heap',\Perm',\Env)\in\localDisjExts{\Heap,\Perm,\Env}.
 \\
\multicolumn{2}{r}{
\extend{\Heap'}{\emptyset}{\Env}{\Perm'}{A_1}
\;\;\Rightarrow\;\;
     \Heap', \Perm * \Perm', \Env \modelsISL A_2}
         \\[0.5em]
    \Heap, \Perm, \Env \modelsISL A_1 \wedge A_2  &\iff
    \Heap, \Perm, \Env \modelsISL A_1 \;\; \land \;\;
    \Heap, \Perm, \Env \modelsISL A_2
    \\[0.5em]
    \Heap, \Perm, \Env \modelsISL A_1 \vee A_2  &\iff
    \Heap, \Perm, \Env \modelsISL A_1 \;\; \lor \;\;
    \Heap, \Perm, \Env \modelsISL A_2
    \\[0.5em]
    \Heap, \Perm, \Env \modelsISL A_1 \imp A_2  &\iff

\forall (\Heap',\Perm * \Perm',\Env)\in\localExts{\Heap,\Perm,\Env}.
\\
\multicolumn{2}{r}{
\extend{\Heap'}{\Perm}{\Env}{\Perm'}{A_1}
\;\;\Rightarrow\;\;
     \Heap', \Perm * \Perm', \Env \modelsISL A_2}
     \end{array}
         \\[0.5em]
              \begin{array}{@{}l@{\;}l}
    \Heap, \Perm, \Env \modelsISL \acc(E.f,`p)
    &\iff
    \Perm[\semCE{E}{\Env,\Heap}, f] \geq `p
    \\[0.5em]
    \Heap, \Perm, \Env \modelsISL E = E'
    &\iff
    \semCE{E}{\Env,\Heap} = \semCE{E'}{\Env,\Heap}
    \\[0.5em]
    \Heap, \Perm, \Env \modelsISL \exists x.\; A
    &\iff
    \exists v.(\Heap, \Perm, \Env[x \mapsto v] \modelsISL  A)
\end{array}
  \end{array}
\]\smallskip
\end{definition}

\noindent Note the similarity between the definitions for magic wand $\wand$ and logical implication $\imp$. This is because both cases involve heap extension in the partial heap semantics; in our total heap semantics we model heap extension by enabling the assignment of new arbitrary values to the part of the heap we have added permissions to.

Evaluation of \ESL{} expressions depends on a given environment and heap, and is defined by:
\[
\begin{array}{rclrclrclrcl}
  \semCE{x}{\Env,\Heap} & =& \Env(x) & \quad
  \semCE{n}{\Env,\Heap} & =& n &\quad
  \semCE{E.f}{\Env,\Heap} & =& \Heap[\semCE{E}{\Env,\Heap},f] &\quad
  \semCE{\Null}{\Env,\Heap}& =& \Null\\
\end{array}
\]

The meaning of separation logic expressions is preserved (and is independent of the heap), as the following lemma shows:
\begin{lemma}\label{lemma:expressions}
$\forall e,\Env,\Heap.\ \semCE{e}{\Env,\Heap} = \semSLE{e}{\Env}$
\end{lemma}
The main aims of the rest of this section are to show that our assertion semantics also preserves the original meaning of separation logic assertions.

\subsection{Strengthening and Weakening Results}
In this subsection, we present some of the technical properties which describe how our semantics behaves when we add and remove permissions, and when we extend states.

The following lemma shows the (intuitive) property that we can always discard superfluous permissions to reach a minimal permission extension:
\addprop{lemma}{Minimisation of Permission Masks}
{lemma:newlemmaone}
{\em If $\Heap,\Perm_1*\Perm_2,\Env\modelsISL A$ then $\exists \Perm_3\subseteq\Perm_2$ such that $\extend{\Heap}{\Perm_1}{\Env}{\Perm_3}{A}$.}
{}{
By complete (strong) induction on $\rds{\Perm_2}$ using subset ordering.

Case split on whether the following formula holds:
\[
\exists \Perm_4 \subseteq \Perm_2.\ \rds{\Perm_4} \subset \rds{\Perm_2} \land
                          \Heap, \Perm_1 * \Perm_4, \Env \modelsISL A
\]
\as{Assume first that} it does. By induction, we know $
\exists \Perm_5 \subseteq \Perm_4. \ \extend{\Heap}{\Perm_1}{\Env}{\Perm_5}{A}
$, and by transitivity of \as{$\subseteq$, we have} $\Perm_5 \subseteq \Perm_2$, as required.

\as{Secondly, assume that the} formula does not hold. Therefore,
\[
\forall \Perm_4 \subseteq \Perm_2.\ \rds{\Perm_4} \subset \rds{\Perm_2}
	\Rightarrow
        \Heap, \Perm_1 * \Perm_4, \Env \notmodelsISL A
\]
Hence,
$
\extend{\Heap}{\Perm_1}{\Env}{\Perm_2}{A}
$. 
\as{We} can prove the obligation by picking $\Perm_3 = \Perm_2$.\qed\smallskip
}

\begin{definition}[\Weakenable{} and Intuitionistic formulas]
We define a formula $A$ to be \emph{\weakenable} if and only if
\[
\forall \Heap, \Perm, \Env, \Perm'. (\Perm\subseteq\Perm' \;\;\land\;\; \Heap,\Perm,\Env\modelsISL A \quad\Rightarrow\quad \Heap,\Perm',\Env\modelsISL A)
\]
We define a formula $A$ to be \emph{intuitionistic} if and only if
\[
\forall \Heap, \Perm, \Env. (\Heap,\Perm,\Env\modelsISL A \quad\Rightarrow\quad \forall (\Heap',\Perm',\Env)\in\globalExts{\Heap,\Perm,\Env}.\  \Heap',\Perm',\Env\modelsISL A)
\]
\end{definition}

\addprop{lemma}{}{lemma:fuzzyedge}{\em
If  $A$ is \weakenable{}, $\extend{\Heap}{\Perm}{\Env}{\Perm'}{A}$,
$\Perm' \subseteq \Perm''$
and $\rds{\Perm'} = \rds{\Perm''}$, then
$\extend{\Heap}{\Perm}{\Env}{\Perm''}{A}$
}{%
}{
As $A$ is \weakenable{}, we know:
\[
\Heap,\Perm * \Perm'',\Env \modelsISL A
\]
If $\Perm''$ is not minimal, then neither was $\Perm'$, which is a contradiction.  So $\Perm''$ must also be a minimal extension.\qed\smallskip
}

The following technical lemma shows that any necessary permissions in a state are still necessary in a state with fewer permissions, provided the assertion we are considering is closed under permission extension:
\addprop{lemma}{Minimal Permission Extensions Closed}{lemma:newlemmafive}
{\em
If $\extend{\Heap}{\Perm_1*\Perm_2}{\Env}{\Perm_3}{A}$ and  $A$ is \weakenable{}, then $\exists \Perm_4\subseteq\Perm_2$ and $\extend{\Heap}{\Perm_1}{\Env}{\Perm_4*\Perm_3}{A}$.
}{%
}
{
First, we prove a more general result, and apply Lemma~\ref{lemma:fuzzyedge} to get the required result.  We prove, that:
\begin{quote}
If we know that $A$ is \weakenable{}, then,
if we also have $\extend{\Heap}{\Perm_1*\Perm_2}{\Env}{\Perm_3}{A}$ then $\exists \Perm_4\subseteq\Perm_2$ and $\Perm_5\subseteq\Perm_3$ such that $\rds{\Perm_5} = \rds{\Perm_3}$ and $\extend{\Heap}{\Perm_1}{\Env}{\Perm_4*\Perm_5}{A}$.
\end{quote}

\noindent We know $\Heap, \Perm_1 * \Perm_2 * \Perm_3, \Env \modelsISL A$, and by Lemma~\ref{lemma:newlemmaone}, we have \as{that} there exists $\Perm_6$ such that $\extend{\Heap}{ \Perm_1}{\Env}{\Perm_6}{A}$ and $\Perm_6 \subseteq \Perm_2 * \Perm_3$.  Choose $\Perm_4$ and $\Perm_5$ such that $\Perm_5 = \Perm_3 \sqcap \Perm_6$, and $\Perm_4 = \Perm_6 - (\Perm_3 \sqcap \Perm_6)$.

We need to show $\Perm_4 \subseteq \Perm_2$ and $\rds{\Perm_5} = \rds{\Perm_3}$.  From our assumptions, we know that
$\Perm_4 * (\Perm_3 \sqcap \Perm_6) = \Perm_6$, and thus $\Perm_4 \subseteq \Perm_6$,  and $\Perm_6 \subseteq \Perm_2 * \Perm_3$. We split into two cases:
\begin{desCription}
\item\noindent{\hskip-12 pt($\Perm_4 \not{\subseteq} \Perm_2$):}\  Therefore, there exists $(\iota,f)$ such that $\Perm_4[\iota,f] > \Perm_2[\iota,f]$.  Therefore,
$\Perm_6[\iota,f] > \Perm_2[\iota,f]$, but by assumption, we know
$\Perm_6[\iota,f] \leq \Perm_2[\iota,f] + \Perm_3[\iota,f]$. Consider two sub-cases:
\begin{desCription}
\item\noindent{\hskip-12 pt({$\Perm_6[\iota,f]
    <\Perm_3[\iota,f]$}):}\
Therefore, as $\Perm_4 = \Perm_6 - (\Perm_3 \sqcap \Perm_6)$ we know $\Perm_4[\iota,f] = \none$ contradicting $\Perm_4[\iota,f] > \Perm_2[\iota,f]$.

\item\noindent{\hskip-12 pt({$\Perm_3[\iota,f] \leq
    \Perm_6[\iota,f]$}):}\
Therefore, as $\Perm_4 * (\Perm_3 \sqcap \Perm_6) = \Perm_6$ we know $\Perm_4[\iota,f] + \Perm_3[\iota,f] = \Perm_6[\iota,f]$.
As $\Perm_6 \subseteq \Perm_2 * \Perm_3$, we know
$\Perm_4[\iota,f] + \Perm_3[\iota,f] \leq \Perm_2[\iota,f] + \Perm_3[\iota,f]$, hence $\Perm_4[\iota,f] \leq \Perm_2[\iota,f]$, contradicting assumption.
\end{desCription}

\item\noindent{\hskip-12 pt($\Perm_4 \subseteq \Perm_2$):}\
We split into two sub-cases:
\begin{desCription}
\item\noindent{\hskip-12 pt($\rds{\Perm_5} = \rds{\Perm_3}$):}\ The result follows directly.

\item\noindent{\hskip-12 pt($\rds{\Perm_5} \subset \rds{\Perm_3}$):}\  For this case, using the assumption that $A$ is \weakenable{}, we get
$\Heap, \Perm_1 * \Perm_2 * \Perm_5, \Env \modelsISL A$, which is in   contradiction with the assumption that $\extend{\Heap}{\Perm_1*\Perm_2}{\Env}{\Perm_3}{A}$.\qed
\end{desCription}
\end{desCription}
}

The validity of assertions in this semantics is closed under permission extension.
\addprop{proposition}{}{prop:weakening}
{\em All formulas $A$ are \weakenable{}.}
{\em By induction on structure of formula.}
{
By induction on $A$.  By $\Perm\subseteq\Perm'$ we can assume there exists  $\Perm''$ such that $\Perm * \Perm'' = \Perm'$.
\begin{desCription}

\item\noindent{\hskip-12 pt($A \equiv \acc(E.f,\pi)$):}\

\[\begin{array}{rclr}
\multicolumn{3}{@{}l}{\modelsSLE{\Heap}{\Perm}{\Env}{A}}\\

 & \Leftrightarrow & \Perm[\semSLE{E}{\Heap,\sigma},f] \geq \pi  & \textit{(by defn)}\\

 & \Rightarrow & \Perm'[\semSLE{E}{\Heap,\sigma},f] \geq \pi  & \textit{(as $ \Perm \subseteq \Perm'$)}\\

 & \Leftrightarrow & \modelsSLE{\Heap}{\Perm'}{\Env}{A}
\end{array}
\]
\item\noindent{\hskip-12 pt($A \equiv E=E'$):}\
\[\begin{array}{rclr}
\multicolumn{3}{@{}l}{\modelsSLE{\Heap}{\Perm}{\Env}{A}}\\

 & \Leftrightarrow & \semSLE{E}{\Heap,\sigma} = \semSLE{E'}{\Heap,\sigma}
  & \textit{(by defn)}\\

 & \Leftrightarrow & \modelsSLE{\Heap}{\Perm'}{\Env}{A}
\end{array}
\]
\item\noindent{\hskip-12 pt($A\equiv E.f\mapsto E'$):}\ Follows using a combination of arguments from previous two cases.

\item\noindent{\hskip-12 pt($A \equiv A_1 * A_2$):}\

\[\begin{array}{rclr}
\multicolumn{3}{@{}l}{\modelsSLE{\Heap}{\Perm}{\Env}{A}}\\

 & \Rightarrow & \exists \Perm_1, \Perm_2.\ \Perm_1 * \Perm_2 = \Perm \land \modelsISL{\Heap}{\Perm_1}{\Env}{A_1}
\land
\modelsISL{\Heap}{\Perm_2}{\Env}{A_2}
\end{array}\smallskip
\]

\noindent We introduce $\Perm_1$ and $\Perm_2$, and define $\Perm_3 = \Perm_2 * \Perm''$. By induction, we know
\[\begin{array}{rclr}
 &\Rightarrow& \Perm_1 * \Perm_2 = \Perm \land \modelsISL{\Heap}{\Perm_1}{\Env}{A_1}
\land
\modelsISL{\Heap}{\Perm_3}{\Env}{A_2}
  & \\
 & \Rightarrow & \Perm_1 * \Perm_3 = \Perm' \land \modelsISL{\Heap}{\Perm_1}{\Env}{A_1}
\land
\modelsISL{\Heap}{\Perm_3}{\Env}{A_2}
\\
 & \Rightarrow & \modelsISL{\Heap}{\Perm'}{\Env}{A}
  & \textit{(by defn)}\\

\end{array}
\]
\item\noindent{\hskip-12 pt($A\equiv A_1 \land A_2$, $A\equiv A_1 \lor A_2$):}\  Trivially by induction.

\item\noindent{\hskip-12 pt($A\equiv A_1 \wand A_2$):}\
By unfolding \as{the} obligation, we can assume
\[
\begin{array}{l}
\Heap,\Perm,\Env \modelsISL A_1 \wand A_2\\
\Perm_1 \perp \Perm'\\
\Heap_1 \agrees{\rds{\Perm'} \cup \overline{\rds{\Perm_1}}} \Heap\\
\extend{\Heap_1}{\emptyset}{\Env}{\Perm_1}{A_1}
\end{array}
\]
and must prove
\[
\Heap_1,\Perm' * \Perm_1, \Env \modelsISL A_2
\]
By assumptions, we know $\Perm_1 \perp \Perm$ and $\Heap_1 \agrees{\rds{\Perm} \cup \overline{\rds{\Perm_1}}} \Heap$ as $\Perm$ is smaller than $\Perm'$.  Therefore, using $\wand$ assumption, we have
\[
\Heap_1,\Perm * \Perm_1, \Env \modelsISL A_2
\]
By inductive hypothesis, we know
\[
\Heap_1,\Perm' * \Perm_1, \Env \modelsISL A_2
\]
as required.

\item\noindent{\hskip-12 pt($A \equiv A_1 \imp A_2$):}\
By unfolding the obligation, we can assume
\[
\begin{array}{l}
\Heap,\Perm,\Env \modelsISL A_1 \imp A_2\\
\Perm_2 \perp \Perm'\\
\Heap_2 \agrees{\rds{\Perm'} \cup \overline{\rds{\Perm_2}}} \Heap\\
\extend{\Heap_2}{\Perm'}{\Env}{\Perm_2}{A_1}
\end{array}
\]
and must prove
\[
\Heap_2,\Perm' * \Perm_2, \Env \modelsISL A_2
\]
By Lemma~\ref{lemma:newlemmafive}, and as $A_1$ is closed under permission extension by inductive hypothesis, we have that there exists a $\Perm_1$ such that
\[
\begin{array}{l}
\Perm_1 \subseteq \Perm'' \\
\extend{\Heap_2}{\Perm}{\Env}{\Perm_1 * \Perm_2}{A_1}\\
\end{array}
\]
As $\Perm \subseteq \Perm'$, we can show
\[
\begin{array}{rcll}
&&\Heap_2 \agrees{\rds{\Perm'} \cup \overline{\rds{\Perm_2}}} \Heap\\
&\Rightarrow &
\Heap_2 \agrees{\rds{\Perm} \cup \overline{\rds{\Perm_2}}} \Heap &
  \textit{(as $\Perm \subseteq \Perm'$)}\\
&\Rightarrow &
\Heap_2 \agrees{\rds{\Perm} \cup \overline{\rds{\Perm_2 * \Perm_1}}} \Heap &
  \textit{(as $\Perm_2 \subseteq \Perm_1 * \Perm_2$)}\\
\end{array}
\]
Therefore by $\imp$ assumption, we have
\[
\Heap_2, \Perm * \Perm_1 * \Perm_2, \Env \modelsISL A_2
\]
As $\Perm_1 \subseteq \Perm''$, by inductive hypothesis we have
\[
\Heap_2, \Perm * \Perm'' * \Perm_2, \Env \modelsISL A_2
\]
as required.

\item\noindent{\hskip-12 pt($A \equiv \exists x.\; A'$):}\
\[\begin{array}{l}
\Heap,\Perm,\Env \modelsISL \exists x.\; A'\\
\Rightarrow \exists v.\; \Heap,\Perm,\Env[x \mapsto v] \modelsISL A'\\
\Rightarrow \exists v.\; \Heap,\Perm',\Env[x \mapsto v] \modelsISL A'\textit{ (by inductive hypothesis)}\\
\Rightarrow \Heap,\Perm',\Env \modelsISL \exists x.\;A'
\end{array}\]
\end{desCription}

}

In our later results, we sometimes need to be able to define when a minimal permission extension in one state corresponds with a minimal permission extension in another. In particular, we want to be able to show that, under certain conditions, the notion of what is a minimal extension in a current state is robust to interference. An important property which helps us here, is to be able to express that the truth of an assertion is stable in all extensions of a particular state. That is, even if the assertion does not hold in the current state, in all extensions which do satisfy the assertion, its truth will be stable. This can be expressed by the following definitions.
\begin{definition}

An assertion $A$ is \emph{extension framed} in a state $(\Heap,\Perm,\Env)$ if and only if, $A$ is stable in all (globally-havoced) extensions, i.e.,
\[
\eframed{\Heap,\Perm,\Env,A}
\iff
\stable{\globalExts{\Heap,\Perm,\Env}\cap\semEl{A}}
\]
We also define the set of states in which $A$ is extension framed:
\[
\eframed{A} = \{ (\Heap,\Perm,\Env) \mid \eframed{\Heap,\Perm,\Env,A} \}
\]
An assertion $A$ is \emph{disjoint extension framed} in the current memory if and only if, in all (globally-havoced) disjoint extensions, the truth of $A$ is stable with the permissions held originally:
\[
\deframed{\Heap,\Perm,\Env,A}
\iff
\stablewith{\Perm}{\globalDisjExts{\Heap,\Perm,\Env} \cap \semEl{A} \}}
\]
We also define the set of states in which $A$ is disjoint-extension framed:
\[
\deframed{A} = \{ (\Heap,\Perm,\Env) \mid \deframed{\Heap,\Perm,\Env,A} \}
\]
\end{definition}\vspace{6 pt}%

\noindent Note that we use the globally-havoced notions of state extension (cf. Definition \ref{defn:extensions}), rather than locally-havoced. The reason for this is that, we need this criterion on assertions to be preserved under interference, at various points in our proofs. The following lemma characterises the essential properties that we require of these definitions.

\addprop{lemma}{}{lem:extfrmuses}
{
\mbox{}\em
\begin{enumerate}[\em(1)]
\item If $(\Heap,\Perm,\sigma) \in \eframed{A}$, then:
  \begin{enumerate}[\em(a)]
    \item if $\Heap' \in \interfere{\Heap}{\Perm}$, then $(\Heap',\Perm,\sigma) \in \eframed{A}$.
    \item if $\Perm' \perp \Perm$  and $\Heap' \in \interfere{\Heap}{\Perm * \Perm' }$ and
                $\Heap,\Perm*\Perm',\sigma \modelsISL A$, then\\  $\Heap', \Perm*\Perm', \sigma \modelsISL A$.
  \end{enumerate}
\item If $(\Heap,\Perm,\sigma) \in \deframed{A}$, then:
  \begin{enumerate}[\em(a)]
    \item if $\Heap' \in \interfere{\Heap}{\Perm}$, then $(\Heap',\Perm,\sigma) \in \deframed{A}$.
    \item if $\Perm' \perp \Perm$  and $\Heap' \in \interfere{\Heap}{\Perm * \Perm' }$ and
                $\Heap,\Perm',\sigma \modelsISL A$, then\\  $\Heap', \Perm', \sigma \modelsISL A$.
  \end{enumerate}
\end{enumerate}
}{}{
\mbox{}
\begin{enumerate}[(1)]
   \item
	\begin{enumerate}[(a)]
	  \item
		This follows directly, since if we have $\Heap' \in \interfere{\Heap}{\Perm}$,
		 then it follows that $\globalExts{\Heap,\Perm,\Env} = \globalExts{\Heap',\Perm,\Env}$.
	  \item
		 By assumptions, we know that $(\Heap, \Perm * \Perm', \Env) \in \globalExts{\Heap,\Perm,\Env}$, and also that $\Heap' \in \interfere{\Heap}{\Perm * \Perm'}$.  As $ \globalExts{\Heap,\Perm,\Env} \cap \semEl{A}$ is stable, we know
$\Heap',\Perm * \Perm',\Env \modelsISL A$ as required.
	\end{enumerate}
   \item
	\begin{enumerate}[(a)]
	   \item
		This follows directly, since if $\Heap' \in \interfere{\Heap}{\Perm}$,
		 then $\globalDisjExts{\Heap,\Perm,\Env} = \globalDisjExts{\Heap',\Perm,\Env}$.
	   \item
		 By assumptions, we know that $(\Heap, \Perm', \Env) \in \globalDisjExts{\Heap,\Perm,\Env}$, and also that $\Heap' \in \interfere{\Heap}{\Perm*\Perm'}$.  As $ \globalDisjExts{\Heap,\Perm,\Env} \cap \semEl{A}$ is stable with $\Perm$, we know
$\Heap',\Perm',\Env \modelsISL A$ as required.

	\end{enumerate}
\end{enumerate}
}


\noindent Note that using $\localExts{\Heap,\Perm,\sigma}$, or $\localDisjExts{\Heap,\Perm,\sigma}$, rather than $\globalExts{\Heap,\Perm,\sigma}$, or $\globalDisjExts{\Heap,\Perm,\sigma}$, invalidates part (1.a), or part (2.a), respectively.

The following technical lemma provides sufficient conditions for a minimal permission extension to be robust to interference in the rest of the heap:
\addprop{lemma}{Preservation of Minimal Extensions}{lemma:newlemmatwo}
{\hfill\em \begin{enumerate}[\em(1)]
  \item For all $(\Heap_1, \Perm_1, \Env) \in \eframed{A},$
  \[\begin{array}{r}
  \forall\Perm_2\perp\Perm_1,\forall \Heap_2\agrees{\Perm_1*\Perm_2}\Heap_1. (\extend{\Heap_1}{\Perm_1}{\Env}{\Perm_2}{A}\qquad\mbox{}\\ \   \Rightarrow\ \ \extend{\Heap_2}{\Perm_1}{\Env}{\Perm_2}{A})
  \end{array}\]
  \item For all $(\Heap_1, \Perm_1, \Env) \in \deframed{A},$
  \[\begin{array}{r}
  \forall\Perm_2\perp\Perm_1,\forall \Heap_2\agrees{\Perm_1*\Perm_2}\Heap_1. (\extend{\Heap_1}{\emptyset}{\Env}{\Perm_2}{A}\qquad\mbox{}\\ \   \Rightarrow\ \ \extend{\Heap_2}{\emptyset}{\Env}{\Perm_2}{A})
  \end{array}\]
  \item If $A$ is self-framing and $\Perm_2\perp\Perm_1$ and $\extend{\Heap_1}{\Perm_1}{\Env}{\Perm_2}{A}$ and   $\Heap_2{\agrees{\Perm_1*\Perm_2}}\Heap_1$, then $\extend{\Heap_2}{\Perm_1}{\Env}{\Perm_2}{A}$
  \end{enumerate}
}{
}{\mbox{}
\begin{enumerate}[(1)]
\item
Assume $\Perm_2 \perp \Perm_1$, $\Heap_2 \agrees{\Perm_1 * \Perm_2} \Heap_1$ and $\extend{\Heap_1}{\Perm_1}{\Env}{\Perm_2}{A}$. By $(\Heap_1,\Perm_1,\Env) \in \eframed{A}$
we have $\Heap_2,\Perm_1* \Perm_2,\Env \modelsISL A$, and thus we are left to prove:
\[
\forall \Perm_3 \subseteq \Perm_2.\ \rds{\Perm_3} \subset \rds{\Perm_2} \Rightarrow
\Heap_2, \Perm_1 * \Perm_3, \Env \notmodelsISL A
\]
We assume $\Perm_3 \subseteq \Perm_2$, $\rds{\Perm_3} \subset \rds{\Perm_2}$ and $\Heap_2, \Perm_1 * \Perm_3, \Env \modelsISL A$ and seek a contradiction. By Lemma~\ref{lem:extfrmuses}(1)(b) we can prove $\Heap_1, \Perm_1 * \Perm_3, \Env \modelsISL A$, and thus using $\extend{\Heap_1}{\Perm_1}{\Env}{\Perm_2}{A}$ we deduce a contradiction.

\item
Assume $\Perm_2 \perp \Perm_1$, $\Heap_2 \agrees{\Perm_1 * \Perm_2} \Heap_1$ and $\extend{\Heap_1}{\emptyset}{\Env}{\Perm_2}{A}$. By $(\Heap_1,\Perm_1,\Env) \in \deframed{A}$
we have $\Heap_2, \Perm_2,\Env \modelsISL A$, and thus we are left to prove:
\[
\forall \Perm_3 \subseteq \Perm_2.\ \rds{\Perm_3} \subset \rds{\Perm_2} \Rightarrow
\Heap_2, \Perm_3, \Env \notmodelsISL A
\]
We assume $\Perm_3 \subseteq \Perm_2$, $\rds{\Perm_3} \subset \rds{\Perm_2}$ and $\Heap_2, \Perm_3, \Env \modelsISL A$ and seek a contradiction. By Lemma~\ref{lem:extfrmuses}(2)(b) we can prove $\Heap_1, \Perm_3, \Env \modelsISL A$, and thus using $\extend{\Heap_1}{\emptyset}{\Env}{\Perm_2}{A}$ we deduce a contradiction.

\item Since $A$ is self-framing, we know $\Heap_2, \Perm_1*\Perm_2, \Env \modelsISL A$, and thus we are left to prove:
\[
\forall \Perm_3 \subseteq \Perm_2.\ \rds{\Perm_3} \subset \rds{\Perm_2} \Rightarrow
\Heap_2, \Perm_1*\Perm_3, \Env \notmodelsISL A
\]
We assume $\Perm_3 \subseteq \Perm_2$, $\rds{\Perm_3} \subset \rds{\Perm_2}$ and $\Heap_2, \Perm_1*\Perm_3, \Env \modelsISL A$ and seek a contradiction.  By self-framing assumption, we know
$\Heap_1, \Perm_1*\Perm_3, \Env \modelsISL A$, but this contradicts initial minimality assumption.
\end{enumerate}
}

We will make use of these technical lemmas in the next subsections, in order to characterise properties of our general semantics.

\subsection{Correspondence with Separation Logic Semantics}

In this subsection, we examine the correspondence between the semantics which our definition implies for the separation logic fragment of our logic, and the traditional semantics of separation logic.
In order to precisely characterise the laws which hold of the logic, we require a notion of semantic entailment.
\begin{definition}[Semantic Entailment, Validity and Equivalence]
A \SLE{} assertion $A$ is \emph{semantically valid} (written $\valid A$) if it holds in all situations; i.e.,
\[\valid A \Leftrightarrow \ \forall \Heap, \Perm, \Env.\ \Heap, \Perm, \Env \modelsISL A\]%
Given \SLE{} assertions $A_1$ and $A_2$, we say that \emph{$A_1$ semantically entails $A_2$} (and write $A_1\entails A_2$) if and only if $A_2$ holds whenever $A_1$ does; i.e.,
\[A_1\entails A_2 \ \Leftrightarrow\  \semEl{A_1} \subseteq \semEl{A_2} \]
Given \SLE{} assertions $A_1$ and $A_2$, we say that \emph{$A_1$ is equivalent to $A_2$} (and write $A_1\isequiv A_2$) if and only if $A_1\entails A_2$ and $A_2\entails A_1$.\smallskip
\end{definition}

For pure assertions, our (rather complex) definition of implication can be simplified to a simple boolean evaluation of the conditional:
\begin{lemma}[Pure Assertions are Boolean Conditionals]\label{lemma:boolean}
If $A_1$ is pure, then:
\[
\Heap, \Perm, \Env \modelsISL A_1 \imp A_2
\ \ \iff\ \
(\Heap, \Perm, \Env \modelsISL A_1\ \ \Rightarrow\ \ \Heap,\Perm,\Env \modelsISL A_2)
\]
\begin{proof}
We first observe that if $A_1$ is pure and $\extend{\Heap}{\Perm}{\Env}{\Perm'}{A_1}$, then
$\Perm' = \emptyset$.  Simplifying the semantic definition of $\imp$ using the $\emptyset$ gives the required semantics.
\end{proof}
\end{lemma}
Note that this property was not true in the semantics of the precursor paper~\cite{ParkinsonSummers'11}, and prevents the former work from correctly modelling Chalice's implication.

The following lemma also shows how our definition of semantics for implication and the magic wand can be simplified if we know that the immediate subformulas are self-framing assertions (in this case, we do not encounter the technical difficulties which led us to employ \emph{minimal} extensions; cf. Section \ref{sect:implication}):

\addprop{lemma}{Simplified Semantics for Self-Framing Conditionals}{lemma:newlemmathree}
{\hfill\em\begin{enumerate}[\em(1)]
\item If $A_1$ and $A_2$ are both self-framing, then:
\begin{enumerate}[\em(a)]
\item $\modelsSLE{\Heap}{\Perm}{\Env}{A_1\imp A_2}$ if and only if:
\[
\forall (\Heap',\Perm',\Env)\in\localExts{\Heap,\Perm,\Env}.\;\;(\Heap', \Perm', \Env \modelsISL A_1 \;\;\Rightarrow\;\;     \Heap', \Perm', \Env \modelsISL A_2)
\]
\item $\modelsSLE{\Heap}{\Perm}{\Env}{A_1\imp A_2}$ if and only if:
\[
\forall (\Heap',\Perm',\Env)\in\globalExts{\Heap,\Perm,\Env}.\;\;(\Heap', \Perm', \Env \modelsISL A_1 \;\;\Rightarrow\;\;     \Heap', \Perm', \Env \modelsISL A_2)
\]
\end{enumerate}
\item If $A_1$ and $A_2$ are both self-framing, then:
\begin{enumerate}[\em(a)]
\item $\modelsSLE{\Heap}{\Perm}{\Env}{A_1\wand A_2}$ if and only if:
\[
\forall (\Heap',\Perm',\Env)\in\localDisjExts{\Heap,\Perm,\Env}.\;\;(\Heap', \Perm', \Env \modelsISL A_1 \;\;\Rightarrow\;\;     \Heap', \Perm * \Perm', \Env \modelsISL A_2)
\]
\item $\modelsSLE{\Heap}{\Perm}{\Env}{A_1\wand A_2}$ if and only if:
\[
\forall (\Heap',\Perm',\Env)\in\globalDisjExts{\Heap,\Perm,\Env}.\;\;(\Heap', \Perm', \Env \modelsISL A_1 \;\;\Rightarrow\;\;     \Heap', \Perm * \Perm', \Env \modelsISL A_2)
\]
\end{enumerate}
\end{enumerate}
}{%
}{
\hfill
\begin{enumerate}[(1)]
\item\label{myi}
\begin{enumerate}[(a)]
\item We need to show that:
\[\qquad\begin{array}{l}
\left(
\begin{array}{r}
\forall \Perm_1, \Heap_1.(\;(\Heap_1,\Perm * \Perm_1,\Env) \in \localExts{\Heap,\Perm,\Env}  \land\
\extend{\Heap_1}{\Perm}{\Env}{\Perm_1}{A_1}\qquad\mbox{}\\ \Rightarrow \Heap_1, \Perm * \Perm_1, \Env \modelsISL A_2)
\end{array}\right)\\
\Leftrightarrow\\
\left(
\begin{array}{r}
\forall \Perm_2, \Heap_2.(\;(\Heap_2,\Perm * \Perm_2,\Env) \in \localExts{\Heap,\Perm,\Env}  \land\
\Heap_2, \Perm * \Perm_2, \Env \modelsISL A_1\qquad\mbox{}\\ \Rightarrow \Heap_2, \Perm * \Perm_2, \Env \modelsISL A_2)\end{array}\right)\end{array}
\]
The right-to-left direction is easy, since the left-hand formula requires that we check the implication in strictly fewer states (only those which are obtained via minimal extensions).
For the left-to-right direction, assume that for some arbitrary $\Perm_2,\Heap_2$ we have $\Perm_2\perp\Perm$ and $\Heap_2\agrees{\rds{\Perm}\cup\overline{\rds{\Perm_2}}}\Heap$ and $\Heap_2,\Perm*\Perm_2,\Env\modelsISL A_1$. Then we need to show that: $\modelsISL{\Heap_2}{\Perm*\Perm_2}{\Env}{A_2}$.
By Lemma \ref{lemma:newlemmaone}, there exists $\Perm_3\subseteq\Perm_2$ such that $\extend{\Heap_2}{\Perm}{\Env}{\Perm_3}{A_1}$. Define $\Heap_3 = \condheap{\Perm*\Perm_3}{\Heap_2}{\Heap}$. Then, by construction, $\Heap_3\agrees{\Perm*\Perm_3}\Heap_2$ and $\Heap_3\agrees{\rds{\Perm}\cup\overline{\rds{\Perm_3}}}\Heap$.

By Lemma \ref{lemma:newlemmatwo}~(3), since $A_1$ is self-framing, we have $\extend{\Heap_3}{\Perm}{\Env}{\Perm_3}{A_1}$. Now, using the assumption from the left-hand-side of our overall goal, choose $\Perm_1 = \Perm_3$ and $\Heap_1 = \Heap_3$, and we obtain $\modelsISL{\Heap_3}{\Perm*\Perm_3}{\Env}{A_2}$. Since $A_2$ is self-framing, we have $\modelsISL{\Heap_2}{\Perm*\Perm_3}{\Env}{A_2}$. Then, by Proposition \ref{prop:weakening}, we obtain $\modelsISL{\Heap_2}{\Perm*\Perm_2}{\Env}{A_2}$ as required.
\item\label{mypartii} By the previous part, it suffices to show that :
\[
\begin{array}{c}
\forall (\Heap',\Perm',\Env)\in\localExts{\Heap,\Perm,\Env}.\;\;(\Heap', \Perm', \Env \modelsISL A_1 \;\;\Rightarrow\;\;     \Heap', \Perm', \Env \modelsISL A_2)\\
\Leftrightarrow\\
\forall (\Heap',\Perm',\Env)\in\globalExts{\Heap,\Perm,\Env}.\;\;(\Heap', \Perm', \Env \modelsISL A_1 \;\;\Rightarrow\;\;     \Heap', \Perm', \Env \modelsISL A_2)\\
\end{array}
\]
The $(\Leftarrow)$ direction is immediate, since $\localExts{\Heap,\Perm,\Env}\subseteq\globalExts{\Heap,\Perm,\Env}$. To show the $(\Rightarrow)$ direction, we assume the former formula, and suppose that we have some $(\Heap',\Perm',\Env)\in\globalExts{\Heap,\Perm,\Env}$ such that $\Heap', \Perm', \Env \modelsISL A_1$ holds. Define $\Heap'' = \condheap{\Perm'}{\Heap'}{\Heap}$. By construction, $(\Heap'',\Perm',\Env)\in\localExts{\Heap,\Perm,\Env}$ and $\Heap''\agrees{\Perm'}\Heap'$. Since $A_1$ is self-framing, we conclude that $\Heap', \Perm', \Env \modelsISL A_1$ holds. Therefore, by the assumed formula, we can conclude that $\Heap'',\Perm',\Env \modelsISL A_2$ is true. Since $A_2$ is self-framing, we conclude $\Heap',\Perm',\Env \modelsISL A_2$ as required.

\end{enumerate}
\item
\begin{enumerate}
\item We need to show that:
\[ \begin{array}{l}
\left(
\begin{array}{r}
\forall \Perm_1, \Heap_1.(\;\Perm_1 {\perp} \Perm\  \land    \Heap_1 \agrees{\rds{\Perm}\cup\overline{\rds{\Perm_1}}} \Heap\  \land\
\extend{\Heap_1}{\emptyset}{\Env}{\Perm_1}{A_1}\qquad\mbox{}\\ \Rightarrow \Heap_1, \Perm * \Perm_1, \Env \modelsISL A_2)\end{array}\right)\\
\Leftrightarrow\\
\left(
\begin{array}{r}
\forall \Perm_2, \Heap_2.(\;\Perm_2 {\perp} \Perm\  \land    \Heap_2 \agrees{\rds{\Perm}\cup\overline{\rds{\Perm_2}}} \Heap\  \land\
\Heap_2, \Perm_2, \Env \modelsISL A_1\qquad\mbox{}\\
 \Rightarrow \Heap_2, \Perm * \Perm_2, \Env \modelsISL A_2)\end{array}\right)\end{array}
\]
The right-to-left direction is easy, since the left-hand formula requires that we check the implication in strictly fewer states (only those which are obtained via minimal extensions).
For the left-to-right direction, assume that for some arbitrary $\Perm_2,\Heap_2$ we have $\Perm_2\perp\Perm$ and $\Heap_2\agrees{\rds{\Perm}\cup\overline{\rds{\Perm_2}}}\Heap$ and $\Heap_2,\Perm_2,\Env\modelsISL A_1$. Then we need to show that: $\modelsISL{\Heap_2}{\Perm*\Perm_2}{\Env}{A_2}$.
By Lemma \ref{lemma:newlemmaone}, there exists $\Perm_3\subseteq\Perm_2$ such that $\extend{\Heap_2}{\emptyset}{\Env}{\Perm_3}{A_1}$. Define $\Heap_3 = \condheap{\Perm*\Perm_3}{\Heap_2}{\Heap}$. Then, by construction, $\Heap_3\agrees{\Perm*\Perm_3}\Heap_2$ and $\Heap_3\agrees{\rds{\Perm}\cup\overline{\rds{\Perm_3}}}\Heap$.

By Lemma \ref{lemma:newlemmatwo}~(3), since $A_1$ is self-framing, we have $\extend{\Heap_3}{\emptyset}{\Env}{\Perm_3}{A_1}$. Now, using the assumption from the left-hand-side of our overall goal, choose $\Perm_1 = \Perm_3$ and $\Heap_1 = \Heap_3$, and we obtain $\modelsISL{\Heap_3}{\Perm*\Perm_3}{\Env}{A_2}$. Since $A_2$ is self-framing, we have $\modelsISL{\Heap_2}{\Perm*\Perm_3}{\Env}{A_2}$. Then, by Proposition \ref{prop:weakening}, we obtain $\modelsISL{\Heap_2}{\Perm*\Perm_2}{\Env}{A_2}$ as required.
\item By similar argument to part (1)(b).
\end{enumerate}
\end{enumerate}
}

\noindent This lemma provides two alternative semantics for the implication and wand connectives, which are both equivalent to our actual semantics of Definition \ref{defn:seplogictotalsemantics} if we restrict the logic to self-framing subformulas. In particular, to model the fragment of our logic which corresponds to separation logic, these alternative semantics are sufficient. The latter alternative for each connective (defined in terms of globally-havoced extensions) is the semantics used in our precursor paper~\cite{ParkinsonSummers'11}, while the former (using locally-havoced extensions) is convenient to simplify several of our proofs.

Note that the concepts of minimal permission extensions, and locally-havoced extensions (neither of which were used in our precursor paper) are \emph{not} motivated by our desire to correctly model separation logic semantics in our total heaps model; as the lemma above makes explicit, a simpler semantics could have been defined if this was our only goal. However, that semantics does not extend to correctly handle the implication in implicit dynamic frames, for which we needed the concept of minimal extensions to get the general case correct, as is motivated in Subsection \ref{sect:implication}.

We now turn to relating our total heap semantics for separation logic with the standard semantics. To do this, we need to relate partial heaps with pairs of total heap and permission mask. Given any total heap $\Heap$ and permission mask $\Perm$ we can construct a corresponding partial heap $\restrict{\Heap}{\Perm}$. Conversely, any partial heap $h$ can be represented as the restriction of a total heap $\Heap$ to the permission mask corresponding to all the permissions in $h$. This representation however, is not unique - there are many such total heaps  $\Heap$ we could choose such that $h = \restrict{\Heap}{\Perm}$. However, the different choices of $\Heap$ can only differ over the locations given no permission in $\Perm$, and \as{Corollary \ref{lemma:seplogicframing}} demonstrates that such differences do not affect the semantics of assertions. For our correspondence result, it is therefore without loss of generality to consider partial heaps constructed by $\restrict{\Heap}{\Perm}$. We can then show that our total heap semantics for \SL{} is sound and complete with respect to the standard semantics:
\addprop{theorem}{Correctness of Total Heap Semantics}{thm:correctness}
{\em
For all \SL-assertions $a$, environments $\Env$, total heaps $\Heap$, and permission masks $\Perm$:
\[\Heap,\Perm, \Env \modelsISL a \ \iff \ (\restrict{\Heap}{\Perm}), \Env \modelsSL a\]
}{
  By induction on the structure of $a$.
}
{
By induction on $a$.
First note, if the property holds of an \SL-assertion,
        then the assertion is self-framing.
        Thus, inductively we can assume all sub-assertions are self-framing.
\begin{desCription}
\item\noindent{\hskip-12 pt($a\equiv \pointsto{e.f}{`p}{e'}$):}\
\[\begin{array}{rclr}
\multicolumn{3}{@{}l}{\modelsSLE{\Heap}{\Perm}{\Env}{a}}\\
 &\Leftrightarrow& \Perm[\semCE{e}{\Env,\Heap}, f] \geq `p \ \land\ \Heap[\semSLE{e}{\Env,\Heap}, f] = \semCE{e'}{\Env,\Heap} & \textit{(by defn.)}\\
&\Leftrightarrow& \multicolumn{2}{l}{(\semCE{e}{\Env,\Heap}, f)\in\dom(\restrict{\Heap}{\Perm}) \ \land \ \first{(\restrict{\Heap}{\Perm})[\semCE{e}{\Env,\Heap}, f]} = \semCE{e'}{\Env,\Heap} \ \land }\\
&&\second{(\restrict{\Heap}{\Perm})[\semCE{e}{\Env,\Heap}, f]} \geq `p & \textit{(by defn{.} of (\restrict{\Heap}{\Perm}))}\\
&\Leftrightarrow& \multicolumn{2}{l}{(\semSLE{e}{\Env}, f)\in\dom(\restrict{\Heap}{\Perm}) \ \land \ \first{(\restrict{\Heap}{\Perm})[\semSLE{e}{\Env}, f]} = \semSLE{e'}{\Env} \ \land }\\
&&\second{(\restrict{\Heap}{\Perm})[\semSLE{e}{\Env}, f]} \geq `p & \textit{(by Lemma \ref{lemma:expressions})}\\
&\Leftrightarrow& (\restrict{\Heap}{\Perm}), \Env \modelsSL a & \textit{(by defn.)}
\end{array}\]
\item\noindent{\hskip-12 pt($a\equiv e{=}e'$):}\
\[\begin{array}{rclr}
\multicolumn{3}{@{}l}{\modelsSLE{\Heap}{\Perm}{\Env}{a}}\\
 &\Leftrightarrow& \semSLE{e}{\Env,\Heap} = \semCE{e'}{\Env,\Heap} & \textit{(by defn.)}\\
&\Leftrightarrow& \semSLE{e}{\Env} = \semSLE{e'}{\Env} & \textit{(by Lemma \ref{lemma:expressions})}\\
&\Leftrightarrow& (\restrict{\Heap}{\Perm}), \Env \modelsSL a  & \textit{(by defn.)}
\end{array}\]
\item\noindent{\hskip-12 pt($a\equiv a_1{*}a_2$):}\
\[\begin{array}{rclr}
\multicolumn{3}{@{}l}{\modelsSLE{\Heap}{\Perm}{\Env}{a}}\\
 &\Leftrightarrow& \exists \Perm_1,\Perm_2.(\Perm = \Perm_1*\Perm_2\ \land\ \\
&&\modelsSLE{\Heap}{\Perm_1}{\Env}{a_1}\ \land\ \modelsSLE{\Heap}{\Perm_2}{\Env}{a_2}) & \textit{(by defn.)}\\
&\Leftrightarrow& \exists \Perm_1,\Perm_2.(\Perm = \Perm_1*\Perm_2\ \land\ \\
&&(\restrict{\Heap}{\Perm_1}), \Env\modelsSL a_1\ \land\ (\restrict{\Heap}{\Perm_2}),\Env\modelsSL a_2) & \textit{(by induction, twice)}\\
&\Leftrightarrow& (\restrict{\Heap}{\Perm}), \Env \modelsSL a & \textit{(by defn.)}
\end{array}\]
\item\noindent{\hskip-12 pt($a\equiv a_1{\wedge}a_2$),($a\equiv a_1{\vee}a_2$):}\
Straightforwardly, by induction.
\item\noindent{\hskip-12 pt($a\equiv a_1{\imp}a_2)$:}\
\[\begin{array}{rclr}
\multicolumn{3}{@{}l}{\modelsSLE{\Heap}{\Perm}{\Env}{a}}\\
 &\Leftrightarrow& \forall \Perm_1, \Heap_1.(\;
   \Perm_1 {\perp} \Perm\ \land\    \Heap_1 \agrees{\rds{\Perm}\cup\overline{\rds{\Perm_1}}} \Heap\  \land\\
&&\multicolumn{2}{l}{\Heap_1, \Perm * \Perm_1, \Env \modelsISL a_1
\ \Rightarrow\ \Heap_1, \Perm * \Perm_1, \Env \modelsISL a_2)}\\
&&& \textit{(by Lemma \ref{lemma:newlemmathree} (1))}\\
&\Leftrightarrow& \forall \Perm_1, \Heap_1.(\;
   \Perm_1 {\perp} \Perm\ \land\    \Heap_1 \agrees{\rds{\Perm}\cup\overline{\rds{\Perm_1}}} \Heap\  \land\\
&&\multicolumn{2}{l}{(\restrict{\Heap_1}{(\Perm{*}\Perm_1)}), \Env \modelsSL a_1\
\Rightarrow\ (\restrict{\Heap_1}{(\Perm{*}\Perm_1)}), \Env \modelsSL a_2)}\\
&&& \textit{(by induction, twice)}\\
&\Leftrightarrow& \forall \Perm_1, \Heap_1.(\;
   \Perm_1 {\perp} \Perm\ \land\    \Heap_1 \agrees{\rds{\Perm}\cup\overline{\rds{\Perm_1}}} \Heap\  \land\\
&&\multicolumn{2}{l}{((\restrict{\Heap_1}{\Perm})*(\restrict{\Heap_1}{\Perm_1})), \Env \modelsSL a_1\
\Rightarrow\ ((\restrict{\Heap_1}{\Perm})*(\restrict{\Heap_1}{\Perm_1})), \Env \modelsSL a_2)}\\
&&& \textit{(by defn.)}\\
&\Leftrightarrow& \forall \Perm_1, \Heap_1.(\;
   \Perm_1 {\perp} \Perm\ \land\    \Heap_1 \agrees{\rds{\Perm}\cup\overline{\rds{\Perm_1}}} \Heap\  \land\\
&&\multicolumn{2}{l}{((\restrict{\Heap}{\Perm})*(\restrict{\Heap_1}{\Perm_1})), \Env \modelsSL a_1\
\Rightarrow\ ((\restrict{\Heap}{\Perm})*(\restrict{\Heap_1}{\Perm_1})), \Env \modelsSL a_2)}\\
&&& \textit{(since }\Heap_1\agrees{\Perm}\Heap\textit{)}\\
&\Leftrightarrow&\multicolumn{2}{l}{ \forall \Perm_1, \PHeap_1.(\;
   \Perm_1 {\perp} \dom(\restrict{\Heap}{\Perm})\ \land\ \dom(\PHeap_1) = \rds{\Perm_1} \ \land}\\
&&\multicolumn{2}{l}{(\forall (`i,f)\in(\dom(\PHeap_1){\cap}\dom(\restrict{\Heap}{\Perm})). \PHeap_1[`i,f] = (\restrict{\Heap}{\Perm})[`i,f])\ \land}\\
&&\multicolumn{2}{l}{((\restrict{\Heap}{\Perm})*\PHeap_1), \Env \modelsSL a_1\
\Rightarrow\ ((\restrict{\Heap}{\Perm})*\PHeap_1), \Env \modelsSL a_2))}\\
&&& \textit{(by defn. of }(\restrict{\Heap_1}{\Perm_1})\textit{)}\\
&\Leftrightarrow& \forall \PHeap_1.(\;
   \PHeap_1 {\perp} (\restrict{\Heap}{\Perm})\ \land\\
&&\multicolumn{2}{l}{((\restrict{\Heap}{\Perm})*\PHeap_1), \Env \modelsSL a_1\
\Rightarrow\ ((\restrict{\Heap}{\Perm})*\PHeap_1), \Env \modelsSL a_2)}\\
&&& \textit{(by defn. of }\PHeap_1 {\perp} (\restrict{\Heap}{\Perm})\textit{)}\\
&\Leftrightarrow& (\restrict{\Heap}{\Perm}), \Env \modelsSL a_1\imp a_2
\end{array}\]
\item\noindent{\hskip-12 pt($a\equiv a_1{\wand}a_2)$:}\
Analogous to previous case, using  Lemma \ref{lemma:newlemmathree} (2) instead of  Lemma \ref{lemma:newlemmathree} (1).

\item\noindent{\hskip-12 pt($a\equiv \exists x.\;a')$:}\ We have:
\[\begin{array}{l}
\Heap,\Perm,\Env \modelsISL \exists x.\;a' \\
\Leftrightarrow\
\exists v.\; \Heap,\Perm,\Env[x\mapsto v] \modelsISL a' \\
\Leftrightarrow\
\exists v.\; \restrict{\Heap}{\Perm},\Env[x\mapsto v] \modelsSL a'\ \textit{by inductive hypothesis} \\
\Leftrightarrow\
\restrict{\Heap}{\Perm},\Env \modelsSL \exists x.\; a'
\end{array}\]

\end{desCription}
}

This result demonstrates that our total heap semantics correctly models the standard semantics of separation logic assertions.

\begin{corollary}\label{lemma:seplogicframing}
All separation logic assertions $a$ (Defn~\ref{defn:seplogicsyntax}) are self-framing.
\end{corollary}

\begin{corollary}
All separation logic assertions $a$ are intuitionistic.
\end{corollary}

\subsection{Separation Logic Laws}
Because our assertion language is more general than that of separation logic, not all properties of the separation logic connectives transfer across to the full generality of \SLE. For example, in separation logic, the assertions $a\wand(b\wand c)$ and $(a*b)\wand c$ are (always) equivalent. This is not quite the case in \SLE.
We can show how various laws which hold for separation logic transfer (in some cases partially) to our more general setting of \SLE. Firstly, we need a technical lemma which shows how to break down minimal permission extensions over (separating and logical) conjunctions:
\addprop{lemma}{Decomposing Minimal Permission Extensions over Conjunctions}{lemma:newlemmafour}
{\hfill\em\begin{enumerate}[\em(1)]
\item\label{lemma:newlemmafour:separating} If \extend{\Heap}{\emptyset}{\Env}{\Perm'}{A_1*A_2} then $\exists \Perm_1,\Perm_2$ such that $\Perm' = \Perm_1*\Perm_2$ and $\extend{\Heap}{\emptyset}{\Env}{\Perm_1}{A_1}$ and $\extend{\Heap}{\emptyset}{\Env}{\Perm_2}{A_2}$.
\item\label{lemma:newlemmafour:logical} If \extend{\Heap}{\Perm}{\Env}{\Perm'}{A_1 \wedge A_2} then $\exists \Perm_1,\Perm_2$ such that $\Perm' = \Perm_1*\Perm_2$ and $\extend{\Heap}{\Perm}{\Env}{\Perm_1}{A_1}$ and $\extend{\Heap}{\Perm*\Perm_1}{\Env}{\Perm_2}{A_2}$.
\end{enumerate}
}{%
}{
\hfill
\begin{enumerate}[(1)]
\item
We prove an equivalent statement:
\[
\begin{array}{l}
\extend{\Heap}{\emptyset}{\Env}{\Perm}{A_1 * A_2}
\ \land\
\Perm_3 * \Perm_4 = \Perm
\ \land\
\Heap,\Perm_3, \Env \modelsISL A_1
\ \land\
\Heap,\Perm_4, \Env \modelsISL A_2
\\
\qquad\Rightarrow
\exists \Perm_1,\Perm_2.\ \Perm_1 * \Perm_2 = \Perm
\ \land\
\extend{\Heap}{\emptyset}{\Env}{\Perm_1}{A_1}
\ \land\
\extend{\Heap}{\emptyset}{\Env}{\Perm_2}{A_2}
\end{array}
\]
by complete (strong) induction on $|\rds{\Perm_3} \cap \rds{\Perm_4}|$. In the proof, we use the shorthand $\Perm[(\iota,f) \mapsto \pi]$ to denote the permission mask that returns $\pi$ for $(\iota,f)$ and behaves like $\Perm$ for all other entries, and also the shorthand $\Perm \setminus (\iota,f)$ for $\Perm[(\iota,f)\mapsto \none]$].

We now consider two cases:
\begin{desCription}
\item\noindent{\hskip-12 pt($\exists (\iota,f) \in \rds{\Perm_3}.\ \Heap,(\Perm_3 \setminus (\iota,f)), \Env \modelsISL A_1$)}\mbox{}
\\By Proposition~\ref{prop:weakening}, we know $\Heap,(\Perm_4 [(\iota,f) \mapsto \Perm[\iota,f]]), \Env \modelsISL A_2$.
We know $\Perm_4[\iota,f] > \none$, otherwise $\Perm_3 * \Perm_4$ was not minimal to start with.
Therefore $|\rds{\Perm_3 \setminus (\iota,f)} \cap \rds{\Perm_4 [(\iota,f) \mapsto \Perm[\iota,f]]}| < |\rds{\Perm_3} \cap \rds{\Perm_4}|$.
By construction, we know $(\Perm_3 \setminus (\iota,f)) * (\Perm_4 [(\iota,f) \mapsto \Perm[\iota,f]]) = \Perm_3 * \Perm_4$, so this case holds by induction
choosing $(\Perm_3 \setminus (\iota,f))$ for $\Perm_3$ and $\Perm_4 [(\iota,f) \mapsto \Perm[\iota,f]]$ for $\Perm_4$ in inductive hypothesis.

\item\noindent{\hskip-12 pt($\forall (\iota,f) \in \rds{\Perm_3}.\ \Heap,(\Perm_3 \setminus (\iota,f)), \Env \notmodelsISL A_1$)}\mbox{}\\
Therefore, $\extend{\Heap}{\emptyset}{\Env}{\Perm_3}{A_1}$. Consider two sub-cases:
\begin{desCription}
\item\noindent{\hskip-12 pt($\exists (\iota,f) \in \rds{\Perm_4}.\ \Heap,(\Perm_4 \setminus (\iota,f)), \Env \modelsISL A_2$)}\mbox{}\\
By Proposition~\ref{prop:weakening}, we know $\Heap,(\Perm_3 [(\iota,f) \mapsto \Perm[\iota,f]]), \Env \modelsISL A_1$.
By construction, we know $\Perm_3[\iota,f] > \none$, otherwise $\Perm_3 * \Perm_4$ was not minimal to start with.
Therefore $|\rds{\Perm_4 \setminus (\iota,f)} \cap \rds{\Perm_3 [(\iota,f) \mapsto \Perm[\iota,f]]}| < |\rds{\Perm_3} \cap \rds{\Perm_4}|$.
We know $(\Perm_4 \setminus (\iota,f)) * (\Perm_3 [(\iota,f) \mapsto \Perm[\iota,f]]) = \Perm_3 * \Perm_4$, so this case holds by induction
choosing $(\Perm_4 \setminus (\iota,f))$ for $\Perm_4$ and $\Perm_3 [(\iota,f) \mapsto \Perm[\iota,f]]$ for $\Perm_3$ in inductive hypothesis.

\item\noindent{\hskip-12 pt($\forall (\iota,f) \in \rds{\Perm_4}.\ \Heap,(\Perm_4 \setminus (\iota,f)), \Env \notmodelsISL A_2$)}\mbox{}\\
Then, $\extend{\Heap}{\emptyset}{\Env}{\Perm_4}{A_2}$.  Hence, we have solution choosing $\Perm_3 = \Perm_1$ and $\Perm_4 = \Perm_2$.
\end{desCription}
\end{desCription}

\item
We can assume $\extend{\Heap}{\Perm}{\Env}{\Perm'}{A_1 \land A_2}$
and hence $\Heap, \Perm * \Perm', \Env \models A_1$ and
$\Heap, \Perm * \Perm', \Env \models A_2$.
By Lemma~\ref{lemma:newlemmaone},
	 we know there exists $\Perm_1'$ such that $\Perm_1' \subseteq \Perm'$
	   and $\extend{\Heap}{ \Perm}{\Env}{\Perm_1'}{A_1}$.
Define
\[
  \Perm_1 = \lambda (\iota,f).\ \text{if } \Perm'_1[\iota,f]=\none\text{ then } \none\text{ else } \Perm'[\iota,f]\text{.}
\]
Note that $\Perm_1 \subseteq \Perm'$.
By Lemma~\ref{lemma:fuzzyedge}, we know $\extend{\Heap}{ \Perm}{\Env}{\Perm_1}{A_1}$.

We know $\Heap, (\Perm * \Perm_1) * (\Perm' - \Perm_1), \Env \modelsISL A_2$, and by Lemma~\ref{lemma:newlemmaone}
we know there exists $\Perm'_2 \subseteq (\Perm' - \Perm_1)$ such that $\extend{\Heap}{\Perm*\Perm_1}{\Env}{\Perm'_2}{A_2}$.
Define
\[
  \Perm_2 = \lambda (\iota,f).\ \text{if } \Perm'_2[\iota,f]=\none\text{ then } \none\text{ else } \Perm'[\iota,f]\text{.}
\]
Note that $\Perm_2 \subseteq \Perm' - \Perm_1$, and thus $\Perm_1 * \Perm_2 \subseteq \Perm'$.
By Lemma~\ref{lemma:fuzzyedge}, we deduce that $\extend{\Heap}{ \Perm}{\Env}{\Perm_2}{A_2}$.

By construction of $\Perm_1$ and $\Perm_2$, either $\Perm_1 * \Perm_2 = \Perm'$ or $\rds{\Perm_1 * \Perm_2} \subset \rds{\Perm'}$.
In the first case we are done. In the second case, we seek a contradiction.
We know $\Heap,\Perm*\Perm_1*\Perm_2,\Env \modelsISL A_2$,
and by Prop~\ref{prop:weakening}, we know $\Heap,\Perm*\Perm_1*\Perm_2,\Env \modelsISL A_1$,
hence $\Heap,\Perm*\Perm_1*\Perm_2,\Env \modelsISL A_1 \land A_2$.
As we know $\Perm_1 * \Perm_2 \subseteq \Perm'$,
but that contradicts the initial assumption of $\Perm'$ being minimal.
\end{enumerate}
}

\noindent Certain \as{technical properties which follow} do not hold for general \as{formulas}, but only for \as{those} that do not behave disjunctively.  Following, O'Hearn \emph{et al.}~\cite{ohearn04}, we call these \as{formulas} \emph{supported}.
A formula such as $\acc(x.f,1) \lor \acc(y.f,1)$ is not supported, while $(b=1 \imp \acc(x.f,1)) * (b\neq 1 \imp \acc(y.f,1)$ is supported.
\begin{definition}[Supported \as{Formulas}]
A formula $A$ is \emph{supported} iff for all $\Heap$, $\Env$, $\Perm_1$ and $\Perm_2$, if $\Heap,\Perm_1,\Env \modelsISL A $ and $\Heap,\Perm_2,\Env \modelsISL A$
, then $\Heap,\Perm_1 \sqcap \Perm_2,\Env \modelsISL A$.
\end{definition}

Supported assertions allow minimal permission extensions to be combined for both $*$ and $\land$.
\addprop{lemma}{Composing Minimal Permission Extensions over Supported
  Conjunctions\vspace{-12 pt}}{lemma:newlemmasix}
{\hfill\em\begin{enumerate}[\em(1)]
\item\label{lemma:newlemmasix:separating} If \extend{\Heap}{\emptyset}{\Env}{\Perm_1}{A_1} and \extend{\Heap}{\emptyset}{\Env}{\Perm_2}{A_2} and $A_1$ and $A_2$ are supported, then $\extend{\Heap}{\emptyset}{\Env}{\Perm_1*\Perm_2}{A_1*A_2}$.
\item\label{lemma:newlemmasix:logical} If \extend{\Heap}{\Perm}{\Env}{\Perm_1}{A_1} and \extend{\Heap}{\Perm*\Perm_1}{\Env}{\Perm_2}{A_2} and $A_1$ and $A_2$ are supported, then $\extend{\Heap}{\Perm}{\Env}{\Perm_1*\Perm_2}{A_1\wedge A_2}$.
\end{enumerate}
}{%
}
{\hfill\begin{enumerate}[(1)]
\item
Proof by contradiction.  We know $\Heap, \Perm_1 * \Perm_2, \Env \modelsISL A_1 * A_2$ holds, therefore we assume that there exists $\Perm' \subseteq \Perm_1 * \Perm_2$ such that $\rds{\Perm'} \subset \rds{\Perm_1 * \Perm_2}$  and $\Heap, \Perm', \Env \modelsISL A_1 * A_2$.
Therefore, there exist $\Perm_3$ and $\Perm_4$ such that $\Perm_3 * \Perm_4 = \Perm'$ and $\Heap,\Perm_3,\Env \modelsISL A_1$ and $\Heap, \Perm_4, \Env \modelsISL A_2$.  From $\rds{\Perm'} \subset \rds{\Perm_1 * \Perm_2}$, we know $(\iota,f) \in \rds{\Perm_1 * \Perm_2}$ and $(\iota,f) \notin \rds{\Perm'}$. W.l.o.g assume $(\iota,f) \in \rds{\Perm_1}$.  As $A_1$ is supported, we know $\Heap,\Perm_1 \sqcap \Perm_3,\Env \modelsISL A_1$. Since $(\Perm_1 \sqcap \Perm_3) \subseteq \Perm_3 \subseteq \Perm'$ we have $(\iota,f) \notin \rds{\Perm_1 \sqcap \Perm_3}$. But this contradicts $\extend{\Heap}{\Perm}{\Env}{\Perm_1}{A_1}$, since $(\iota,f) \in \rds{\Perm_1}$.

\item
By assumptions and Prop~\ref{prop:weakening}, we have $\Heap,\Perm * \Perm_1 * \Perm_2, \Env \modelsISL A_1 \land A_2$.
We show $\extend{\Heap}{\Perm}{\Env}{\Perm_1*\Perm_2}{A_1\wedge A_2}$ by contradiction.  Assume there exists $\Perm_3 \subseteq (\Perm_1 * \Perm_2)$ such that $\rds{\Perm_3} \subset \rds{\Perm_1 * \Perm_2}$ and
$\Heap, \Perm * \Perm_3, \Env \modelsISL A_1 \land A_2$. Then there exists $(\iota,f) \in \rds{\Perm_1 * \Perm_2}$ such that $(\iota,f) \notin \rds{\Perm_3}$. Case split:
\begin{desCription}
\item\noindent{\hskip-12 pt($(\iota,f) \in \rds{\Perm_1}$)}\
Note that $(\Perm * \Perm_3) \sqcap (\Perm * \Perm_1) = \Perm * (\Perm_1 \sqcap \Perm_3)$, thus as $A_1$ is supported we have $\Heap, \Perm * (\Perm_1 \sqcap \Perm_3), \Env \modelsISL A_1$.
 but this contradicts
$\extend{\Heap}{\Perm}{\Env}{\Perm_1}{A_1}$

\item\noindent{\hskip-12 pt($(\iota,f) \in \rds{\Perm_2}$)}\
Note that $(\Perm * \Perm_3) \sqcap (\Perm * \Perm_1 * \Perm_2) =
\Perm * (\Perm_3 \sqcap (\Perm_1 * \Perm_2))$. As $A_2$ is supported, we know $\Heap, \Perm * (\Perm_3 \sqcap (\Perm_1 * \Perm_2)),\Env \modelsISL A_2$, but this contradicts $\extend{\Heap}{\Perm*\Perm_1}{\Env}{\Perm_2}{A_2}$.
\end{desCription}

\end{enumerate}
}

\noindent We can now show which of the usual separation logic laws
carry over to our more general logic, and under which conditions:
\addprop{proposition}{}{prop:stdproperties} {\em For all \SLE{}
  assertions $A_1$, $A_2$, $A_3$:
\begin{enumerate}[\em(1)]
\item $ A_1*(A_1\wand A_2) \entails  A_2$
\item $ A_1\land(A_1\imp A_2) \entails  A_2$
\item \label{prop:bullet:currying}\label{prop:bullet:uncurrying}

\begin{enumerate}[\em(a)]
\item

$\deframed{A_1} \cap \semEl{A_1\wand(A_2\wand A_3)}  \subseteq
\semEl{(A_1*A_2)\wand A_3}$

\item if $A_1$ and $A_2$ are supported, then:\\
$\deframed{A_1} \cap \semEl{(A_1 * A_2)\wand A_3}  \subseteq
\semEl{(A_1\wand (A_2\wand A_3)}$

\item if both $A_1*A_2$ and $A_3$ are self-framing, then:\\
$\deframed{A_1} \cap \semEl{(A_1 * A_2)\wand A_3}  \subseteq
\semEl{(A_1\wand (A_2\wand A_3)}$

\end{enumerate}

\item
\begin{enumerate}[\em(a)]
\item
$\eframed{A_1} \cap \semEl{A_1\imp(A_2\imp A_3)} \subseteq \semEl{(A_1\land A_2)\imp A_3}$

\item
if $A_1$ and $A_2$ are supported, then:\\
$\eframed{A_1} \cap \semEl{(A_1\wedge A_2)\imp A_3} \subseteq \semEl{A_1\imp(A_2\imp A_3)}$

\item if both $A_1\wedge A_2$ and $A_3$ are self-framing,  then:\\
$\eframed{A_1} \cap \semEl{(A_1\land A_2)\imp A_3} \subseteq \semEl{A_1\imp(A_2\imp A_3)}$
\end{enumerate}

\item If $A_1 \entails (A_2\wand A_3)$ then $(A_1*A_2)\entails A_3$
\item If $A_1$ is self-framing and $(A_1*A_2)\entails A_3$ then $A_1 \entails (A_2\wand A_3)$
\end{enumerate}

}{%
} {
\mbox{}
\begin{enumerate}[(1)]
\item Assume $\modelsISL{\Heap}{\Perm}{\Env}{A_1*(A_1\wand A_2)}$. We seek to prove that $\modelsISL{\Heap}{\Perm}{\Env}{A_2}$.
From our assumption, there exist $\Perm_1, \Perm_2$ such that $\Perm_1*\Perm_2 = \Perm$ and $\modelsISL{\Heap}{\Perm_1}{\Env}{A_1}$ and $\modelsISL{\Heap}{\Perm_2}{\Env}{A_1\wand A_2}$. From the latter, we have that
\[
\begin{array}{r}
\forall \Perm_3\perp\Perm_2, \forall \Heap_3\agrees{\rds{\Perm_2}\cup\overline{\rds{\Perm_3}}}\Heap. \extend{\Heap_3}{\emptyset}{\Env}{\Perm_3}{A_1} \qquad\mbox{}\\
\Rightarrow \modelsISL{\Heap_3}{\Perm_2*\Perm_3}{\Env}{A_2}\text{.}
\end{array}
\]
From $\modelsISL{\Heap}{\Perm_1}{\Env}{A_1}$, by Lemma \ref{lemma:newlemmaone}, we know that there exists $\Perm_3\subseteq\Perm_1$ such that $\extend{\Heap}{\emptyset}{\Env}{\Perm_3}{A_1}$. Combining these facts, we obtain that $\modelsISL{\Heap}{\Perm_2*\Perm_3}{\Env}{A_2}$ holds. By Proposition \ref{prop:weakening}, we obtain $\modelsISL{\Heap}{\Perm_2*\Perm_1}{\Env}{A_2}$ as required.

\item Assume $\modelsISL{\Heap}{\Perm}{\Env}{A_1\wedge(A_1\imp A_2)}$. We seek to prove that $\modelsISL{\Heap}{\Perm}{\Env}{A_2}$.
From our assumption, we obtain both $\modelsISL{\Heap}{\Perm}{\Env}{A_1}$ and $\modelsISL{\Heap}{\Perm}{\Env}{A_1\imp A_2}$. From the latter, we have \[\forall \Perm_1\perp\Perm, \forall \Heap_1\agrees{\rds{\Perm}\cup\overline{\rds{\Perm_1}}}\Heap. \extend{\Heap_1}{\Perm}{\Env}{\Perm_1}{A_1} \Rightarrow \modelsISL{\Heap_1}{\Perm*\Perm_1}{\Env}{A_2}\] Taking $\Heap_1 = \Heap$ and $\Perm_1 = \emptyset$ in the above (and noting that from $\modelsISL{\Heap}{\Perm}{\Env}{A_1}$ we can easily obtain $\extend{\Heap}{\Perm}{\Env}{\emptyset}{A_1}$), we can obtain $\modelsISL{\Heap}{\Perm}{\Env}{A_2}$ as required.

\item In the following, we assume (as in the statement of the Lemma) that $(\Heap, \Perm, \Env) \in \deframed{A_1}$
\begin{enumerate}[(a)]
\item We assume $\Heap, \Perm,\Env \modelsISL A_1\wand(A_2\wand A_3)$ and seek to prove $\Heap, \Perm,\Env \modelsISL (A_1*A_2)\wand A_3$.  Thus, we can assume
\[
\begin{array}{l}
(\Heap_3,\Perm_3,\Env) \in \localDisjExts{\Heap,\Perm,\sigma}\\
\extend{\Heap_3}{\emptyset}{\Env}{\Perm_3}{A_1 * A_2}
\end{array}
\]
and must prove
\[
\Heap_3,\Perm*\Perm_3,\Env \modelsISL A_3
\]
By Lemma \ref{lemma:newlemmafour}~(\ref{lemma:newlemmafour:separating}), there exist $\Perm_1\perp\Perm_2$ such that $\Perm_3=\Perm_1*\Perm_2$ and both
\[
\begin{array}{l}
\extend{\Heap_3}{\emptyset}{\Env}{\Perm_1}{A_1}\\
\extend{\Heap_3}{\emptyset}{\Env}{\Perm_2}{A_2}
\end{array}
\]
By the definition of $\localDisjExts{\Heap,\Perm,\sigma}$, we can show
\[
\begin{array}{l}
(\Heap', \Perm_1, \Env) \in \localDisjExts{\Heap,\Perm,\sigma}\\
(\Heap_3, \Perm_2, \Env) \in \localDisjExts{\Heap',\Perm * \Perm_1,\sigma}\\
\end{array}
\]
where $\Heap' = (\Perm_1 ? \Heap_3 : \Heap)$. By assumptions, we know $\Heap_3 \agrees{\Perm} \Heap$, thus $\Heap' \agrees{\Perm} \Heap_3$.  By construction, $\Heap' \agrees{\Perm_1} \Heap_3$, and thus $\Heap' \agrees{\Perm * \Perm_1} \Heap_3$.

By $\deframed{\Heap,\Perm,\Env}$ assumption, and Lemma~\ref{lemma:newlemmatwo}~(2), we get
\[
\extend{\Heap'}{\emptyset}{\Env}{\Perm_1}{A_1}
\]
Now using  $\Heap, \Perm,\Env \modelsISL A_1\wand(A_2\wand A_3)$, we get
\[
\Heap', \Perm*\Perm_1,\Env \modelsISL A_2\wand A_3
\]
and thus
\[
\Heap_3, \Perm*\Perm_1*\Perm_2,\Env \modelsISL A_3
\]
as required.

\item and (c)\ \ We prove these two cases together, since they are almost identical. In the proof, we case split on which extra assumption to use: either $A_1$ and $A_2$ are supported (for part (b)) or both $A_1*A_2$ and $A_3$ are self-framing (for part (c)).

We assume $\Heap, \Perm,\Env \modelsISL (A_1*A_2)\wand A_3$ and seek to prove $\Heap, \Perm,\Env \modelsISL A_1\wand(A_2\wand A_3)$.  Thus, we can assume
\[
\begin{array}{l}
(\Heap_1,\Perm_1,\Env) \in \localDisjExts{\Heap,\Perm,\sigma}\\
\extend{\Heap_1}{\emptyset}{\Env}{\Perm_1}{A_1}\\
(\Heap_2,\Perm_2,\Env) \in \localDisjExts{\Heap_1,\Perm*\Perm_1,\sigma}\\
\extend{\Heap_2}{\emptyset}{\Env}{\Perm_2}{A_2}\\
\end{array}
\]
and must prove
\[
\Heap_2,\Perm*\Perm_1*\Perm_2,\Env \modelsISL A_3
\]
By definition of $\localDisjExts{\Heap,\Perm,\sigma}$ we can show
\[
(\Heap_2,\Perm_1*\Perm_2,\Env) \in \localDisjExts{\Heap,\Perm,\sigma}\\
\]
By Lemma~\ref{lem:extfrmuses}, we know
\[
(\Heap_1,\Perm,\Env) \in \deframed{A_1}
\]
and thus, by Lemma~\ref{lemma:newlemmatwo}~(2) we know
\[
  \extend{\Heap_2}{\emptyset}{\Env}{\Perm_1}{A_1}\\
\]
Now, we case-split on whether we are proving part (b) or (c):
\begin{desCription}
\item[part (b)] Then we can assume that $A_1$ and $A_2$ are supported. By Lemma \ref{lemma:newlemmasix}~(\ref{lemma:newlemmasix:separating}), we can obtain $\extend{\Heap_2}{\emptyset}{\Env}{\Perm_1*\Perm_2}{A_1 * A_2}$. Thus, using $\Heap, \Perm,\Env \modelsISL (A_1*A_2)\wand A_3$ we get
\[
\Heap_2,\Perm*\Perm_1*\Perm_2,\Env \modelsISL A_3
\]
as required.

\item[part (c)]
Then, we can assume that both $A_1*A_2$ and $A_3$ are self-framing.
 By Lemma \ref{lemma:newlemmaone}, there exists $\Perm_3\subseteq(\Perm_1*\Perm_2)$ such that $\extend{\Heap_2}{\emptyset}{\Env}{\Perm_3}{A_1*A_2}$. Define $\Heap_3 = \condheap{\Perm_3}{\Heap_2}{\Heap}$. Then we have $\Heap_3\agrees{\Perm*\Perm_3}\Heap_2$. Since $A_1*A_2$ is self-framing, by Lemma \ref{lemma:newlemmatwo}~(3), we have $\extend{\Heap_3}{\emptyset}{\Env}{\Perm_3}{A_1*A_2}$.
We need to show that
\[
(\Heap_3,\Perm_3,\Env) \in \localDisjExts{\Heap,\Perm,\Env}
\]
which follows as $\Heap_3 \agrees{\Perm} \Heap$ (since $\Heap \agrees{\Perm} \Heap_1 \agrees{\Perm} \Heap_2 \agrees{\Perm} \Heap_3$). Thus, by assumption, we get
$\modelsISL{\Heap_3}{\Perm*\Perm_3}{\Env}{A_3}$. Since $A_3$ is self-framing, and since $\Heap_2\agrees{\Perm*\Perm_3}\Heap_3$, we obtain $\modelsISL{\Heap_2}{\Perm*\Perm_3}{\Env}{A_3}$. By Proposition \ref{prop:weakening}, we have $\modelsISL{\Heap_2}{\Perm*\Perm_1*\Perm_2}{\Env}{A_3}$ as required.

\end{desCription}

\end{enumerate}

\item  In the following, we assume (as in the statement of the Lemma) that $(\Heap, \Perm, \Env) \in \eframed{A_1}$

\begin{enumerate}[(a)]
\item
 We assume $\Heap, \Perm,\Env \modelsISL A_1\imp(A_2\imp A_3)$ and seek to prove that  $\Heap, \Perm,\Env \modelsISL (A_1\wedge A_2)\imp A_3$.  Thus, we can assume
\[
\begin{array}{l}
(\Heap_3,\Perm* \Perm_3,\Env) \in \localExts{\Heap,\Perm,\sigma}\\
\extend{\Heap_3}{\Perm}{\Env}{\Perm_3}{A_1 \wedge A_2}
\end{array}
\]
and must prove
\[
\Heap_3,\Perm * \Perm_3,\Env \modelsISL A_3
\]
By Lemma \ref{lemma:newlemmafour}~(\ref{lemma:newlemmafour:logical}), there exist $\Perm_1\perp\Perm_2$ such that $\Perm_3=\Perm_1*\Perm_2$ and both
\[
\begin{array}{l}
\extend{\Heap_3}{\Perm}{\Env}{\Perm_1}{A_1}\\
\extend{\Heap_3}{\Perm*\Perm_1}{\Env}{\Perm_2}{A_2}
\end{array}
\]
By the definition of $\localExts{\Heap,\Perm,\sigma}$, we can show
\[
\begin{array}{l}
(\Heap', \Perm* \Perm_1, \Env) \in \localExts{\Heap,\Perm,\sigma}\\
(\Heap_3, \Perm * \Perm_1 * \Perm_2, \Env) \in \localExts{\Heap',\Perm * \Perm_1,\sigma}\\
\end{array}
\]
where $\Heap' = (\Perm_1 ? \Heap_3 : \Heap)$. By assumptions, we know $\Heap_3 \agrees{\Perm} \Heap$, thus $\Heap' \agrees{\Perm} \Heap_3$.  By construction, $\Heap' \agrees{\Perm_1} \Heap_3$, and thus $\Heap_3 \agrees{\Perm * \Perm_1} \Heap'$.

By $\eframed{\Heap,\Perm,\Env}$ assumption, and Lemma~\ref{lemma:newlemmatwo}~(1), we get
\[
\extend{\Heap'}{\Perm}{\Env}{\Perm_1}{A_1}
\]
Now using  $\Heap, \Perm,\Env \modelsISL A_1\imp(A_2\imp A_3)$, we get
\[
\Heap', \Perm*\Perm_1,\Env \modelsISL A_2\imp A_3
\]
and thus
\[
\Heap_3, \Perm*\Perm_1*\Perm_2,\Env \modelsISL A_3
\]
as required.

\item and (c)\ \
 We prove these two cases together, since they are almost identical. In the proof, we case split on which extra assumption to use: either $A_1$ and $A_2$ are supported (for part (b)) or both $A_1\wedge A_2$ and $A_3$ are self-framing (for part (c)).

We assume $\Heap, \Perm,\Env \modelsISL (A_1\wedge A_2)\imp A_3$ and seek to prove $\Heap, \Perm,\Env \modelsISL A_1\imp(A_2\imp A_3)$.  Thus, we can assume
\[
\begin{array}{l}
(\Heap_1,\Perm_1 * \Perm,\Env) \in \localExts{\Heap,\Perm,\sigma}\\
\extend{\Heap_1}{\Perm}{\Env}{\Perm_1}{A_1}\\
(\Heap_2,\Perm*\Perm_1 * \Perm_2,\Env) \in \localExts{\Heap_1,\Perm*\Perm_1*\Perm_2,\sigma}\\
\extend{\Heap_2}{\Perm*\Perm_1}{\Env}{\Perm_2}{A_2}\\
\end{array}
\]
and must prove
\[
\Heap_2,\Perm*\Perm_1*\Perm_2,\Env \modelsISL A_3
\]
By definition of $\localExts{\Heap,\Perm,\sigma}$ we can show
\[
(\Heap_2,\Perm*\Perm_1*\Perm_2,\Env) \in \localExts{\Heap,\Perm,\sigma}\\
\]
By Lemma~\ref{lem:extfrmuses}, we know
\[
(\Heap_1,\Perm * \Perm_1,\Env) \in \eframed{A_1}
\]
and thus, by Lemma~\ref{lemma:newlemmatwo}~(2) we know
\[
  \extend{\Heap_2}{\Perm}{\Env}{\Perm_1}{A_1}\\
\]
Now, we case-split on whether we are proving part (b) or (c):
\begin{desCription}
\item[part (b)] Then we can assume that $A_1$ and $A_2$ are supported. By Lemma \ref{lemma:newlemmasix}~(\ref{lemma:newlemmasix:logical}), we can obtain $\extend{\Heap_2}{\Perm}{\Env}{\Perm_1*\Perm_2}{A_1 \wedge A_2}$. Thus, using $\Heap, \Perm,\Env \modelsISL (A_1 \wedge A_2)\imp A_3$ we get
\[
\Heap_2,\Perm*\Perm_1*\Perm_2,\Env \modelsISL A_3
\]
as required.

\item[part (c)]
Then, we can assume that both $A_1\wedge A_2$ and $A_3$ are self-framing.
 By Lemma \ref{lemma:newlemmaone}, there exists $\Perm_3\subseteq(\Perm_1*\Perm_2)$ such that $\extend{\Heap_2}{\Perm}{\Env}{\Perm_3}{A_1 \wedge A_2}$. Define $\Heap_3 = \condheap{\Perm_3}{\Heap_2}{\Heap}$. Then we have $\Heap_3\agrees{\Perm*\Perm_3}\Heap_2$. Since $A_1 \wedge A_2$ is self-framing, by Lemma \ref{lemma:newlemmatwo}~(3), we have $\extend{\Heap_3}{\Perm}{\Env}{\Perm_3}{A_1*A_2}$.
We need to show that
\[
(\Heap_3,\Perm * \Perm_3,\Env) \in \localExts{\Heap,\Perm,\Env}
\]
which follows as $\Heap_3 \agrees{\Perm} \Heap$. By assumption, we get
$\modelsISL{\Heap_3}{\Perm*\Perm_3}{\Env}{A_3}$. Since $A_3$ is self-framing, and since $\Heap_2\agrees{\Perm*\Perm_3}\Heap_3$, we obtain $\modelsISL{\Heap_2}{\Perm*\Perm_3}{\Env}{A_3}$. By Proposition \ref{prop:weakening}, we have $\modelsISL{\Heap_2}{\Perm*\Perm_1*\Perm_2}{\Env}{A_3}$ as required.

\end{desCription}

\end{enumerate}

\item We can assume that:
\[\begin{array}{l}\forall \Heap,\Perm,\Env. (\modelsISL{\Heap}{\Perm}{\Env}{A_1}\Rightarrow \\
\quad\forall \Perm_1\perp\Perm, \forall \Heap_1\agrees{\rds{\Perm}\cup\overline{\rds{\Perm_1}}}\Heap. \\
\quad\quad (\extend{\Heap_1}{\emptyset}{\Env}{\Perm_1}{A_2} \Rightarrow \modelsISL{\Heap_1}{\Perm*\Perm_1}{\Env}{A_3}))
\end{array}\]
We need to know that, assuming that (for some $\Heap_2,\Perm_2$) $\modelsISL{\Heap_2}{\Perm_2}{\Env}{A_1*A_2}$ holds, we can deduce that $\modelsISL{\Heap_2}{\Perm_2}{\Env}{A_3}$ also holds. The former means that there exist $\Perm_3$ and $\Perm_4$ such that $\Perm_2 = \Perm_3*\Perm_4$ and both $\modelsISL{\Heap_2}{\Perm_3}{\Env}{A_1}$ and $\modelsISL{\Heap_2}{\Perm_4}{\Env}{A_2}$ hold. By Lemma \ref{lemma:newlemmaone}, there exists $\Perm_5\subseteq\Perm_4$ such that $\extend{\Heap_2}{\emptyset}{\Env}{\Perm_5}{A_2}$ holds. Now we apply our original assumption, defining $\Heap = \Heap_2$ and $\Heap_1 = \Heap_2$ and $\Perm = (\Perm_3 * (\Perm_4 - \Perm_5))$ and $\Perm_1 = \Perm_5$ (note that, by Proposition \ref{prop:weakening}, we have $\modelsISL{\Heap}{\Perm}{\Env}{A_1}$). From the assumption, we obtain $\modelsISL{\Heap}{\Perm*\Perm_1}{\Env}{A_3}$, i.e., $\modelsISL{\Heap_2}{\Perm_3*\Perm_4}{\Env}{A_3}$, i.e., $\modelsISL{\Heap_2}{\Perm_2}{\Env}{A_3}$ as required.
\item We can assume that
\[\forall\Heap,\Perm,\Env. (\modelsISL{\Heap}{\Perm}{\Env}{A_1*A_2} \Rightarrow \modelsISL{\Heap}{\Perm}{\Env}{A_3})\]
i.e., we (equivalently) assume that:
\[\begin{array}{r}
\forall\Heap,\Perm_1,\Perm_2,\Env. (\Perm_1\perp\Perm_2 \land \modelsISL{\Heap}{\Perm_1}{\Env}{A_1} \land \modelsISL{\Heap}{\Perm_2}{\Env}{A_2} \Rightarrow\quad\mbox{}\\
\quad \modelsISL{\Heap}{\Perm_1*\Perm_2}{\Env}{A_3})\end{array}\]
We need to show that, if we assume (for some $\Heap_1$ and $\Perm_1$) that $\modelsISL{\Heap_1}{\Perm_1}{\Env}{A_1}$, then we can deduce that $\modelsISL{\Heap_1}{\Perm_1}{\Env}{A_2\wand A_3}$ holds, i.e., that:
\[
\begin{array}{l}
\forall \Perm_2\perp\Perm_1,\forall \Heap_2\agrees{\rds{\Perm_1}\cup\overline{\rds{\Perm_2}}}\Heap_1. \\
\qquad\qquad(\extend{\Heap_2}{\emptyset}{\Env}{\Perm_2}{A_2} \Rightarrow \modelsISL{\Heap_2}{\Perm_1*\Perm_2}{\Env}{A_3})
\end{array}
\]
To show this, we assume $\Perm_2\perp\Perm_1$ and $\Heap_2\agrees{\rds{\Perm_1}\cup\overline{\rds{\Perm_2}}}\Heap_1$ and $\extend{\Heap_2}{\emptyset}{\Env}{\Perm_2}{A_2}$ and need to prove $\modelsISL{\Heap_2}{\Perm_1*\Perm_2}{\Env}{A_3}$. Since $A_1$ is self-framing, and since $\Heap_2\agrees{\Perm_1}\Heap_1$, we know that $\modelsISL{\Heap_2}{\Perm_1}{\Env}{A_1}$. Then, letting $\Heap = \Heap_2$, we can apply our original assumption to obtain $\modelsISL{\Heap_2}{\Perm_1*\Perm_2}{\Env}{A_3}$ as required.
\end{enumerate}
}\smallskip

\noindent To see that the usual separation logic laws do not all hold in general, consider for example the two assertions $A_1 \stackrel{\textit{def}}{=}(x.f=1 \wand (\acc(x.f,1) \wand \mathsf{false}))$ and $A_2 \stackrel{\textit{def}}{=}(x.f=1 * \acc(x.f,1)) \wand \mathsf{false}$. The assertion $A_2$ is equivalent to $\acc(x.f,\_)$, that is a permissions mask, which cannot be extended with disjoint full access to $x.f$. However, the assertion $A_1$ is also true in models where the heap maps $x.f$ to a value other than $1$, as the outer wand does not get to change the current heap.

The usual separation logic laws do however hold for self-framing assertions which (by Lemma \ref{lemma:seplogicframing}) includes all separation logic assertions.
\begin{corollary}
For all self-framaing \SLE{} assertions $A_1$, $A_2$, $A_3$:
\begin{enumerate}[\em(1)]
\item $A_1*(A_1\wand A_2)\entails A_2$
\item $A_1\land(A_1\imp A_2)\entails A_2$
\item $A_1\wand(A_2\wand A_3)\isequiv (A_1*A_2)\wand A_3$
\item $A_1\imp(A_2\imp A_3) \isequiv (A_1\land A_2)\imp A_3$
\item $A_1 \entails (A_2\wand A_3)$ if and only if $(A_1*A_2)\entails A_3$
\end{enumerate}
\end{corollary}

\subsection{Existentials and Substitution}

Next we consider when it is valid to replace a variable with an expression it is equal to; that is, under what condition is
$\exists x. x=E * A $ equivalent to $A[E/x]$.  If the expression does not depend on the heap, then this equivalence holds.
\begin{lemma}\label{substitution2}
For any \emph{separation logic expression} $e$:
\[
(\exists x. x = e * A)  \isequiv A[e/x]
\]
\end{lemma}
\begin{proof}
We prove
\[
\Heap,\Perm,\Env \modelsISL A[e/x]
\Leftrightarrow
\Heap,\Perm,\Env[x \mapsto \semCE{e}{\Heap,\Env}] \modelsISL A
\]
and
\[
\semCE{E[E'/x]}{\Heap,\Env}
\Leftrightarrow
\semCE{E}{\Heap,\Env[x \mapsto \semCE{E'}{\Heap,\Env}]}
\]
by straightforward inductions on structures of $A$ and $E$, respectively.
\end{proof}
However, if the expression depends on the heap, then the problem is more challenging.  Consider the example formula
\[
\exists v.\; v=x.f * \acc(x.f,\pi) * (\acc(x.f,\pi) \wand v=5)
\]
This formula is semantically equivalent (noting that changes to the heap do not affect the interpretation of $v$) to
\[
\acc(x.f,\pi) * x.f = 5
\]
However, if we apply the standard substitution on the formula, replacing $v$ with the expression $x.f$, then we get
\[
\acc(x.f,\pi) * (\acc(x.f,\pi) \wand x.f=5)
\]
which is equivalent to $\textit{false}$ (recall that the semantics for the $\wand$ connective considers ``adding on'' new permission for $x.f$ in this case, which includes considering changing its value arbitrarily). More abstractly, the difficulty here is that the semantics of $\wand$, and $\imp$ consider \emph{changes} to the current heap; in general this is incompatible with treating heap-dependent expressions as purely syntactic entities which can be moved around amongst subformulas freely, as we would if we wanted a substitution property for such expressions. In particular, the meaning of a heap-dependent expression can differ in different positions in a formula, depending on its nesting under $\wand$ and $\imp$ connectives\footnote{One can compare with the analogous situation in standard separation logic: an SL formula such as $\pointsto{x.f}{\pi}{u} * (\pointsto{x.f}{\pi}{v} \imp u=v)$ is not actually valid in traditional intuitionistic separation logic semantics for the same reasons; the semantics of the implication connective includes the concept of ``adding on'' new access to $x.f$ when evaluating the implication.}. For this reason, if we wanted such a property, we would need to restrict the uses of $\wand$ and $\imp$ to enable the substitution of expressions with heap dependencies.
To illustrate this, we define a class of formulas, that are \emph{substitutable}. These are the formulas that only contain pure formulas on the left of $\wand$ and $\imp$.
\newcommand{\subs}[1]{\mathsf{subst}(#1)}
\begin{definition}[Substitutable formulas]
We define a formula as \emph{substitutable}, $\subs{A}$, by
\[
\begin{array}{ll}
\subs{A_1 \multimap A_2}
&\iff \subs{A_1} \land \subs{A_2} \text{ and }A_1 \text{ is pure}\\
&\hfill
\quad(\text{where } \multimap\ \in \{\wand,\imp\}) \\
\subs{A_1 \circ A_2}
&\iff \subs{A_1} \wedge \subs{A_2}  \qquad(\text{where }\circ \in \{\lor,\land,*\}) \\
 \subs{\exists x.\; A}  & \iff
  \subs{A}
\\
\subs{\acc(E.f,`p)} & \iff
\subs{E=E'}  \iff
\subs{\pointsto{E.f}{`p}{E'}} \iff \text{always} \\

\end{array}
\]
\end{definition}

\noindent As substitutable formulas only have pure formulas on the left of $\wand$ and $\imp$, there is a single heap that is used to evaluate the entire formula.
\begin{lemma}\label{substitution}
If $\subs{A}$, then
\[
(\exists x. x = E * A)  \isequiv A[E/x]
\]
\end{lemma}
\begin{proof}
We prove
\[
\Heap,\Perm,\Env \modelsISL A[E/x]
\Leftrightarrow
\Heap,\Perm,\Env[x \mapsto \semCE{E}{\Heap,\Env}] \modelsISL A
\]
by straightforward induction on the $\subs{A}$ predicate. The $\wand$ and $\imp$ cases use that for pure formulas they behave like boolean conditionals (Lemma \ref{lemma:boolean}).  We reuse the expression substitutability proof from previous lemma.
\end{proof}

In Section \ref{sec:mapping}, we will present an encoding from the SL fragment to the IDF fragment of our logic, which preserves semantics. A natural question to ask is, can we encode back from IDF to SL, at least for those IDF assertions which are self-framing? In general, it is surprisingly difficult to define a suitable syntactic translation. A tempting approach is to convert all $\acc(x.f,`p)$ assertions into $\pointsto{x.f}{`p}{v}$ for some fresh logical variable $v$, and then to replace any heap-dependent expressions $x.f$ with $v$ elsewhere in the assertion. But this approach fails in two ways: firstly, it does not deal correctly with aliasing. The criteria for an IDF to be self-framing take account of constraints imposed by the assertion itself; for example, $\acc(x.f) * x = y * y.f = 4$ is self-framing. This makes a syntactic replacement of heap-dependent expressions challenging. Furthermore, the correctness of the replacement of all heap-dependent expressions with logical variables, depends on a substitution property holding for such expressions. As discussed above, this does not hold for the general logic; the meaning of a heap-dependent expression is actually fixed by the ``closest scoped'' occurrence of a permission to that location, with respect to implications and wands; in the presence of aliasing this is hard to determine.

For the subsyntaxes of these logics typically supported by tools, which generally only allow for pure assertions on the left of $\imp$ formulas (and do not support $\wand$ in general), we do get a substitution property, as shown above. However, the problem of correctly handling aliasing between heap locations when translating heap-dependent expressions, still seems to make defining a correct syntactic translation challenging.




\section{Verification Conditions}
\label{sec:vc}
In this section, we precisely connect the semantics of our assertion language with Chalice.
Chalice does not provide a direct model for its assertion language. It instead defines the semantics of assertions using the weakest pre-condition semantics of the commands \inhale{} and \exhale{}.  We show that this semantics precisely corresponds with the semantics in \SLE{}.

\subsection{Chalice}
Chalice is defined by a translation into Boogie2~\cite{LeinoBoogie2}, which generates  verification conditions on a many-sorted classical logic with first-order quantification.  It has sorts for mathematical maps, which are used by Chalice to encode both the heap and the permission mask.  We use $\phi$ to range over formulas in this logic, and $\sigma \modelsFO \phi$ to mean $\phi$ holds in the standard semantics of first-order logic given the interpretation of free variables $\sigma$. Similarly, $\modelsFO \phi$ means that $\phi$ holds in all such interpretations.

The definitions throughout this section generate expressions that have these two specific free variables: $\heapvar$ for the current heap, and $\permvar$ for the current permission mask.  Thus, $\heapvar[x,f] = 5$ means that in the current heap the variable $x$'s field named $f$ contains the value $5$.  In the assertion logic, this corresponds to $x.f=5$, in which the heap access is implicit.

To define the verification conditions for Chalice, we need to be able to translate expressions into the underlying logic using access to the map $\heapvar$. We can provide a syntactic translation from the Chalice assertion logic into the first-order logic.
\begin{definition}
We translate expressions that implicitly access the heap into expressions that explicitly access the heap as follows:
\[
\begin{array}{l}
\toSOE{x} = x \qquad
\toSOE{\Null} = \Null\qquad
\toSOE{E.f} = \heapvar[\toSOE{E}, f]\\
\end{array}
\]
we translate boolean expressions as:
\[
\begin{array}{@{}r@{}}
\toSO{B_1 * B_2} =
		\toSO{B_1}
\land	\toSO{B_2}
\;\quad
\toSO{E = E'} = \toSOE{E} = \toSOE{E'}
\;\quad
\toSO{E \neq E'} = \toSOE{E} \neq \toSOE{E'}
\\

\end{array}
\]\smallskip
\end{definition}

\noindent First, we must show some basic facts about the properties of Chalice assertions (cf. Definition \ref{defn:chalicesyntax}): every Chalice boolean expression is pure, and every Chalice assertion is supported.
\begin{lemma}
Every Chalice boolean expression $B$ is pure.
\begin{proof}
By trivial induction on $B$.
\end{proof}
\end{lemma}

\begin{lemma}\label{chalicesupported}
Every Chalice formula $p$ is supported.
\proof
We must show
\[
\Heap,\Perm,\Env \modelsISL p\ 
\land\ 
\Heap,\Perm',\Env \modelsISL p
\ \Rightarrow\ 
\Heap,\Perm \sqcap \Perm',\Env \modelsISL p
\]
We proceed by induction on $p$
\begin{desCription}
\item\noindent{\hskip-12 pt$p \equiv B$:}\ As $B$ is pure, we know that if it is satisfied in a state, it will also be satisfied with any alternative permission mask.

\item\noindent{\hskip-12 pt$p \equiv \acc(E.f,\pi)$:}\
$\Perm$ and $\Perm'$ must map the field location to $\pi$ or greater, therefore $\Perm\sqcap \Perm'$ will also map the field location to $\pi$ or greater.
\item\noindent{\hskip-12 pt$p \equiv p_1 * p_2$:}\
  Assume $\Heap,\Perm,\Env \modelsISL p_1 * p_2$, and  $\Heap,\Perm',\Env \modelsISL p_1 * p_2$.  Therefore there exist $\Perm_1$, $\Perm_2$, $\Perm_1'$ and $\Perm_2'$ such that $\Perm_1*\Perm_2 = \Perm$ and $\Perm_1'*\Perm_2' = \Perm'$ and $\Heap,\Perm_1,\Env \modelsISL p_1$ and $\Heap,\Perm_2,\Env \modelsISL p_2$ and $\Heap,\Perm_1',\Env \modelsISL p_1$ and $\Heap,\Perm_2',\Env \modelsISL p_2$. By induction, we know $\Heap,\Perm_1 \sqcap \Perm_1',\Env \modelsISL p_1$ and $\Heap,\Perm_2 \sqcap \Perm_2',\Env \modelsISL p_2$, and thus we can deduce that
$\Heap,(\Perm_1 \sqcap \Perm_1')*(\Perm_2 \sqcap \Perm_2'),\Env \modelsISL p_1 * p_2$.  We can show $(\Perm_1 \sqcap \Perm_1' )*(\Perm_2 \sqcap \Perm_2' ) \subseteq (\Perm_1 * \Perm_2) \sqcap (\Perm_1'* \Perm_2')$, which by Proposition~\ref{prop:weakening} proves the obligation.

\item\noindent{\hskip-12 pt$p \equiv B \imp p'$:}\ Case split on $\Heap,\emptyset, \Env \modelsISL B$.  If  $B$ is true, then using Lemma~\ref{lemma:boolean} the result follows directly by induction.  If $B$ is false, then using Lemma~\ref{lemma:boolean} we have $\Heap,\emptyset,\Env \modelsISL B \imp p$ as required.\qed
\end{desCription}
\end{lemma}

\noindent Chalice does not allow arbitrary formulas to be used as argument to \inhale{} and \exhale{}: it restricts the formulas to be self-framing.  Chalice does not use the semantic check from earlier, but instead uses a more-syntactic formulation that checks self-framing from left-to-right.  Note that this means that syntactic self-framing is not symmetric with respect to $*$.  For instance, $\acc(x.f,\pi) * x.f=5$ is syntactically self-framing, but $x.f=5 * \acc(x.f,\pi)$ is not.  Somewhat surprisingly this is required by the way the verification conditions are generated.
In Chalice the check is actually implemented by a Boogie program.  Here, we use the logic to define an equivalent condition\footnote{The end result of this section can be used to prove it is equivalent to verifying the Boogie program that Chalice would generate.}.

\begin{definition}[Syntactic Self-Framing]
We define a condition $\framed{E}$ to express that all the fields mentioned in $E$ are accessible.
\[
\begin{array}{r@{\iff}l}
\framed{E.f} & \framed{E} \land \acc(E.f,\_)  \\
\framed{x} &  \mathsf{True}\\
\framed{\Null} &  \mathsf{True}
\end{array}
\]
We lift this to boolean expressions as
\[
\begin{array}{@{}r@{\iff}l@{}}
\framed{E=E'}
   & \framed{E} \land \framed{E'}\\
\framed{E\neq E'}
   & \framed{E} \land \framed{E'}\\
\framed{B_1 * B_2}
   & \framed{B_1} \land (B_1 \imp \framed{B_2})\\
\end{array}
\]%
We lift this to formulas as
\[
\begin{array}{@{}r@{\iff}l@{}}
\framed{B \imp p}
   & \framed{B} \land (B \imp \framed{p})\\
\framed{\acc(E.f,\pi)}
   & \framed{E}\\
\framed{p_1 * p_2}
   & \framed{p_1} \land
        (p_1 \wand \framed{p_2})
\end{array}
\]
Note that when we check that $p_2$ is framed in $p_1*p_2$, we can use the assertion $p_1$; these checks do not treat $*$ as commutative.

A formula, $p$, is syntactically self-framing, if and only if $\modelsISL\framed{p}$.\smallskip
\end{definition}

\noindent We prove some basic facts about $\framed{E}$ and $\framed{B}$: (1) in any state in which $\framed{E}$ holds, changing the value at any locations without permissions does not affect $E$'s evaluation; (2) $\framed{E}$ is (semantically) self-framing; (3) in any state in which $\framed{B}$ holds, changing the value at any locations without permissions does not affect $B$'s evaluation.

\addprop{lemma}{}{lem:sframed:basic}{\mbox{}\em
\begin{enumerate}[\em(1)]
\item
If $\Heap,\Perm,\Env \modelsISL \framed{E}$, and $\Heap' \agrees{\Perm} \Heap$
then $\semSLE{E}{\Heap,\Env} = \semSLE{E}{\Heap',\Env}$.

\item $\framed{E}$ is self-framing

\item
If $\Heap,\Perm,\Env \modelsISL \framed{B}$, and $\Heap' \agrees{\Perm} \Heap$
then $\Heap,\Perm,\Env \modelsISL B$ if and only if $\Heap',\Perm,\Env \modelsISL B$.

\item $\framed{B}$ is self-framing.
\end{enumerate}
}{}{
\mbox{}
\begin{enumerate}[(1)]
\item  Follows by straightforward induction on $E$.

\item Follows by induction on $E$, and using previous property.  The base cases hold trivially. For the inductive case ($E.f$), we assume
\[
\begin{array}{l l l}
\Heap,\Perm,\Env \modelsISL \framed{E}& \qquad
\Perm[\semSLE{E}{\Heap,\Env}, f] \geq \pi & \qquad
\Heap \agrees{\Perm} \Heap'
\end{array}
\]
and need to show that
\[
\begin{array}{l l }
\Heap',\Perm,\Env \modelsISL \framed{E} & \qquad
\Perm[\semSLE{E}{\Heap',\Env}, f] \geq \pi\\
\end{array}
\]
The first part follows from the inductive hypothesis.  The second part follows as we know $\semSLE{E}{\Heap,\Env} = \semSLE{E}{\Heap',\Env}$ by the previous part of the lemma.

\item By induction on $B$.  The base cases hold trivially.  For the inductive case, assume
\[
\begin{array}{l}
\Heap,\Perm,\Env \modelsISL \framed{B_1}\\
\Heap,\Perm,\Env \modelsISL B_1 \imp \framed{B_2}\\
\Heap \agrees{\Perm} \Heap'\\
\end{array}
\]
By inductive hypothesis, we know
\[
\Heap,\Perm,\Env \modelsISL B_1
\iff
\Heap',\Perm,\Env \modelsISL B_1
\]
We case split on whether or not $B_1$ holds. For the first case, assume $\Heap,\Perm,\Env \modelsISL B_1$. Therefore, by Lemma~\ref{lemma:boolean} we know
\[
\Heap,\Perm,\Env \modelsISL \framed{B_2}
\]
and thus, by inductive hypothesis
\[
\Heap,\Perm,\Env \modelsISL B_2
\iff
\Heap',\Perm,\Env \modelsISL B_2
\]
Hence, we know
\[
\Heap,\Perm,\Env \modelsISL B_1 * B_2
\iff
\Heap',\Perm,\Env \modelsISL B_1 * B_2
\]
as required.

For the second case, assume $\Heap,\Perm,\Env \notmodelsISL B_1$. Therefore
\[
\Heap',\Perm,\Env \notmodelsISL B_1
\]
and thus we know
\[
\Heap,\Perm,\Env \modelsISL B_1 * B_2
\iff
\Heap',\Perm,\Env \modelsISL B_1 * B_2
\]
as required.

\item  By induction on $B$. The base cases follow directly from previous parts of this lemma. For the inductive case, we assume
\[
\begin{array}{l}
\Heap,\Perm,\Env \modelsISL \framed{B_1}\\
\Heap,\Perm,\Env \modelsISL B_1 \imp \framed{B_2}\\
\Heap \agrees{\Perm} \Heap'
\end{array}
\]
and we seek to prove
\[
\begin{array}{l}
\Heap',\Perm,\Env \modelsISL \framed{B_1}\\
\Heap',\Perm,\Env \modelsISL B_1 \imp \framed{B_2}\\
\end{array}
\]
The first obligation follows by inductive hypothesis.  Using Lemma~\ref{lemma:boolean}, we can assume $\Heap',\Perm,\Env \modelsISL B_1$, and must prove $\Heap',\Perm,\Env \modelsISL \framed{B_2}$.
Thus, by previous part, we know $\Heap,\Perm,\Env \modelsISL B_1$, and by Lemma~\ref{lemma:boolean} we know
\[
\Heap,\Perm,\Env \modelsISL \framed{B_2}\\
\]
By inductive hypothesis, we obtain
\[
\Heap',\Perm,\Env \modelsISL \framed{B_2}\\
\]
as required.
\end{enumerate}
}

%
%

\noindent The key property we require of the $\framed{p}$ definition is that it allows a wand ($\wand$) of a separating conjunction, to be considered as a sequence of wands.

\addprop{lemma}{}{lem:framed:approx}
{\mbox{}\em
\begin{enumerate}[\em(1)]
\item
$
\framed{p_1}
\land
((p_1 * p_2) \wand p)
\modelsISL
p_1 \wand (p_2 \wand p)
$
\item
$
\framed{p_1}
\land
(p_1 \wand (p_2 \wand p))
\modelsISL
(p_1 * p_2) \wand p
$
\end{enumerate}
}{
This proof follows from Proposition~\ref{prop:stdproperties}.\ref{prop:bullet:uncurrying}.a and~\ref{prop:stdproperties}.\ref{prop:bullet:uncurrying}.b, Lemma~\ref{chalicesupported}, and showing
\[
\forall p. \;  \semEl{\framed{p}} \subseteq \deframed{p}
\]
This is proved by induction on $p$.
}{
We break this proof into two steps. First we prove that the $\framed{p_1}$ condition has a semantic meaning in terms of $\deframed{p_1}$, and then show this semantic meaning allows the restructuring of the assertion.
That is, we show that
\begin{equation}
\forall p. \;  \semEl{\framed{p}} \subseteq \deframed{p} \label{framed:proof1}
\end{equation}
and then show
\begin{equation}
\deframed{p_1}
\cap
\semEl{(p_1 * p_2) \wand p}
\subseteq
\semEl{p_1 \wand (p_2 \wand p)}
\label{framed:proof2}
\end{equation}
and
\begin{equation}
\deframed{p_1}
\cap
\semEl{p_1 \wand (p_2 \wand p)}
\subseteq
\semEl{(p_1 * p_2) \wand p}
\label{framed:proof3}
\end{equation}

To prove \eqref{framed:proof2} we use Proposition~\ref{prop:stdproperties}(\ref{prop:bullet:uncurrying})(b) and Lemma~\ref{chalicesupported}, and  \eqref{framed:proof3} is just a restatement of Proposition~\ref{prop:stdproperties}(\ref{prop:bullet:currying})(a).

To prove \eqref{framed:proof1} we use induction on $p$.
\begin{desCription}
\item\noindent{\hskip-12 pt($p \equiv \acc(E,f,\pi)$)}\ This requires that we prove
\[
  \semEl{\framed{E}} \subseteq \deframed{\acc(E.f,\pi)}
\]
By expanding the definition of $\deframed{\acc(E.f,\pi)}$ and the semantics of $\acc(E.f,\pi)$ we can assume
\[
\begin{array}{@{}c @{}}
\Heap,\Perm,\Env \modelsISL \framed{E} \qquad\quad
\Heap \agrees{\Perm} \Heap' \quad\qquad
\Perm \perp \Perm'  \\
\Perm' [\semSLE{E}{\Heap',\Env},f]  \geq \pi  \qquad \quad
\Heap' \agrees{\Perm * \Perm'} \Heap''
\end{array}
\]
and are required to prove $
\Perm'[\semSLE{E}{\Heap'',\Env},f] \geq \pi
$.
By definition of $\agrees{}$, we can get $\Heap' \agrees{\Perm} \Heap''$, and thus use Lemma~\ref{lem:sframed:basic}, to give $\semSLE{E}{\Heap',\Env} = \semSLE{E}{\Heap'',\Env}$ as required.

\item\noindent{\hskip-12 pt($p\equiv B$)}\ This case requires that we prove
\[
\semEl{\framed{B}} \subseteq\deframed{B}
\]
By expanding the definition of $\deframed{B}$ we can assume
\[
\begin{array}{@{}c  @{}}
\Heap,\Perm,\Env \modelsISL \framed{B} \qquad
\Heap \agrees{\Perm} \Heap' \qquad
\Perm \perp \Perm'  \\
\Heap',\Perm',\Env \modelsISL B  \qquad
\Heap' \agrees{\Perm * \Perm'} \Heap''
\end{array}
\]
and are required to prove $\Heap'',\Perm',\Env \modelsISL B$.  By definition of $\agrees{}$, we know
$\Heap' \agrees{\Perm} \Heap''$, and thus use Lemma~\ref{lem:sframed:basic} to give $\Heap',\Perm',\Env \modelsISL B \iff \Heap'',\Perm',\Env \modelsISL B$ as required.

\item\noindent{\hskip-12 pt($p \equiv p_1 * p_2$)}\
We assume
\[\begin{array}{l}
\semEl{\framed{p_1}} \subseteq \deframed{p_1}\\
\semEl{\framed{p_2}} \subseteq \deframed{p_2}
\end{array}\]
and, expanding the definition of $\framed{p_1*p_2}$, we must show
\[
\semEl{\framed{p_1} \land (p_1 \wand \framed{p_2}) }
   \subseteq \deframed{p_1 * p_2}
\]
We can assume, by expanding the definition of $\deframed{p_1 * p_2}$, and the definition of the semantics of $*$:
\[
\begin{array}{l}
\Heap,\Perm,\Env \modelsISL \framed{p_1}\\
\Heap,\Perm,\Env \modelsISL p_1 \wand \framed{p_2}\\
\Heap \agrees{\Perm} \Heap'\\
\Perm_1 \perp \Perm \land \Perm_2 \perp \Perm \land \Perm_1 \perp \Perm_2\\
\Heap',\Perm_1,\Env \modelsISL p_1\\
\Heap',\Perm_2,\Env \modelsISL p_2\\
\Heap' \agrees{\Perm * \Perm_1 * \Perm_2} \Heap'' \\
\end{array}
\]
and we are left with proving:
\[
\Heap'',\Perm_1 * \Perm_2 ,\Env\modelsISL p_1 * p_2
\]
Let $\Heap_1 = \condheap{\rds{\Perm} \cup \overline{\rds{\Perm_1}}}{\Heap}{\Heap'}$. By inductive hypothesis, we know $(\Heap,\Perm,\Env) \in \deframed{p_1}$, and thus we know $\Heap_1,\Perm_1,\Env \modelsISL p_1$.

By Lemma~\ref{lemma:newlemmaone}, we can prove that there exist $\Perm_1', \Perm_1''$ such that $\Perm_1' * \Perm_1'' = \Perm_1$ and $\extend{\Heap_1}{\emptyset}{\Env}{\Perm_1'}{p_1}$.  Therefore, we know $\Heap_1,\Perm_1' * \Perm, \Env \modelsISL \framed{p_2}$, and thus $(\Heap_1,\Perm_1' * \Perm, \Env) \in \deframed{p_2}$. As $\Heap \agrees{\Perm} \Heap'$, we know
$\Heap_1 \agrees{\Perm} \Heap'$ by construction.  Moreover, by construction we know $\Heap_1 \agrees{\rds{\Perm_1} \setminus \rds{\Perm}} \Heap'$, which with the previous gives $\Heap_1 \agrees{\Perm_1} \Heap'$. Thus, we know $\Heap_1 \agrees{\Perm_1 * \Perm} \Heap'$, which we can weaken to
$\Heap_1 \agrees{\Perm_1' * \Perm} \Heap'$.  Thus, we know
$(\Heap',\Perm_1' * \Perm, \Env) \in \deframed{p_2}$

As all assertions are weakening-closed (cf. Proposition \ref{prop:weakening}), we have $\Heap', \Perm_1'' * \Perm_2, \Env \modelsISL p_2$.
We know $\Perm_1'' * \Perm_2 \perp \Perm_1' * \Perm$, thus using $\deframed{p_2}$, we know $\Heap'',\Perm_1'' * \Perm_2,\Env \modelsISL p_2$.

By $(\Heap,\Perm,\sigma) \in \deframed{p_1}$ and $\Heap \agrees{\Perm} \Heap'$, we know $(\Heap',\Perm,\sigma) \in \deframed{p_1}$.
By weakening assumption we know $\Heap' \agrees{\Perm * \Perm_1'} \Heap''$.
As $\Perm_1' \perp \Perm$, we know
$\Heap'', \Perm_1', \Env \modelsISL p_1$
and thus
\[
\Heap'',\Perm_1 * \Perm_2 ,\Env\modelsISL p_1 * p_2
\]
as required.

\item\noindent{\hskip-12 pt($p \equiv B \imp p'$)}\ This case requires
\[
\semEl{\framed{B} \land (B \imp \framed{p'})} \subseteq  \deframed{B \imp p'}
\]
By expanding the definition of $\deframed{B \imp p'}$ and using Lemma~\ref{lemma:boolean}, we can assume
\[
\begin{array}{l}
\Heap,\Perm,\Env \modelsISL \framed{B}\\
\Heap,\Perm,\Env \modelsISL B   \Rightarrow \Heap,\Perm,\Env \modelsISL \framed{p'}\\
\Heap \agrees{\Perm} \Heap' \\
\Perm \perp \Perm'\\
\Heap', \Perm', \Env \modelsISL B \Rightarrow \Heap',\Perm',\Env \modelsISL p' \\
\Heap'' \agrees{\Perm * \Perm'} \Heap'\\
\Heap'', \Perm', \Env \modelsISL B
\end{array}
\]
and must prove
\[
\Heap'', \Perm', \Env \modelsISL p'
\]
By the $\framed{B}$ assumption and by Lemma~\ref{lem:sframed:basic}, and since  $B$ is pure, we can obtain that
$\Heap', \Perm', \Env \modelsISL B$ and $\Heap, \Perm, \Env \modelsISL B$.
Therefore, we know
\[
\begin{array}{l}
\Heap,\Perm,\Env \modelsISL \framed{p'}\\
\Heap',\Perm',\Env \modelsISL p'\\
\end{array}
\]
and must show
\[
\Heap'',\Perm',\Env \modelsISL p'\\
\]
which follows directly by definition of $\framed{p'}$ and the inductive hypothesis.

\end{desCription}
}


\begin{figure}[t]
\[
\begin{array}{l}
\wpCh{\exhale(B), \phi} = \wpCh{\assert\ \toSO{B}, \phi}
\\[0.5em]
\wpCh{\exhale(p_1 * p_2),\phi}
=\wpCh{\exhale(p_1);\exhale(p_2), \phi}
\\[0.5em]
\wpCh{\exhale(\acc(E.f, \pi)),\phi}
\\\quad =
\wpCh{
\assert(\permvar[\toSOE{E}{},f]) \geq \pi; \permvar[\toSOE{E}{},f] := \permvar[\toSOE{E}{},f] - \pi
,\phi}
\\[0.5em]
\wpCh{\exhale(B \imp p),\phi}
=\wpCh{(\assume(\toSO{B});\exhale(p)) + \assume(\lnot\toSO{B}), \phi}
\\[0.5em]

%
\wpCh{\inhale(B), \phi}
= \wpCh{\assume\ \toSO{B}, \phi}
\\[0.5em]

\wpCh{\inhale(p_1 * p_2), \phi}
=
\wpCh{\inhale(p_1);\inhale(p_2), \phi}
\\[0.5em]

\wpCh{\inhale(\acc(E.f,`p)),\phi}
\\
\quad =
\wpCh{\assume(\permvar[\toSOE{E},f] = \none);
\permvar[\toSOE{E},f] := \pi;
\havoc(\heapvar[\toSOE{E},f]),\phi}
\\\qquad{} \land
\wpCh{\assume(\none < \permvar[\toSOE{E},f] \leq \full - \pi);
\permvar[\toSOE{E},f] {{+}{=}} \pi,\phi}

\\[0.5em]
\wpCh{\inhale(B \imp p),\phi}
=\wpCh{(\assume(\toSO{B});\inhale(p)) + \assume(\lnot \toSO{B}), \phi}
\\[0.5em]
%
%
\end{array}
\]
where
\[
  \begin{array}{l}
    \wpCh{\permvar[o,f] := x, \phi} = \phi[upd(\permvar,(o,f),x)/\permvar]\\
    \wpCh{\havoc(\heapvar[x,f]),\phi} = \phi[upd(\heapvar,(x,f),z)/\heapvar]   \qquad \text{fresh\ z}\\
    \wpCh{\assume\ \phi', \phi} = \phi' \imp \phi\\
    \wpCh{\assert\ \phi', \phi} = \phi' \land \phi\\
    \wpCh{C_1;C_2, \phi} = \wpCh{C_1, \wpCh{C_2, \phi}}\\
   \wpCh{C_1+C_2, \phi} = \wpCh{C_1, \phi} \land \wpCh{C_2, \phi}\\

  \end{array}
\]
where $upd(a,b,c)[b] = c$ and $upd(a,b,c)[d] = a[d]$ provided $d\neq b$.
\caption{Abridged weakest pre-condition semantics for Chalice~\cite{LeinoM09}}\label{fig:wp:chalice}
\end{figure}

We can now provide the definitions of the weakest pre-conditions of the commands \inhale{} and \exhale.
In Figure~\ref{fig:wp:chalice}, we present the weakest pre-conditions of commands in Chalice from~\cite{LeinoM09}. We write $\wpCh{C,\phi}$ for the weakest pre-condition of the command $C$ given the post-condition $\phi$.  Chalice models the inhaling and exhaling of permissions by mutating the permission mask variable.   To exhale an equality (or any formula not mentioning the permission mask) we simply assert that it must be true.  This does not need to modify the permission mask.
To exhale $p*q$, first we exhale $p$ and then $q$.
When an accessibility predicate is exhaled, first we check that the permission mask contains sufficient permission, and then we remove the permission from the mask.

To inhale an equality is simply the same as assuming it. To inhale $p*q$, we first inhale $p$ and then $q$.  There are two cases for inhaling a permission: (1) we don't currently have any permission to that location; and (2) we do currently have permission to that location.  The first case proceeds by adding the permission, and then havocing the contents of that location; that is, making sure any previous value of the variable has been forgotten.  The second case simply adds the permission to the permission mask.


\subsection{Relationship}
In the rest of this section, we show that the verification conditions (VCs) generated by Chalice are equivalent to those generated by $\SLE$.  We focus on the \inhale\ and \exhale\ commands, as these represent the semantics of the Chalice assertion language.  By showing the equivalence, we show that our model of \SLE\ is also a model for Chalice.

\newcommand{\wpSL}[1]{wp_{\mathrm{sl}}(#1)}
We write $\wpSL{C,A}$, to be the weakest pre-condition in \SLE\ of the formula $A$ with respect to the command $C$. We treat \inhale\ and \exhale\ as the multiplicative versions of \assume\ and \assert\ (see \S\ref{sec:assumeassert}), and thus have the following weakest pre-conditions:
\[
\begin{array}{l}
\wpSL {\exhale(A_1), A_2 } = A_1 * A_2
\qquad
\wpSL{\inhale(A_1), A_2 } = A_1 \wand A_2
\end{array}
\]

\newcommand{\equivvc}[1]{equiv(#1)}
\noindent Our core result is to show that both $\inhale$ and $\exhale$ have equivalent VCs in the two approaches.

\newcommand{\interp}[1]{\mathsf{interp}(#1)}

First we must extend the first order logic we are considering to additionally contain a proposition to represent separation logic assertions.
\begin{definition}[$\interp{A}$]\label{defn:interp} We extend the many sorted first order logic with an additional atomic proposition $\interp{A}$, which represents the interpretation of an arbitrary TPL formula in first order logic.
\[
\Env,\heapvar\mapsto \Heap, \permvar \mapsto \Perm \modelsFO \interp{A}
\quad \iff \quad  \Heap,\Perm,\Env \modelsISL A
\]
\end{definition}\vskip5 pt

Note, this definition is not required by Chalice, but it allows us to express our proof by induction on the structure of the formula, by providing a single logic in which we can describe both the Chalice VCs and the TPL judgements.

\begin{definition}[$\equivvc{C}$]
We define the VCs of a command as equivalent in both systems, $\equivvc{C}$, iff for every \SLE\ assertion $A$, we have
\[
\modelsFO \interp{\wpSL{C,A}}
	\iff \wpCh{C, \interp{A}}
\]
\end{definition}


The key to showing that our semantics for \SLE{} correctly reflects that of Chalice is to show that the VCs generated for the \inhale\ and \exhale\ commands are equivalent.  The \exhale\ is straightforward.
\begin{lemma}\label{lemma:exhale:vcequiv}
$\forall p.\;\equivvc{\exhale\;p}$
\begin{proof}
By induction on $p$.
\end{proof}
\end{lemma}

The proof for \inhale\ is more involved.  This depends on the inhaled formula being syntactically self-framing. \newcommand{\isframed}[1]{\texttt{assertFrm}(#1)}
\newcommand{\asssl}[1]{\texttt{assertTPL}(#1)}
We define two Boogie commands, $\asssl{A}$ to assert that a TPL assertion must hold at this point in the execution, and $\isframed{p}$ to assert that a Chalice formula will be framed if it is inhaled in the current state.
\begin{definition}\mbox{}
\begin{iteMize}{$\bullet$}
\item $\asssl{A} = \assert(\interp{A})$
\item $\isframed{p} = \asssl{\framed{p}}$
\end{iteMize}
\end{definition}

We lift the ability to curry and uncurry $\wand$ into the VC world.  This is required to allow us to prove that an inhale of a $*$ can be replaced by two inhales, as Chalice does.
\begin{lemma}\mbox{}\label{lemma:vc:rearrange}
\[
\wpCh{\isframed{p_1},\interp{p_1 \wand (p_2 \wand A)}}
\iff
\wpCh{\isframed{p_1},\interp{(p_1 * p_2) \wand A}}
\]
\end{lemma}
\begin{proof}
The left to right direction follows using Lemma~\ref{lem:framed:approx}:
\[
\begin{array}{@{\Rightarrow\;}l@{}}
\multicolumn{1}{@{}l@{}}{ \wpCh{\isframed{p_1}, \interp{p_1 \wand (p_2 \wand A)}}}
\\
\wpCh{\isframed{p_1},\interp{\framed{p_1}} \land \interp{p_1 \wand (p_2 \wand A)}}
\\
\wpCh{\isframed{p_1},\interp{\framed{p_1} \land (p_1 \wand (p_2 \wand A))}}
\\
\wpCh{\isframed{p_1},\interp{\framed{p_1} \land ((p_1 * p_2) \wand A)}}
\end{array}
\]
The reverse direction follows similarly.
\end{proof}

We want to show that if $p$ is syntactically self-framing, then $\inhale\ p$ is equivalent in both approaches.  However, we need to prove a stronger fact that accounts for the permissions we may have inhaled so far.  In particular, as $\inhale\ p_1*p_2$ is implemented by first inhaling $p_1$ and then $p_2$, when we consider inhaling $p_2$ it need not be self-framing.  However, the context will have inhaled sufficient permissions that it is framed in that context.  We prove that the VCs are equivalent in a context in which the inhale is framed.

 We consider states in which, if we extend the environment to satisfy $p$, then $p$ will be framed. In these states, asserting the formula $p \wand A$ and then inhaling $p$, is equivalent to inhaling $p$, and then asserting that $A$ must hold.
\begin{lemma}
\label{lemma:inhale:exchange}
\label{lemma:inhale:vcequiv}
\[
\begin{array}{l}
\wpCh{\isframed{p};\asssl{p \wand A}; \inhale\ p, \phi}
\\\qquad\quad{}\iff
\wpCh{\isframed{p};\inhale\ p; \asssl{A}, \phi}

\end{array}
\]
\end{lemma}
\proof
We abbreviate $\isframed{p}$ to
\renewcommand{\isframed}[1]{\texttt{assFrm}(#1)}
$\isframed{p}$, and $\asssl{A}$ to
\renewcommand{\asssl}[1]{\texttt{assTPL}(#1)}
$\asssl{A}$.

By induction on $p$. We first consider the $*$ case:
\[
\begin{array}{@{\iff}l@{}}
\multicolumn{1}{@{}l@{}}{
  \wpCh{\isframed{p_1*p_2};\asssl{(p_1 * p_2) \wand A}; \inhale\ p_1*p_2, \phi}
}
\\
  \wpCh{\isframed{p_1 * p_2};\isframed{p_1};\asssl{(p_1 * p_2) \wand A};
	      \inhale\ p_1* p_2, \phi}
\\
  \wpCh{\isframed{p_1*p_2};\isframed{p_1};\asssl{p_1 \wand p_2 \wand A};
	      \inhale\ p_1; \inhale\ p_2, \phi}
\\
  \wpCh{\isframed{p_1}; \asssl{p_1\wand \framed{p_2}};
	      \inhale\ p_1; \asssl{p_2 \wand A};\inhale\ p_2, \phi}
\\
  \wpCh{\isframed{p_1};
	      \inhale\ p_1; \isframed{p_2};\asssl{p_2 \wand A};\inhale\ p_2, \phi}
\\
  \wpCh{\isframed{p_1};
	      \inhale\ p_1; \isframed{p_2};\inhale\ p_2;\asssl{A}, \phi}
\\
  \wpCh{\isframed{p_1}; \asssl{p_1\wand \framed{p_2}};
	      \inhale\ p_1; \inhale\ p_2;\asssl{A}, \phi}
\\
\wpCh{\isframed{p_1*p_2}; \inhale\ p_1 * p_2; \asssl{A}, \phi}
\\
\end{array}
\]

For the access permission case, we can subdivide this into three further cases, (1) where the model contains no permission for $E.f$; (2) where the model contains more than $\none$, and less than or equal to $\full - \pi$ permission $E.f$; and (3) where the model contains more that $\full - \pi$ permission for $E.f$.  The third case is trivial, so we just present the first two.  First we consider the case $ \permvar[\toSOE{E}.f]=\none$ :
\[
\begin{array}{@{\iff}l@{}}
\multicolumn{1}{@{}l@{}}{
  \wpCh{\asssl{\acc(E.f,\pi) \wand A};
		\inhale\ \acc(E.f,\pi),
             \phi}}
\\
  \wpCh{
		\asssl{\acc(E.f,\pi) \wand A};
		\permvar[\toSOE{E},f] = \pi; \havoc\ \heapvar[\toSOE{E},f],
             \phi}
\\
\wpCh{\inhale\ \acc(E.f,\pi); \asssl{A}, \phi}
\\
\end{array}
\]
Second, we consider the case: $0<\permvar[\toSOE{E}.f]\leq 1-\pi$,
\[
\begin{array}{@{\iff}l@{}}
\multicolumn{1}{@{}l@{}}{
  \wpCh{\asssl{\acc(E.f,\pi) \wand A};
		\inhale\ \acc(E.f,\pi),
             \phi}}
\\
  \wpCh{
		\asssl{\acc(E.f,\pi) \wand A};
		\permvar[\toSOE{E},f] {{+}{=}} \pi,
             \phi}
\\
\wpCh{\inhale\ \acc(E.f,\pi); \asssl{A}, \phi}
\\
\end{array}
\]

Finally, we present the implication case.  We split this into two cases depending on whether the left of the implication holds.  First we assume $\toSO{B}$ holds in the current model:
\[
\begin{array}{@{\Leftrightarrow\;}l@{}}
\multicolumn{1}{@{}l}{\wpCh{\isframed{B\imp p}; \asssl{(B\imp p) \wand A}; \inhale(B \imp p),\phi}}\\

\wpCh{\isframed{B\imp p}; \asssl{ p \wand A}; \inhale\ p,\phi}\\

\wpCh{\isframed{B\imp p};\asssl{B \imp \framed{p}}; \asssl{p \wand A}; \inhale\ p,\phi}\\

\wpCh{\isframed{B \imp p};\isframed{p}; \asssl{p \wand A}; \inhale(p),\phi}\\

\wpCh{\isframed{B \imp p};\isframed{p};  \inhale(p); \asssl{A},\phi}\\

\wpCh{\isframed{B \imp p};  \inhale(p); \asssl{A},\phi}\\

\wpCh{\isframed{B \imp p};  \inhale(B \imp p); \asssl{A},\phi}\\

\end{array}
\]
and secondly, we assume $\toSO{B}$ does not hold in the model:
\[
\begin{array}{@{\Leftrightarrow\;}l@{}}
\multicolumn{1}{@{}l}{\wpCh{\isframed{B\imp p}; \asssl{(B\imp p) \wand A}; \inhale(B \imp p),\phi}}\\

\wpCh{\isframed{B\imp p}; \asssl{\textsf{true} \wand A}; \inhale(\textsf{true}),\phi}\\

\wpCh{\isframed{B\imp p}; \asssl{\textsf{true} \wand A},\phi}\\

\wpCh{\isframed{B\imp p}; \asssl{A},\phi}\\

\wpCh{\isframed{B\imp p}; \inhale(B \imp p); \asssl{A},\phi}
\rlap{\hbox to 105 pt{\hfill\qEd}}\\
\end{array}
\]

\begin{corollary}\label{corollary:inhale}
If $p$ is syntactically self-framing, then
$\equivvc{\inhale\;p}$.
\end{corollary}
\begin{proof}
By the previous lemma, we know:
\[
\begin{array}{l}
\wpCh{\isframed{p};\asssl{p \wand A}; \inhale\ p, \textsf{true}}
\\\qquad\iff
\wpCh{\isframed{p}; \inhale\ p; \asssl{A}, \textsf{true}}
\end{array}
\]
As $\wpCh{\inhale\ p,\textsf{true}}\equiv\textsf{true}$ and $\wpCh{\asssl{A},\textsf{true}} \equiv \interp{A}$, we know
\[
\wpCh{\isframed{p}, \interp{p \wand A}}
\iff
\wpCh{\isframed{p}; \inhale\ p, \interp{A}}
\]
As $p$ is syntactically self-framing we have
\[
\interp{p \wand A}
\iff
\wpCh{\inhale\ p, \interp{A}}
\]
By the definition of separation logic weakest pre-conditions, we have
\[
\interp{\wpSL{\inhale\ p, A}}
\iff
\wpCh{\inhale\ p, \interp{A}}
\]
as required.
\end{proof}

\begin{remark}
Without the syntactic self-framing requirement on inhales, it would be unsound to break $\inhale\ A _1* A_2$ into $\inhale\ A_1;\inhale\ A_2$.  In particular, in the Chalice semantics, the behaviour of
$
\inhale(A_1 * A_2)
$
and
$
\inhale(A_2 * A_1)
$
are different.  For instance, consider $\inhale(x.f = 3 * \acc(x.f))$ and $\inhale(\acc(x.f) * x.f = 3)$.
\[
\begin{array}{l}
\wpCh{\inhale(x.f = 3 * \acc(x.f)), \toSO{x.f=3}}
\iff
x.f \neq 3
\\
\wpCh{\inhale(\acc(x.f) * x.f = 3), \toSO{x.f=3}}
\iff
\mathsf{true}
\end{array}
\]
The translation given by Smans \emph{et al.}~\cite{Smans'09} does not suffer this problem as it does the analogue of $\inhale$ in a single step. However, it checks self-framing in a similar way, and thus would also rule out the first \inhale.
\end{remark}

\section{Mapping Separation Logic into Implicit Dynamic Frames}\label{sec:mapping}
We are now in a position to draw together our various results, and show that SL-based verification can be simulated using IDF and Chalice. The overall approach is to show that, if one calculates weakest pre-conditions for a program using SL specifications, then there is a corresponding translated program in which one uses IDF specifications, and can calculate Chalice weakest pre-conditions which turn out to be equivalent to the SL ones. Just as in Section \ref{sec:tsl}, we can use the projection of a total heap down to a permissions mask, to relate the evaluation of the two resulting assertions.

Just as in the preceding section, we focus our attention on inhale and exhale statements, since all commands which deal with changes to the footprint/permissions held in the state (e.g., method calls, fork/join of threads, acquire/release of locks) can be de-sugared down to these (other commands such as variable assignment can be treated uniformly in both worlds). Therefore, we aim to prove that the two different ways of calculating weakest pre-conditions produce equivalent results for both inhale and exhale statements.

Firstly, we state a few simple results which distribute equivalent assertions over various constructions.
\begin{lemma}[Distributing Equivalences]\label{lemma:distribute}
For all \SLE{} assertions $A_1$ and $A_2$ such that $A_1\isequiv A_2$, we have: \begin{enumerate}[\em(1)]
\item\label{lemma:distribute:parti} For all \SLE{} assertions $A_3$: $\wpSL{\inhale(A_1), A_3 } \isequiv \wpSL{\inhale(A_2), A_3 }$ and similarly $\wpSL{\exhale(A_1), A_3 } \isequiv \wpSL{\exhale(A_2), A_3 }$.
\item\label{lemma:distribute:partii} The first-order assertions $\interp{A_1}$ and $\interp{A_2}$ are equivalent.
\item\label{lemma:distribute:partiv} For all \SLE{} assertions $A_3$, the first-order assertions $\wpCh{\inhale(A_3),\interp{A_1}}$ and $\wpCh{\inhale(A_3),\interp{A_2}}$ are equivalent (and analogously for $\exhale(A_3)$).
\item\label{lemma:distribute:partiii} For all \SLE{} assertions $A_3$, the first-order assertion $\interp{\wpSL{\inhale(A_1), A_3 }}$ is equivalent to $\interp{\wpSL{\inhale(A_2), A_3 }}$. Similarly $\interp{\wpSL{\exhale(A_1), A_3 }}$ is equivalent to $\interp{\wpSL{\exhale(A_2), A_3 }}$.
\end{enumerate}
\begin{proof}
The first three parts follow straightforwardly from the corresponding definitions (and the fact that the definitions for weakest pre-conditions never inspect the post-condition). The fourth part is simply a combination of the first two.
\end{proof}
\end{lemma}

We now need to identify the fragment of separation logic which can be encoded into the syntax of Chalice (cf. Definition \ref{defn:chalicesyntax}). This syntax roughly corresponds with the syntaxes supported by most separation-logic-based tools (in particular, no points-to predicates are permitted on the left of implications; a very common restriction which avoids needing to implement the full technical complexity of the connective; cf. Lemma \ref{lemma:boolean}). In order to avoid introducing further meta-variables to our notation, we will reuse the notation for full separation logic ($a$ for assertions, $e$ for expressions), but will clarify explicitly when we mean the restricted assertion syntax defined here.
\begin{definition}[Restricted Separation Logic]
Expressions $e$, boolean expressions $b$ and \emph{restricted separation logic assertions} $a$ are given by the following syntax definitions:
\[
  \begin{array}{rcl}
   e &\ ::=\ & x \mid \Null \mid n\\
   b & ::= & e=e \mid e \neq e \mid b * b \\
   a & ::= & b \mid \pointsto{e.f}{`p}{e} \mid a * a \mid  b \imp a
          {}\mid  \exists v.\; \pointsto{e.f}{`p}{v} * a \\
  \end{array}
\]\smallskip
\end{definition}

\noindent We allow a restricted form of existential in the syntax.  It requires that the existential is witnessed by a particular field in the heap.  This restriction is often implicit in tools for separation logic that support existentials. Without this restriction tools are typically incomplete.

We can represent the separation logic points-to predicate in terms of the Chalice accessibility predicate and a (heap-dependent) equality.
\begin{proposition}\label{prop:mapping}
For all $e$,$f$,$e'$,$`p$ we have $\pointsto{e.f}{`p}{e'} \isequiv \acc(e.f,`p) * e.f=e'$.
\begin{proof}
Directly from the semantics.
\end{proof}
\end{proposition}
Thus, we define the obvious translation from restricted separation logic assertions to those of Chalice:
\begin{definition}[Mapping Restricted Separation Logic to Chalice]
We define a mapping \SLtoC{a} from restricted separation logic assertions to Chalice assertions (cf. Definition \ref{defn:chalicesyntax}), recursively as follows:
\[\begin{array}{rcl}
\SLtoC{b} &=& b\\
\SLtoC{\pointsto{e_1.f}{`p}{e_2}} &=& (\acc(e_1.f,`p) * e_1.f = e_2)\\
\SLtoC{a_1 * a_2} &=& \SLtoC{a_1} * \SLtoC{a_2}\\
\SLtoC{b \imp a} &=& b \imp \SLtoC{a}\\
\SLtoC{\exists v.\; \pointsto{e.f}{`p}{v} * a }
	&=& \acc(e.f,`p) * \SLtoC{a}[e.f/v]\\
\end{array}\]
\end{definition}
As the existential is witnessed by a particular heap location, in the translation to Chalice the existential can be eliminated by substituting the heap dependent expression.
The translation preserves the semantics of the original assertion, which is a simple generalisation of Proposition \ref{prop:mapping}:
\begin{lemma}[Mapping to Chalice Preserves Semantics]\label{lemma:fullmapping}
For all restricted separation logic assertions $a$, we have $a\isequiv \SLtoC{a}$.
\begin{proof}
By straightforward induction on the structure of $a$, using Definition \ref{defn:seplogictotalsemantics} and Proposition \ref{prop:mapping}.
The existential case uses Lemma~\ref{substitution}.
\end{proof}
\end{lemma}

We can now combine the results of this section to relate the two different notions of weakest pre-conditions in combination with the appropriate translation from one assertion syntax to the other.
\begin{theorem}[Weakest Pre-condition Calculations Equivalent]\label{theorem:wpsequiv}
For any restricted SL assertion $a$, and any \SLE{} assertion $A$, we have:
\[\begin{array}{rcl} \modelsFO \interp{\wpSL{\inhale(a),A}} &\iff& \wpCh{\inhale(\SLtoC{a}), \interp{A}}\\
&\textit{and}\\
\modelsFO \interp{\wpSL{\exhale(a),A}} &\iff& \wpCh{\exhale(\SLtoC{a}), \interp{A}}\\\end{array} \]
\begin{proof}
Consider the first case of the result (for inhale statements). By Lemma \ref{lemma:distribute}(\ref{lemma:distribute:partiii}) and Lemma \ref{lemma:fullmapping}, we have that:
\[\modelsFO \interp{\wpSL{\inhale(a),A}} \iff \interp{\wpSL{\inhale(\SLtoC{a}),A}} \]
By Corollary \ref{corollary:inhale} (noting that $\SLtoC{a}$ is a syntactically self-framing Chalice assertion), we also know that:
\[ \modelsFO \interp{\wpSL{\inhale(\SLtoC{a}),A}} \iff \wpCh{\inhale(\SLtoC{a}), \interp{A}} \]
Combining these two lines, we have the claimed result.

The case for exhale statements is analogous, using Lemma \ref{lemma:exhale:vcequiv} instead of Corollary \ref{corollary:inhale}.
\end{proof}
\end{theorem}
Finally, we can draw together these results with the main result of Section \ref{sec:tsl}, to show the equivalence of the two overall approaches.
\begin{corollary}[Verifying Restricted Separation Logic in Chalice]
For any restricted SL assertion $a$, and any (unrestricted) SL assertion $a'$, and any environment $\Env$, total heap $\Heap$, and permission mask $\Perm$, we have both:
\[\begin{array}{rcl}
(\restrict{\Heap}{\Perm}), \Env &\modelsSL& \wpSL{\exhale(a),a'} \\
&\Leftrightarrow& \\
 \Env,\heapvar\mapsto \Heap, \permvar \mapsto \Perm &\modelsFO& \wpCh{\exhale(\SLtoC{a}), \interp{\SLtoC{a'}}}\\
 \end{array}
\] and: \[
\begin{array}{rcl}
(\restrict{\Heap}{\Perm}), \Env &\modelsSL& \wpSL{\inhale(a),a'} \\
&\Leftrightarrow& \\\
 \Env,\heapvar\mapsto \Heap, \permvar \mapsto \Perm &\modelsFO& \wpCh{\inhale(\SLtoC{a}), \interp{\SLtoC{a'}}}\\
\end{array} \]

\begin{proof}
By Theorem \ref{theorem:wpsequiv} and Definition \ref{defn:interp}, we have:
\[\begin{array}{c} \Heap,\Perm,\Env \modelsISL \wpSL{\exhale(a),a'} \ \Leftrightarrow\ \Env,\heapvar\mapsto \Heap, \permvar \mapsto \Perm \modelsFO \wpCh{\exhale(\SLtoC{a}), \interp{a'}}\\
\textit{and}\\
\Heap,\Perm,\Env \modelsISL \wpSL{\inhale(a),a'} \ \Leftrightarrow\ \Env,\heapvar\mapsto \Heap, \permvar \mapsto \Perm \modelsFO \wpCh{\inhale(\SLtoC{a}), \interp{a'}}\\\end{array} \]

By Lemmas \ref{lemma:fullmapping} and \ref{lemma:distribute}(\ref{lemma:distribute:partii}), we have that the two assertions $\interp{a'}$ and $\interp{\SLtoC{a'}}$ are equivalent.
Therefore, by Lemma \ref{lemma:distribute}(\ref{lemma:distribute:partiv}) and Theorem \ref{thm:correctness}, we obtain the two desired equivalences.
\end{proof}
\end{corollary}


In this section, we have shown that the encoding of \inhale\ and \exhale\ into Boogie2 is equivalent to the separation logic weakest pre-condition semantics.  As a consequence, we have shown two things: (1) our model accurately reflects the semantics of Chalice's assertion language, and (2) a fragment of separation logic can be directly encoded into Chalice precisely preserving its semantics.

\section{Related Work}
\label{sec:relatedwork}
In this paper, we have provided a logic related to separation logic~\cite{IshtiaqOHearn'01,OHearnReynoldsYang'01}, which allows arbitrary expressions over the heap.  We have modified the standard presentation of an object-oriented heap for separation logic~\cite{parkinson_thesis} to separate the notion of access from value (and thereby also relate to implicit dynamic frames \cite{Smans_thesis}).
Most previous separation logics have combined these two concepts.  One notable exception is the separation logic for reasoning about Cminor~\cite{cminor}.   This logic also separates the ability to access memory, the mask, from the actual contents of the heap.  The choice in this work was to enable a reuse of a existing operational semantics for Cminor, rather than producing a new operational semantics involving partial states.   In the Cminor separation logic, they do not consider the definition of magic wand, or weakest pre-condition semantics, which is crucial for the connection with Chalice~\cite{LeinoM09}. Benton and Leperchey~\cite{benton-leperchey} also provide a logic for sequential program reasoning that uses total heaps and maps defining which locations can be accessed.

Smans' original presentation of IDF was implemented in a tool, VeriCool~\cite{Smans'09,Smans_thesis}. The results in this paper should also apply to the verification conditions generated by VeriCool.  In recent work, Smans \emph{et al.}~\cite{SmansJP10} describe an IDF approach as a separation logic. However, they do not present a model of the assertions, just the VCs of their analogues to inhale and exhale.  Hence, the work does not provide the strong connection between the VCs and the model of separation logic that we have provided. Vericool does however have a sound implementation of abstract predicates and pure functions (in fact, two different approaches; one for verification condition generation, as in \cite{Smans'09}, and one for symbolic execution, as in \cite{SmansJP10}). However, the approach for verification condition generation requires the formalisation of weakest pre-conditions in the presence of background axioms (used to define the meanings of predicates and pure methods). The use of these axioms cannot be summarised simply as part of a weakest pre-condition calculation, since the approach taken allows the prover to instantiate these axioms in an unbounded (and potentially non-terminating) way. Similarly, a comparison with a symbolic-execution-based approach would require more technical machinery than our current arguments based on first-order verification conditions. A comparison based on a runtime semantics for a language (as used to formalise soundness in \cite{Smans'09}) might work better than one dealing with a verification semantics, but this is beyond the scope of our work.

There have been many other approaches based on dynamic frames~\cite{Kassios'06,Kassios_thesis} to enable automated verification with standard verification tool chains; for instance,  Dafny~\cite{Rustan10dafny:an} and Region Logic~\cite{Banerjee'08}. Like Chalice, both also encode into Boogie2.  The connection between these logics and separation logic is less clear.  They explicitly talk about the footprint of an assertion, rather than implicitly. However, our new separation logic might facilitate future comparisons.

\section{Extensions and Applications}
\label{sec:extensions}
In this section we highlight the potential impact of our connection between separation logic and the implicit dynamic frames of Chalice, by explaining several ways in which ideas from one world can be transferred to the other.

\subsubsection*{Supporting Extra Connectives}
%

Our \ESL{} logic supports many more connectives than have previously existed in implicit dynamic frames logics. For example, the support for a ``magic wand'' in the logic (or indeed an unrestricted logical implication) is a novel contribution, which paves the way for investigating how to extend Chalice to support this much-richer assertion language. While a formal semantics for the magic wand does not immediately tell us how to implement inhaling and exhaling such assertions correctly, it provides us with a means of formally evaluating such a proposal. Furthermore, our direct semantics for the assertion logic of Chalice provides a means of judging whether a particular implementation is faithful to the intended logical semantics.

In addition, while the notions of minimal permission extensions and locally-havoced heap extensions are technically complex, it seems that the resulting semantics for the magic wand may actually simplify the problem of defining a suitable weakest-precondition semantics. This is because, whereas the classical semantics of the separation logic magic wand involves a quantification over states (which is problematic for encoding to a first-order prover), the semantics we present in this paper can, in the (common) case of supported assertions, eliminate the need for the quantifier altogether; we need only check the unique, minimal extension of the initial state to make the left-hand side true, if such an extension exists. Exploring the practical consequences of these observations will be interesting future work.

\subsubsection*{Evaluating the Chalice Implementation}
Various design decisions in the Chalice methodology can be evaluated using our formal semantics. For example, Chalice deals with potential interference from other threads by ``havocing'' heap locations whenever permission to the location is newly granted. An alternative design would be to ``havoc'' such locations whenever all permission to them was \emph{given up} in an exhale, instead. This would provide different weakest pre-conditions for Chalice commands, and it would be interesting to investigate what differences this design decision makes from a theoretical perspective. Our results provide the necessary basis for such investigations.

Separation logics typically feature recursive (abstract) predicates in their assertion language. The Chalice tool also includes an experimental implementation of recursive predicates (without arguments), along with the use of ``functions'' in specifications to describe properties of the state in a way which could support information hiding. In the course of investigating how to extend our results to handle predicates in the assertion logics, we discovered that the current approach to handling predicates/functions in Chalice is actually unsound in the presence of functions and the decision to havoc on inhales rather than exhales. We, along with other Chalice contributors, are now working on a redesign of Chalice predicates based on our findings. As above, the formal semantics and connections we have provided give us excellent tools for evaluating such a redesign.

\subsubsection*{Implementing Separation Logic}
One exciting outcome of the results we have presented is that a certain fragment of separation logic specifications can be directly represented in implicit dynamic frames and automatically verified using the Chalice tool. This is a consequence of three results:
\begin{enumerate}[(1)]
\item We have shown that our total heap semantics for separation logic coincides with its prior partial heaps semantics.
\item We have shown that we can replace all ``points-to'' predicates with logical primitives from implicit dynamic frames, preserving semantics.
\item We have shown that the Chalice weakest-pre-condition calculation agrees with the weakest pre-conditions used in separation logic verification.
\end{enumerate}
The critical aspect which is missing is the treatment of predicates - once we can extend our correspondence results to handle recursively-defined predicates in the logics (which are used in virtually all separation logic verification examples), then it will be possible to exploit our work to use Chalice to implement separation logic verification. This will open up many interesting practical areas of work, in comparing the performance and encodings of verification problems between Chalice and separation logic based tools.

\subsubsection*{Old Expressions}
We have also observed that the use of a total heap semantics seems to make it easy to support certain extra specification features in a separation logic assertion language. In particular, the use of ``old'' expressions in method contracts (allowing post-conditions to explicitly mention values of heap locations in the pre-state of the method call) is awkward to support in a partial heaps semantics, since it expresses relationships between partial heap fragments which may not have obviously-related domains. As a consequence, separation logic based tools typically do not support this feature, and typically use logical variables to connect old and new values of heap locations. However, with our total heap semantics it seems rather easy to evaluate old expressions by simply replacing our total heap with a copy of the pre-heap.  Consider the following two specifications, where the left one uses old expressions and the right one a logical variable $v$:
\[
\begin{array}{ll}
\begin{array}{l}
\mathsf{requires}\ \acc(x.f)\\
\mathsf{ensures}\ \acc(x.f) * x.f = old(x.f) + 1
\end{array}
&
\begin{array}{l}
\mathsf{requires}\ \acc(x.f) * x.f = v\\
\mathsf{ensures}\ \acc(x.f) * x.f = v + 1
\end{array}
\end{array}
\]
To use the logical variable specification, we must find a witness for the logical variable $v$, while with old expressions this witness is not required as it simply places a constraint on the possible old and new heaps. This is because the assertions describing the value relationship in the old expression specification only appears in the post-condition, which, from the caller's perspective, ends up as an assume. On the other hand, in the logical variable alternative, the variable appears both in the pre- and post-conditions, hence it also ends up used in an assert (when the caller exhales the pre-condition).
Moving to old expressions may have benefits for building tool support for separation logic.

\subsection*{Acknowledgements}

We thank Mike Dodds, David Naumann, Ioannis Kassios, Peter M\"uller, Sophia Drossopoulou and the anonymous ESOP and LMCS reviewers for feedback on this work.


\bibliographystyle{abbrv}


%
%

\bibliography{oo}
\newpage

\appendix

\section{Proofs}
\proofslist
\end{document}